\documentclass[12pt]{article}
\usepackage[margin=1in]{geometry}
\usepackage[english]{babel}
\usepackage[utf8]{inputenc}
\usepackage{graphicx}
\graphicspath{{figures/}}
\usepackage{amssymb,amsmath,amsthm}
\usepackage{stmaryrd,xspace}
\usepackage{hyperref}
\usepackage{enumitem}
\setitemize{noitemsep,topsep=2pt,parsep=0pt,partopsep=0pt}
\setenumerate{noitemsep,topsep=2pt,parsep=0pt,partopsep=0pt}
\usepackage{xcolor}
\usepackage{multicol,multirow}
\usepackage{changepage} % for adjustwidth used in subproofs
\usepackage{tikz}
\usetikzlibrary{decorations.pathreplacing}
\usepackage{booktabs}
\newtheorem{theorem}{Theorem}

\newtheorem{lemma}[theorem]{Lemma}
\newtheorem{corollary}[theorem]{Corollary}

\newtheorem{remark}[theorem]{Remark}
\newtheorem{example}[theorem]{Example}

\newtheorem{notation}[theorem]{Notation}

\newcommand{\N}{\mathbb{N}}

\newcommand{\1}{\ensuremath{\normalfont \texttt{1}}\xspace}
\newcommand{\0}{\ensuremath{\normalfont \texttt{0}}\xspace}
\newcommand{\bool}{\{\0,\1\}}

\newcommand{\intz}[1]{\ensuremath{\llbracket #1 \rrbracket}\xspace}
\newcommand{\set}[1]{\left\{#1\right\}}
\newcommand{\zu}{\ensuremath{\set{\0,\1}}\xspace}
\newcommand{\ignore}[1]{}
\newcommand{\Pn}{\ensuremath{\mathcal{P}_n}\xspace}
\newcommand{\Pf}{\ensuremath{\mathcal{P}_{|f|}}\xspace}
\newcommand{\Px}[1]{\ensuremath{\mathcal{P}_{#1}}\xspace}

\newcommand{\negi}[2]{\overline{#1}^{#2}}
\newcommand{\loopless}[1]{\hat{#1}}
\newcommand{\lab}{lab}
\newcommand{\labminus}{\ominus}
\newcommand{\labplus}{\oplus}
\newcommand{\dynamics}[1]{\mathcal D(#1)}
\newcommand{\dynamicsu}[2]{\mathcal D_{#2}(#1)}
\newcommand{\updates}[1]{\mathcal U(#1)}
\newcommand{\sens}[1]{\mu_s(#1)}
\newcommand{\GECA}[1]{G^\text{ECA}_{#1}}
\newcommand{\updatesECA}[1]{\mathcal U^\text{ECA}(#1)}
\newcommand{\dleft}[1]{\overleftarrow{\tt d\!\!}_{#1}}
\newcommand{\dright}[1]{\overrightarrow{\tt d\!\!}_{#1}}
\newcommand{\dset}[1]{{\tt d}_{#1}}
\newcommand{\countplus}[1]{|#1|_{\labplus}}
\newcommand{\countminus}[1]{|#1|_{\labminus}}
\newcommand{\countplusminus}[1]{|#1|_{\labplus\labminus}}
\newcommand{\partplusminus}[2]{L_{#1#2}}
\newcommand{\sumplusminus}[2]{S_{#1#2}}
\newcommand{\ie}{\emph{i.e.}\@\xspace}
\newcommand{\wrt}{\emph{w.r.t}\@\xspace}

\newcommand{\TODO}[1]{\textcolor{red}{{\fbox{TODO} #1}}}

\usepackage{authblk}

\title{Non-maximal sensitivity to synchronism in periodic elementary cellular automata: exact asymptotic measures}
\author[1]{Pedro Paulo Balbi}
\author[2]{Enrico Formenti}
\author[3]{K\'evin Perrot}
\author[2]{Sara Riva}
\author[1]{Eurico L. P. Ruivo}
\affil[1]{Universidade Presbiteriana Mackenzie, FCI, S\~{a}o Paulo, Brazil} 
\affil[2]{Universit\'e C\^ote d'Azur, CNRS, I3S, France}
\affil[3]{Universit\'e publique}
\date{}

\begin{document}
\maketitle

%%%%%%%%%%%%%%%%%%%%%%%%%%%%%%%%
% abstract
%%%%%%%%%%%%%%%%%%%%%%%%%%%%%%%%
\begin{abstract}
  In~\cite{pmmor19} and~\cite{rmmop18} the authors showed that elementary
  cellular automata rules 0, 3, 8, 12, 15, 28, 32, 34, 44,
  51, 60, 128, 136, 140, 160, 162, 170, 200 and 204 (and
  their conjugation, reflection, reflected-conjugation) are not maximum
  sensitive to synchronism, \ie they do not have a different dynamics
  for each (non-equivalent) block-sequential update schedule (defined as
  ordered partitions of cell positions). In this work we present exact
  measurements of the sensitivity to synchronism for these rules,
  as functions of the size. These exhibit a surprising variety of
  values and associated proof methods, such as the special pairs of
  rule $128$, and the connection to the bissection of Lucas numbers of rule $8$.
\end{abstract}

%%%%%%%%%%%%%%%%%%%%%%%%%%%%%%%%
% introduction
%%%%%%%%%%%%%%%%%%%%%%%%%%%%%%%%
\section{Introduction}

Cellular automata (CAs) are discrete dynamical systems with respect
to time, space and state variables, which have been widely studied both as
mathematical and computational objects as well as suitable models for
real-world complex systems.

The dynamics of a CA is locally-defined: every agent (\emph{cell}) computes
its future state based upon its present state and those of their
neighbors, that is, the cells connected to it.
In spite of their apparent simplicity, they may display non-trivial
global emergent behavior, some of them even reaching computational
universality \cite{Cook2004,Gardner1970}.

Originally, CAs are updated in a synchronous fashion, that is, every cell 
of the lattice is updated simultaneously. However, over the last decade,
\emph{asynchronous} cellular automata have attracted increasing
attention in its associated scientific community.

A comprehensive and detailed overview of asynchronous CAs is given in
\cite{Fates2014}. There are different ways to define asynchronism in CAs,
be it deterministically or stochastically.

Here, we deal with a
deterministic version of asynchronism, known as \emph{block-sequential},
coming from the model of Boolean networks and first characterized
for this more general model in \cite{Aracena2009,Aracena2011}.
Under such an update scheme, the 
lattice of the CA is partitioned into blocks of cells, each one is 
assigned a priority of being updated, and this priority ordering is kept
 fixed throughout the time evolution. For the sake of simplicity, 
from now on, whenever we refer to \emph{asynchronism}, we will mean 
\emph{block-sequential}, deterministic asynchronism.

In previous works (\cite{pmmor19,rmmop18}), the notion of \emph{maximum
sensitivity to asynchronism} was established. Basically, a CA rule was said
to present maximum sensitivity to asynchronism when, for any two different
block-sequential
update schedules, the rule would yield different dynamics.
%Out of the 88 dynamically independent elementary cellular automata (ECAs) rules,
%59 possess maximum sensitivity to asynchronism, while the remaining 19 
Out of the 256 elementary cellular automata rules (ECAs),
200 possess maximum sensitivity to asynchronism, while the remaining 56
rules do not. Therefore, it is natural to try and define a \emph{degree} of
sensitivity to asynchronism to the latter.

Here, such a notion of a measure to the sensitivity to asynchronism is presented
and general analytical formulas for sensitivities of the non-maximal sensitive
rules are provided.
The results (to be presented on Table~\ref{tab:results}
at the end of Section~\ref{s:definitions})
exhibit an interesting range of values requiring the introduction of various techniques,
from measures tending to $0$ (insensitive rules)
to measures tending to $1$ (almost max-sensitive),
with one rule tending to some surprising constant between $0$ and $1$.

This paper is organized as follows. In Section~\ref{s:definitions}, fundamental
definitions and results on Boolean networks, update digraphs and elementary
cellular automata are given. Then, in Section~\ref{s:exp}, experimental
measures of sensitivity to asynchronism are given for rules which do not possess
maximal sensitivity to asynchronism. Such experimental measures pave the way to
the theoretical results in Section~\ref{s:theory}, in which formal
expressions to the sensitivity to asynchronism of such rules are provided for
configurations of arbitrary size. Finally, concluding remarks are made in
Section~\ref{s:conclusions}.

%%%%%%%%%%%%%%%%%%%%%%%%%%%%%%%%
% definitions
%%%%%%%%%%%%%%%%%%%%%%%%%%%%%%%%
\section{Definitions}
\label{s:definitions}

Elementary cellular automata will be presented in the more general framework of
Boolean automata networks, for which the variation of update schedule benefits
from useful considerations already studied in the literature.
Figure~\ref{fig:ECA-digraph} illustrates the definitions.

\begin{figure}[!h]
  \centerline{
    \includegraphics{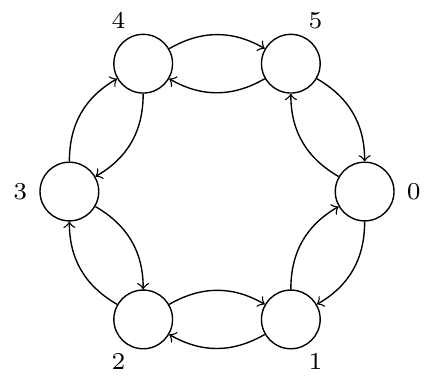}
    \hspace{2cm}
    \includegraphics{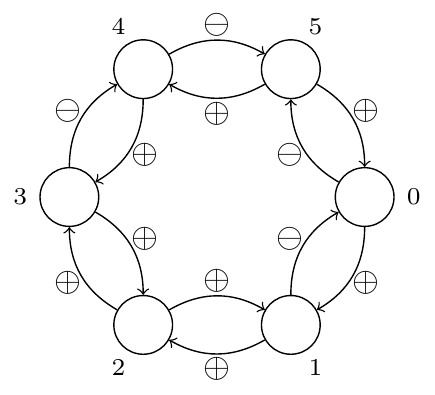}
  }
  \caption{
    Left: interaction digraph $\GECA{6}$ of the ECA rule $128$
    for $n=6$, with local functions $f_i(x) = x_{i-1} \wedge
    x_i \wedge x_{i+1}$ for all $i \in \{0,\dots,5\}$.
    Right: update digraph corresponding to the update schedules
    $\Delta=(\{1,2,3\},\{0,4\},\{5\})$ and $\Delta'=(\{1,2,3\},\{0\},\{4\},\{5\})$,
    which are therefore equivalent ($\Delta \equiv \Delta'$).
    For example, $f^{(\Delta)}(\1\1\1\0\1\1)=\1\1\0\0\0\0$ whereas for the
    synchronous update schedule we have
    $f^{(\Delta^\texttt{sync})}(\1\1\1\0\1\1)=\1\1\0\0\0\1$.
  }
  \label{fig:ECA-digraph}
\end{figure}

%%%%%%%%%%%%%%%%%%%%%%%%%%%%%%%%
\subsection{Boolean networks}
\label{ss:bn}

A Boolean Network (BN) of size $n$ is an arrangement of $n$ finite Boolean automata
(or components)
interacting each other according to a \emph{global rule} $f\colon\zu^n\to\zu^n$
which describes how the global state changes after one time step. Let
$\intz{n}=\set{0,\dots,n-1}$. Each automaton is identified with a unique integer
$i\in\intz{n}$ and $x_i$ denotes the current state of the automaton $i$.
A \emph{configuration} $x\in\zu^n$ is a snapshot of the current state of all automata
and represents the global state of the BN.

For convenience, we identify configurations with words on $\zu^n$.
%\ignore{
%Given a set of $n$ {\em components} $\intz{n}=\{0,\dots,n-1\}$, a {\em
%configuration} $x$ is an element of $\bool^n$ which assigns a Boolean {\em
%state} to each component $i$, denoted $x_i$ for simplicity. 
%}
%\ignore{For convenience, we denote}
Hence, for example, $\0\1\1\1\1$ or $\0\1^4$ both denote
the configuration $(\0,\1,\1,\1,\1)$. Remark that the global function
\ignore{A {\em Boolean network} (BN)}
$f:\bool^n \to \bool^n$ of a BN of size $n$ induces
a set of $n$ \emph{local functions} $f_i:\bool^n \to \bool$, one per each component, 
such that $f(x)=(f_0(x),f_1(x),\dots,f_{n-1}(x))$ for all $x \in \bool^n$.  This gives a
static description of a discrete dynamical system, and it remains to set the
order in which components are updated in order to get a dynamics.  Before going
to update schedules, let us first introduce interaction digraphs.

The component $i$ \emph{influences}
the component $j$ if $\exists x \in \bool^n : f_j(x) \neq f_j(\negi{x}{i})$, where
$\negi{x}{i}$ is the configuration obtained from $x$ by flipping the state of 
component $i$. Note that in literature one may also consider {\em positive} and
{\em negative} influences, but they will not be useful for the present study.
The \emph{interaction digraph} $G_f=(V,A)$ of a BN $f$ represents the effective
dependencies among its set of components %. 
% Now, we have
\[
  V=\intz{n} \quad\text{and}\quad A=\set{ (i,j) \mid i \text{ influences } j}.
\]
It will turn out to be pertinent to consider $\loopless{G_f}=(V,A)$, obtained
from $G_f$ by removing the loops (arcs of the form $(i,i)$).

For $n\in\N$, denote \Pn the set of ordered partitions of \intz{n} and $|f|$
the size of a BN $f$.
A \emph{block-sequential update schedule} $\Delta=(\Delta_1,\dots,\Delta_k)$ is
an element of \Pf. It defines the following dynamics
$f^{(\Delta)}:\bool^n \to \bool^n$,
\[
  f^{(\Delta)}=f^{(\Delta_k)} \circ \dots \circ f^{(\Delta_2)} \circ f^{(\Delta_1)}
  \quad\text{with}\quad f^{(\Delta_j)}(x)_i=\begin{cases}
    f_i(x) & \text{if } i \in \Delta_j,\\
    x_i & \text{if } i \notin \Delta_j.
  \end{cases}
\]
In words, the components are updated in the order given by $\Delta$:
sequentially part after part, and in parallel within each part. The \emph{parallel}
or \emph{synchronous} update schedule is $\Delta^\texttt{sync}=(\intz{n})$ 
and we have $f^{(\Delta^\texttt{sync})}=f$.
In this article, since only block-sequential update schedules are considered, 
they are simply called \emph{update schedule} for short. They are
\begin{itemize}
  \item ``\emph{fair}'' in the sense that all components are updated the exact
    same number of times,
  \item ``\emph{periodic}'' in the sense that the same ordered partition is
    repeated.
\end{itemize}

Given a BN $f$ of size $n$ and an update schedule $\Delta$, the 
\emph{transition digraph} $D_{f^{(\Delta)}}=(V,A)$ is such that
\[
  V=\bool^n \quad\text{and}\quad A=\{ (x,f^{(\Delta)}(x)) \mid x \in \bool^n \}.
\]
It describes the \emph{dynamics} of $f$ under the update schedule $\Delta$. The set
of all possible dynamics of the BN $f$, at the basis of the 
measure of sensitivity to synchronism, is then defined as
\[
  \dynamics{f} = \set{ D_{f^{(\Delta)}} \mid\Delta\in\Pf}.
\]
%\ignore{
%$$
%  \dynamics{f} = \{ D_{f^{(\Delta)}} \mid
%  \Delta \text{ an update schedule} \}.
%$$
%}
%%%%%%%%%%%%%%%%%%%%%%%%%%%%%%%%
\subsection{Update digraphs and equivalent update schedules}
\label{ss:update-digraph}

For a given BN, some update schedules always give the same dynamics. Indeed,
if, for example, two components do not influence each other, their
order of updating has no effect on the dynamics (see Example~\ref{fig:ECA-digraph} for
a detailed example). In~\cite{Aracena2009}, the notion of \emph{update digraph}
has been introduced in order to study update schedules.

%which consists in labelling
%the arcs $(i,j)$ of the loopless interaction digraph
%To express equivalent classes of update schedules, we present the notion of {\em
%update digraph} introduced in~\cite{Aracena2009}, which consists in labelling
%the arcs $(i,j)$ of the loopless interaction digraph according to whether
%component $i$ is updated prior to component $j$ or not. 
Given a BN $f$ with loopless interaction digraph $\loopless{G_f}=(V,A)$ and an update schedule
$\Delta\in\Pn$, define $\lab_\Delta : A \to \{\labplus,\labminus\}$ as
\[
  \forall (i,j) \in A : \lab_\Delta((i,j))=\begin{cases}
    \labplus & \text{if } i \in \Delta_{a},j \in \Delta_{b}
    \text{ with } 1 \leq b \leq a \leq k,\\
    \labminus & \text{if } i \in \Delta_{a},j \in \Delta_{b}
    \text{ with } 1 \leq a < b \leq k.
  \end{cases}
\]
The \emph{update digraph} $U_{f^{(\Delta)}}$ of the BN $f$ for the update 
schedule $\Delta\in\Pn$ is the loopless interaction digraph decorated with
$\lab_\Delta$, \ie $U_{f^{(\Delta)}}=(V,A,\lab_\Delta)$. Note that loops are removed because
they bring no meaningful information: indeed, an edge $(i,i)$ would always be labeled
$\labplus$. Now we have that, if two update schedules define the same update
digraph then they also define the same dynamics.

\begin{theorem}[\cite{Aracena2009}]\label{theorem:lab}
  Given a BN $f$ and two update schedules $\Delta,\Delta'$,
  if $\lab_\Delta=\lab_{\Delta'}$ then $D_{f^{(\Delta)}}=D_{f^{(\Delta')}}$.
\end{theorem}

A very important remark is that not all labelings correspond to \emph{valid}
update digraphs (\ie such that there are update schedules giving these
labelings). For example, if two arcs $(i,j)$ and $(j,i)$ belong to the
interaction digraph and are both labeled $\labminus$, it would mean that $i$ is
updated prior to $j$ and $j$ is updated prior to $i$, which is contradictory.
Fortunately there is a nice characterisation of \emph{valid} update digraphs.

\begin{theorem}[\cite{Aracena2011}]\label{theorem:lab_valid}
  Given $f$ with $\loopless{G_f}=(V,A)$,
  the label function $\lab: A \to \set{ \labplus,\labminus}$ is valid
  if and only if there is no cycle $(i_0,i_1,\dots,i_k)$,
  with $i_0=i_k$ and $k>0$, such that
  \begin{itemize}
    \item $\forall 0 \leq j < k:
      ((i_j,i_{j+1}) \in A \wedge \lab((i_j,i_{j+1}))=\labplus) \vee
      ((i_{j+1},i_j) \in A \wedge \lab((i_{j+1},i_j))=\labminus)$,
    \item $\exists 0 \leq i < k: \lab((i_{j+1},i_j))=\labminus$.
  \end{itemize}
\end{theorem}

In words, Theorem~\ref{theorem:lab_valid} states that a labeling is valid if
and only if the multi-digraph where the labeling is unchanged but the
orientation of arcs labeled $\labminus$ is reversed, does not contain a cycle
with at least one arc label $\labminus$ ({\em forbidden cycle}).

According to Theorem~\ref{theorem:lab}, update digraphs define equivalence
classes of update schedules: $\Delta \equiv \Delta'$ if and only if
$\lab_\Delta=\lab_{\Delta'}$. Given a BN $f$, 
the set of equivalence classes of update schedules is therefore defined as
\[
  \updates{f}=\set{ U_{f^{(\Delta)}} \mid \Delta\in\Pf}.
\]

%%%%%%%%%%%%%%%%%%%%%%%%%%%%%%%%
\subsection{Sensitivity to synchronism}
\label{ss:sensitivity}

The sensitivity to synchronism $\sens{f}$ of a BN $f$ quantifies the proportion of 
distinct dynamics  \wrt non-equivalent update schedules. The idea is that
when two or more update schedules are equivalent then $\sens{f}$ decreases,
while it increase when distinct update schedules bring to different dynamics. 
More formally, given a BN $f$ we define
\[
  \sens{f}=\frac{|\dynamics{f}|}{|\updates{f}|}.
\]

Obviously, it holds that $\frac{1}{|\updates{f}|} \leq \sens{f} \leq 1$, and a BN
$f$ is as much sensible to synchronism as it has different dynamics when the
update schedule varies. The extreme cases are a BN $f$ with
$\sens{f}=\frac{1}{|\updates{f}|}$ that has always the same dynamics $D_{f^{(\Delta)}}$
for any update schedule $\Delta$, and a BN $f$ with $\sens{f}=1$ which has a
different dynamics for different update schedules (for each $\Delta \not
\equiv \Delta'$ it holds $D_{f^{(\Delta)}} \neq D_{f^{(\Delta')}}$).
A BN $f$ is \emph{max-sensitive} to synchronism iff $\sens{f}=1$.
Note that a BN $f$ is max-sensitive if and only if
\begin{align}
  \label{eq:max-sensitive}
  \forall\Delta\in\Pf\forall\Delta'\in\Pf\;(\Delta\not\equiv \Delta') \Rightarrow\exists x \in \bool^n
  \exists i\in\intz{n}\,f^{(\Delta)}(x)_i \neq f^{(\Delta')}(x)_i\enspace.
\end{align}

%%%%%%%%%%%%%%%%%%%%%%%%%%%%%%%%
\subsection{Elementary cellular automata}
\label{ss:eca}

In this study we investigate the sensitivity to synchronism of \emph{elementary cellular automata} (ECA)
over periodic configurations. Indeed, they are a subclass of BN
in which all components (also called \emph{cells} in this context) have 
the same local rule, as follows. Given a size $n$, the ECA of local function
$h: \bool^3 \to \bool$ is the BN $f$ such that
\[
  \forall i \in \intz{n}: f_i(x)=h(x_{i-1},x_i,x_{i+1})
\]
where components are taken modulo $n$ (this will be the case 
throughout all the paper without explicit mention).
We use \emph{Wolfram numbers}~\cite{Wolfram2002} to designate each of the
$256$ ECA local rule $h: \bool^3 \to \bool$ as the number
\[
  w(h)=\sum_{(x_1,x_2,x_3)\in\bool^3} h(x_1,x_2,x_3)2^{2^2x_1+2^1x_2+2^0x_3}.
\]
Given a Boolean function $h: \bool^3 \to \bool$, consider the following transformations over local rules:
$\tau_i(h)(x,y,z)=h(x,y,z)$, $\tau_r(h)(x,y,z)=h(z,y,x)$, $\tau_n(h)(x,y,z)=1-h(1-z,1-y,1-x)$
and $\tau_{rn}(h)(x,y,z)=1-h(1-z,1-y,1-x)$ for all $x,y,z\in\bool$.
%In~\cite{CattaneoFMM97},
%it is proved the previous transformations preserve topological dynamics and hence,
In our
context, they preserve the sensitivity to synchronism.
For this reason we consider only 88 ECA rules up to $\tau_i$, $\tau_r$, $\tau_n$ and $\tau_{rn}$. 
%For ECAs
%topological conjugacy preserves the sensitivity to synchronism (see Lemma~\ref{lemma:topo} below),
%for this reason we consider only 83 ECA up to topological conjugacy~\cite{CattaneoFMM97,Epperlein2015}.
Table~\ref{tab:88-topo} reports these equivalence classes of ECA,
the smallest Wolfram number per class is indicated.

\begin{table}%[htb]
\centerline{
\begin{tabular}{|p{13cm}|}
\hline
0, 1, 2, 3, 4, 5, 6, 7, 8, 9, 10, 11, 12, 13, 14, 15, 18, 19, 22, 23, 24, 25,
26, 27, 28, 29, 30, 32, 33, 34, 35, 36, 37, 38, 40, 41, 42, 43, 44, 45, 46, 50,
51, 54, 56, 57, 58, 60, 62, 72, 73, 74, 76, 77, 78, 90, 94, 104, 105, 106, 108,
110, 122, 126, 128, 130, 132, 134, 136, 138, 140, 142, 146, 150, 152, 154, 156,
160, 162, 164, 168, 170, 172, 178, 184, 200, 204, 232\\
\hline
\end{tabular}
}
\caption{ECA local rules up to $\tau_i$, $\tau_r$, $\tau_n$ and $\tau_{rn}$.}
\label{tab:88-topo}
\end{table}

%The dynamic of ECAs is classically studied under the synchronous update
%schedule, and we are interested in measuring how sensible they are to
%variations of update schedule. 
The definitions of Subsection~\ref{ss:sensitivity} %are given for a fixed size, and 
are applied to
ECA rules as follows. Given a size $n$, the {\em ECA interaction digraph of
size $n$} $\GECA{n}=(V,A)$ %, is given by von-Neumann neighborhood of
%radius one: 
is such that
$V=\intz{n}$ and $A=\set{ (i+1,i),(i,i+1) \mid i \in \intz{n}}$.

In~\cite{pmmor19,rmmop18}, it is proved that
%Denoting  , it is
%proven  using Theorem~\ref{theorem:lab_valid} that
\[
  |\updatesECA{n}|=3^n-2^{n+1}+2.
\]
where $\updatesECA{n}$ is the set of valid labelings of $\GECA{n}$.
The sensitivity to synchronism of ECAs is measured relatively to the family of
ECAs, and therefore relatively to this count of valid labelings of $\GECA{n}$,
even for rules where some arcs do not correspond to effective influences (one may think of rule $0$). Except from this subtlety, the measure is 
correctly defined by considering, for an ECA rule number $\alpha$ and a 
size $n$, that $h_\alpha\colon\bool^3 \to \bool$ is its local rule, and 
that $f_{\alpha,n}\colon\bool^n \to \bool^n$
is its global function on periodic configurations of size $n$,
\[
  \forall x \in \bool^n\;f_{\alpha,n}(x)_i= h_\alpha(x_{i-1},x_i,x_{i+1}).
\]
Then, the sensitivity to synchronism of ECA rule number $\alpha$ is given by
\[
  \sens{f_{\alpha,n}}=\frac{|\dynamics{f_{\alpha,n}}|}{3^n-2^{n+1}+2}.
\]
An ECA rule number $\alpha$ is ultimately
\emph{max-sensitive to synchronism} when
%\[
%  \exists n_0 \in \N : \forall n \geq n_0 : \sens{f_{\alpha,n}}=1.
%\]
\[
  \lim_{n\to+\infty} \sens{f_{\alpha,n}}=1.
\]

The following result provides a first overview of sensitivity to
synchronism in ECA.

\begin{theorem}[\cite{pmmor19,rmmop18}]
  \label{theorem:max}
  For any size $n \geq 7$, the nineteen
  ECA rules 0, 3, 8, 12, 15, 28, 32, 34, 44,
  51, 60, 128, 136, 140, 160, 162, 170, 200 and 204
  are not max-sensitive to synchronism.
  The remaining sixty nine other rules are max-sensitive to synchronism.
\end{theorem}

Theorem~\ref{theorem:max} gives a precise measure of sensitivity for the sixty 
nine maximum sensitive rules, for which $\sens{f_{\alpha,n}}=1$ for all $n \geq
7$, but for the nineteen that are not maximum sensitive it only informs that
$\sens{f_{\alpha,n}}<1$ for all $n \geq 7$. In the rest of this paper we study the
precise dependency on $n$ of $\sens{f_{\alpha,n}}$ for these rules, filling the
huge gap between $\frac{1}{3^n-2^{n+1}+2}$ and
$\frac{3^n-2^{n+1}+1}{3^n-2^{n+1}+2}$. This will offer a finer view on the
sensitivity to synchronism of ECA.
The results are summarized in Table~\ref{tab:results}.

\begin{table}
  \centerline{
    \renewcommand{\arraystretch}{1.25}
    \begin{tabular}{|c|c|c|c|}
      \hline
      Class & Rules ($\alpha$) & Sections & Sensitivity ($\sens{f_{\alpha,n}}$)\\
      \hline
      I & $0, 51, 200, 204$ & \ref{ss:51_204} & $\frac{1}{3^n-2^{n+1}+2}$ for any $n \geq 3$\\
      \multirow{2}{*}{II} & $3,12,15,34,60,136,170$ & \multirow{2}{*}{\ref{cl:low}} & \multirow{2}{*}{$\frac{2^n-1}{3^n-2^{n+1}+2}$ for any $n \geq 4$}\\
      & $28,32,44,140$ & & \\    
      III & $8$ & \ref{s:8} & $\frac{\phi^{2n}+\phi^{-2n}-2^n}{3^n-2^{n+1}+2}$ for any $n \geq 5$\\
      IV & $128,160,162$ & \ref{s:128_160_162} &$\frac{3^n-2^{n+1}-cn+2}{3^n-2^{n+1}+2}$ for any $n \geq 5$\\
      \hline
    \end{tabular}}
  \caption{The rules are divided into four classes ($\phi$ is the golden ratio).}
  \label{tab:results}
\end{table}

%%%%%%%%%%%%%%%%%%%%%%%%%%%%%%%%
% experiments
%%%%%%%%%%%%%%%%%%%%%%%%%%%%%%%%
\section{Experimental measures of sensitivity to synchronism}
\label{s:exp}

This section presents some numerical calculations of $\sens{f_{\alpha,n}}$
for rules that are not max-sensitive to synchronism according to
Theorem~\ref{theorem:max}. An interesting variety of behaviors
(for $n=3$ to $10$) is observed. It will be characterized in Section~\ref{s:theory}.

\begin{center}
  \includegraphics[scale=.8,page=1]{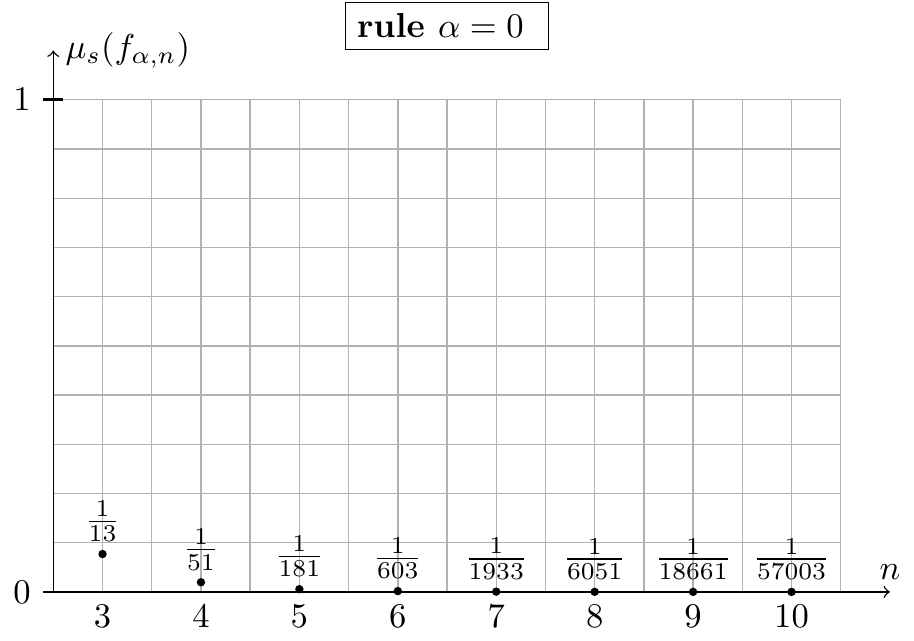}
  \includegraphics[scale=.8,page=2]{ECA-sensitivity-plot.pdf}\\
  \includegraphics[scale=.8,page=3]{ECA-sensitivity-plot.pdf}
  \includegraphics[scale=.8,page=4]{ECA-sensitivity-plot.pdf}\\
  \includegraphics[scale=.8,page=5]{ECA-sensitivity-plot.pdf}
  \includegraphics[scale=.8,page=6]{ECA-sensitivity-plot.pdf}\\
  \includegraphics[scale=.8,page=7]{ECA-sensitivity-plot.pdf}
  \includegraphics[scale=.8,page=8]{ECA-sensitivity-plot.pdf}\\
  \includegraphics[scale=.8,page=9]{ECA-sensitivity-plot.pdf}
  \includegraphics[scale=.8,page=10]{ECA-sensitivity-plot.pdf}\\
  \includegraphics[scale=.8,page=11]{ECA-sensitivity-plot.pdf}
  \includegraphics[scale=.8,page=12]{ECA-sensitivity-plot.pdf}\\
  \includegraphics[scale=.8,page=13]{ECA-sensitivity-plot.pdf}
  \includegraphics[scale=.8,page=14]{ECA-sensitivity-plot.pdf}\\
  \includegraphics[scale=.8,page=15]{ECA-sensitivity-plot.pdf}
  \includegraphics[scale=.8,page=16]{ECA-sensitivity-plot.pdf}\\
  \includegraphics[scale=.8,page=17]{ECA-sensitivity-plot.pdf}
  \includegraphics[scale=.8,page=18]{ECA-sensitivity-plot.pdf}
  \includegraphics[scale=.8,page=19]{ECA-sensitivity-plot.pdf}
\end{center}

%%%%%%%%%%%%%%%%%%%%%%%%%%%%%%%%
% main study
%%%%%%%%%%%%%%%%%%%%%%%%%%%%%%%%
\section{Theoretical measures of sensitivity to synchronism}
\label{s:theory}

This section contains the main results of the paper, regarding the
dependency on $n$ of $\sens{f_{\alpha,n}}$ for ECA rules that are 
not max-sensitive to synchronism.

As illustrated in Table~\ref{tab:results}, the
ECA rules can be divided into four classes according to their sensitivity functions.
Each class will require specific proof techniques but all of them
have interaction digraphs as a common denominator.

As a starting point, one can consider the case of
%%%%%%%%%%%%%%%%%%%%%%%%%%%%%%%%
% 0, 3, 12, 15, 34, 51, 60, 136, 170, 204
%\subsection{ECA rules with incomplete interaction digraphs}
%
%Some 
ECA rules have an interaction digraph which is a proper
subgraph of $\GECA{n}$. Indeed, when considering them as BN
many distinct update schedules give the same labelings and hence,
by Theorem~\ref{theorem:lab} and the definition of $\sens{f_{\alpha,n}}$,
they cannot be max-sensitive.
%have proper subgraphs of the ECA
%from Figure~\ref{fig:ECA-digraph} (their interaction digraph has a strict
%subset of these edges), and therefore cannot be max-sensitive to 
%synchronism from the definition of $\sens{f_{\alpha,n}}$ for ECA and
%Theorem~\ref{theorem:lab}: for them some update schedules of $\updatesECA{n}$
%correspond to the same labeling of their interaction digraph $G_{f_{\alpha,n}}$
%when we se them as a BN. 
This is the case of the following set of ECA rules
$\mathcal{S}=\set{0, 3, 12, 15, 34, 51, 60, 136, 170, 204}$.
Indeed, denoting $G_{f_{\alpha,n}}=(\intz{n},A_{f_{\alpha,n}})$ the
interaction digraph of ECA rule $\alpha$ of size $n$ for $\alpha\in\mathcal{S}$, one finds $\forall n \geq3$ and 
$\forall i\in \intz{n}$:
\begin{multicols}{2}
  \begin{itemize}
    \item $(i+1,i) \notin A_{f_{0,n}}$ and $(i-1,i) \notin A_{f_{0,n}}$,
    \item $(i+1,i) \notin A_{f_{3,n}}$,
    \item $(i+1,i) \notin A_{f_{12,n}}$,
    \item $(i+1,i) \notin A_{f_{15,n}}$,
    \item $(i,i+1) \notin A_{f_{34,n}}$,
    \item $(i+1,i) \notin A_{f_{51,n}}$ and $(i-1,i) \notin A_{f_{51,n}}$,
    \item $(i+1,i) \notin A_{f_{60,n}}$,
    \item $(i,i+1) \notin A_{f_{136,n}}$,
    \item $(i,i+1) \notin A_{f_{170,n}}$,
    \item $(i+1,i) \notin A_{f_{204,n}}$ and $(i-1,i) \notin A_{f_{204,n}}$.
  \end{itemize}
\end{multicols}

%%%%%%%%%%%%%%%%%%%%%%%%%%%%%%%%
% toolbox
%\subsection{An handful of useful results}
%\subsection{Useful results and notations}
%\label{s:toolbox}
Let us now introduce some useful results and notations that will be widely
used in the sequel.
Given an update schedule $\Delta$, in order to study the chain of influences
involved in the computation of the image at cell $i \in \intz{n}$, define 
\begin{align*}
  \dleft{\Delta}(i)=&\max\set{ k \in \N \mid
  \forall j \in \N : 0 < j < k \implies \lab_\Delta((i-j,i-j+1))=\labminus}\\
  \dright{\Delta}(i)=&\max\set{ k \in \N \mid
  \forall j \in \N : 0 < j < k \implies \lab_\Delta((i+j,i+j-1))=\labminus}.
\end{align*}
These quantities are well defined because $k=1$ is always a possible value,
and moreover, if
$\dleft{\Delta}(i)$ or $\dright{\Delta}(i)$ is greater than $n$, then there is a
forbidden cycle in the update digraph of schedule $\Delta$
(Theorem~\ref{theorem:lab_valid}). Note that for any $\Delta\in\Pn$
\[
  \lab_\Delta((i-\dleft{\Delta}(i),i-\dleft{\Delta}(i)+1))=\labplus
  \quad\text{and}\quad
  \lab_\Delta((i+\dright{\Delta}(i),i+\dright{\Delta}(i)-1))=\labplus.
\]
See Figure~\ref{fig:d} for an illustration.

\begin{figure}
  \centerline{\includegraphics{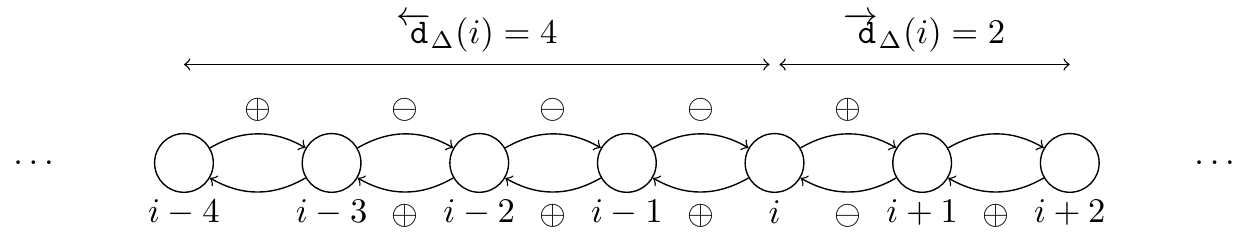}}
  \caption{illustration of the chain of influences for some update schedule $\Delta$.}
  \label{fig:d}
\end{figure}

The purpose of these quantities is that it holds for any $x \in \bool^n$, 
\begin{align}
  \label{eq:d}
  \begin{tikzpicture}[baseline=-1.5cm]
    % caution: full of hardcoded coordinates
    \node[inner sep=0pt] (i) at (0,0) {$(\qquad,x_i,\qquad)$};
    \node[left,inner sep=0pt] at (i.west) {$f_\alpha^{(\Delta)}(x)_i=r_\alpha$};
    \draw[decorate,decoration={brace,aspect=0.9,amplitude=4pt}] (-3.5,-.6) -- ++ (3,0);
    \draw[decorate,decoration={brace,aspect=0.1,amplitude=4pt}] (.5,-.6) -- ++ (3,0);
    \node at (-2,-1) {$r_\alpha(\qquad,x_{i-1},x_i)$};
    \node at (2,-1) {$r_\alpha(x_i,x_{i+1},\qquad)$};
    \draw[decorate,decoration={brace,aspect=0.8,amplitude=4pt}] (-4.3,-1.6) -- ++ (2,0);
    \draw[decorate,decoration={brace,aspect=0.2,amplitude=4pt}] (2.5,-1.6) -- ++ (2,0);
    \node at (-3.5,-2) {$\dots$};
    \node at (3.5,-2) {$\dots$};
    \draw[decorate,decoration={brace,aspect=0.5,amplitude=4pt}] (-7.5,-2.6) -- ++ (7,0);
    \draw[decorate,decoration={brace,aspect=0.5,amplitude=4pt}] (.5,-2.6) -- ++ (7,0);
    \node at (-4,-3) {$r_\alpha(x_{i-\!\!\dleft{\Delta}(i)},x_{i-\!\!\dleft{\Delta}(i)+1},x_{i-\!\!\dleft{\Delta}(i)+2})$};
    \node at (4,-3) {$r_\alpha(x_{i+\!\!\dright{\Delta}(i)-2},x_{i+\!\!\dright{\Delta}(i)-1},x_{i+\!\!\dright{\Delta}(i)})$};
  \end{tikzpicture}
\end{align}
\ie the quantities $\dleft{\Delta}(i)$ and $\dright{\Delta}(i)$ are the
lengths of the chain of influences at cell $i$ for the update schedule $\Delta$,
on both sides of the interaction digraph.
If the chains of influences at some cell $i$ are identical for two update schedules,
then the images at $i$ we be identical for any configuration, as stated in the following lemma.
\begin{lemma}
  \label{lemma:di}
  For any ECA rule $\alpha$, any $n \in \N$, any $\Delta,\Delta'\in\Pn$ and any $i\in\intz{n}$,
  it holds that
  \[
    \dleft{\Delta}(i) = \dleft{\Delta'}(i) \wedge \dright{\Delta}(i) = \dright{\Delta'}(i)
    \text{ implies }
    \forall x \in \{\0,\1\}^n\;f_{\alpha,n}^{(\Delta)}(x)_i=f_{\alpha,n}^{(\Delta')}(x)_i.
  \]
\end{lemma}

\begin{proof}
  This is a direct consequence of Equation~\ref{eq:d}, because the nesting of local rules
  for $\Delta$ and $\Delta'$ are identical at cell $i$.
\end{proof}

For any rule $\alpha$, size $n$, and update schedules $\Delta,\Delta'\in\Pn$,
it holds that
\begin{align}
  \label{eq:dequiv}
  \forall i \in \intz{n}: \dleft{\Delta}(i) = \dleft{\Delta'}(i) \wedge \dright{\Delta}(i) = \dright{\Delta'}(i)
  \quad\iff\quad \Delta \equiv \Delta'
\end{align}
and this implies $D_{f^{(\Delta)}_{\alpha,n}}=D_{f^{(\Delta')}_{\alpha,n}}$.
Remark that it is possible that
$\dleft{\Delta}(i)+\dright{\Delta}(i) \geq n$, in which case the image at cell
$i$ depends on the whole configuration. Moreover the previous inequality may be strict, 
meaning that the dependencies on both sides may overlap for some cell.
This will be a key in computing the dependency on $n$ of the
sensitivity to synchronism for rule $128$ for example.
Let
\[
  \dset{\Delta}(i)=
  \{ j \leq i \mid i-j \leq \dleft{\Delta}(i) \} \cup
  \{ j \geq i \mid j-i \leq \dright{\Delta}(i) \}
\]
be the set of cells that $i$ depends on under update schedule $\Delta\in\Pn$.
When $\dset{\Delta}(i) \neq \intz{n}$ then cell $i$ does not depend on
the whole configuration, and $\dset{\Delta}(i)$ describes precisely
$\Delta$, as stated in the following lemma.

\begin{lemma}
  \label{lemma:dset}
  For any $\Delta,\Delta'\in\Pn$, it holds that
  \[
    \forall i\in\intz{n}\;d_\Delta(i) = d_{\Delta'}(i) \neq \intz{n}\text{ implies }\Delta \equiv \Delta'.
  \]
\end{lemma}

\begin{proof}
  If $\dset{\Delta}(i) \neq \intz{n}$ then 
  $\dleft{\Delta}(i)$ and $\dright{\Delta}(i)$ do not overlap.
  Moreover, remark that 
%  from the linear structure of
%  ECAs influences we can deduce 
  $\dleft{\Delta}(i)$ and $\dright{\Delta}(i)$ can be deduced from $\dset{\Delta}(i)$.
  Indeed,
  \begin{align*}
  \dleft{\Delta}(i)=&\max\set{j \mid \forall k\in\intz{j}, i-j+k\in\dset{\Delta}(i)}\\
  \dright{\Delta}(i)=&\max\set{j \mid \forall k\in\intz{j}, i+j-k\in\dset{\Delta}(i)}
  \end{align*}
  
 % (for example with $n=6$ and $\dset{\Delta}(3)=\{1,2,3,4\}$
 % we deduce $\dleft{\Delta}(3)=2$ and $\dright{\Delta}(3)=1$).
  The result follows since knowing $\dright{\Delta}(i)$ and $\dleft{\Delta}(i)$
  for all $i \in \intz{n}$ allows to completely reconstruct $\lab_\Delta$,
  which would be the same as $\lab_{\Delta'}$ if $d_\Delta(i) = d_{\Delta'}(i)$
  for all $i\in\intz{n}$ (Formula~\ref{eq:dequiv}).
\end{proof}
%
% end toolbox
%%%%%%%%%%%%%%%%%%%%%%%%%%%%%%%%

% 0
%\subsubsection{Constant map (ECA rule 0)}
\subsection{Class I: Insensitive rules}
\label{ss:0}

This class contains the simplest dynamics, with sensitivity function $\frac{1}{3^n-2^{n+1}+2}$, and it is a good starting point for
our analysis.

%The following result is obvious, since we have
%$\forall n \geq 1: \forall x \in \bool^n: f_{0,n}(x)=\0^n$.

\begin{theorem}
  \label{theorem:0}
  $\sens{f_{0,n}}=\frac{1}{3^n-2^{n+1}+2}$ for any $n \geq 1$ and for $\alpha\in\set{0, 51, 204}$.
\end{theorem}
\begin{proof}
The result for ECA rule $0$ is obvious since $\forall n \geq 1: \forall x \in \bool^n: f_{0,n}(x)=\0^n$.
The ECA Rule 51 is based on the boolean function $r_{51}(x_{i-1},x_i,x_{i-1})=\neg x_i$ and ECA rule 204 is
the identity. Therefore, similarly to ECA rule $0$, for any $n$ their
interaction digraph has no arcs. Hence, there is only one equivalence 
class of update digraph, and one dynamics.
\end{proof}

% 51, 204
%\subsubsection{Identity and negation maps (ECA rules 51 and 204)}
\label{ss:51_204}

%Rule 51 is $r_{51}(x_{i-1},x_i,x_{i-1})=\neg x_i$ and rule 204 is
%$r_{204}(x_{i-1},x_i,x_{i-1})=x_i$. Similarly to rule 0, for any $n$ their
%interaction digraph has no arcs. Hence, there is only one equivalence 
%class of update digraph, and one dynamics.
%
%\begin{theorem}
%  \label{theorem:51-204}
%  $\sens{f_{51,n}}=\sens{f_{204,n}}=\frac{1}{3^n-2^{n+1}+2}$ for any $n \geq 1$.
%\end{theorem}

% 200
%\subsubsection{The ECA rule $200$}
\label{s:200}

The ECA rule $200$ also belongs to Class I. It has the following local function $r_{200}(x_1,x_2,x_3)=x_2 \wedge (x_1 \vee x_3)$.
Indeed, it is almost equal to the identity (ECA rule $204$), except for $r_{200}(\0,\1,\0)=\0$.
It turns out that, even if its interaction digraph has all of the $2n$ arcs,
this rule produces always the same dynamics, regardless of the update schedule.

\begin{theorem}
  \label{theorem:200}
  $\sens{f_{200,n}}=\frac{1}{3^n-2^{n+1}+2}$ for any $n \geq 1$.
\end{theorem}
\begin{proof}
  We prove that
  $f_{200,n}^{(\Delta)}(x)=f_{200,n}^{(\Delta^\texttt{sync})}(x)$ for any configuration 
  $x \in \bool^n$ and for any update schedule $\Delta \in\Pn$. 
  For any $i \in \intz{n}$ such that $x_i=\0$, the ECA rule $200$ is the identity,
  therefore it does not depend on the states of its neighbors which may have
  been updated before itself, \ie $f_{200,n}^{(\Delta)}(x)_i=\0=f_{200,n}^{(\Delta^\texttt{sync})}(x)_i$.
  Moreover, for any $i\in\intz{n}$ such that $x_i=\1$, if its two neighbors
  $x_{i-1}$ and $x_{i+1}$ are both in state $\0$ then they
   will remain in state $\0$ and
  $f_{200,n}^{(\Delta)}(x)_i=\0=f_{200,n}^{(\Delta^\texttt{sync})}(x)_i$, otherwise the ECA $200$ is the
  identity map and the two neighbors of cell $i$ also apply the identity,
  thus again $f_{200,n}^{(\Delta)}(x)_i=\1=f_{200,n}^{(\Delta^\texttt{sync})}(x)_i$.
\end{proof}

% % % % % % % % % % % % % % % % % % % % % %
% counting smaller update schedules
\subsection{Class II: Low sensitivity rules}
\label{cl:low}
%\subsubsection{Counting update schedules on directed cycles}

This class contains rules whose sensitivity function equals $\frac{2^n-1}{3^n-2^{n+1}+2}$.
This is a very interesting class that demands the development of specific arguments
and tools. However, the starting point is always the interaction digraph.
%\smallskip

\subsubsection{One-way ECAs.}

The following result counts the number of equivalence classes of update
schedules for ECA rules $\alpha$ having only arcs of the form $(i,i+1)$, or only
arcs of the form $(i+1,i)$ in their interaction digraph $G_{f_{\alpha,n}}$.

\begin{lemma}
  \label{lemma:count_uni}
  For the ECA rules $\alpha \in \set{3, 12, 15, 34, 60, 136, 170}$,
%  it holds $|\updates{f_{\alpha,n}}| = 2^n-1$.
  it holds that $|\updates{f_{\alpha,n}}| \leq 2^n-1$.
\end{lemma}
\begin{proof}
  The interaction digraph of these rules is the directed cycle on $n$
  vertices (with $n$ arcs). There can be only a forbidden cycle of length $n$ in the case that all arcs are
  labeled $\labminus$ (see Theorem~\ref{theorem:lab_valid}). 
  Except for the all $\labplus$ labeling (which is valid),
  any other labeling prevents the formation of an invalid cycle, since 
  the orientation of at least one arc is unchanged (labeled $\labplus$), and
  the orientation of at least one arc is reversed (labeled $\labminus$).
\end{proof}

% 170
%\subsubsection{Left shift map (ECA rule 170)}
\label{ss:170}

In the sequel we are going to exploit Lemma~\ref{lemma:count_uni} to obtain one of the main results of this section.
The ECA rule $170$, which is based on the following Boolean function:
$r_{170}(x_{i-1},x_i,x_{i+1})=x_{i+1}$, shows
the pathway.
%It reaches the bound given by
%Lemma~\ref{lemma:count_uni}.

\begin{theorem}
  \label{theorem:170}
  $\sens{f_{170,n}}=\frac{2^n-1}{3^n-2^{n+1}+2}$ for any $n \geq 2$.
\end{theorem}

\begin{proof}
  Let $f=f_{170,n}$ and $n \geq 2$. By definition, one finds that for any two 
  non-equivalent update schedules $\Delta \not\equiv \Delta'$
  it holds that
  \[
    \exists i_0\in\intz{n}\; \lab_{\Delta}((i_0+1,i_0))=\labplus \wedge
    \lab_{\Delta'}((i_0+1,i_0))=\labminus.
  \]
  Furthermore, since having $\lab_{\Delta'}((i+1,i))=\labminus$ for all $i\in\intz{n}$
  creates an invalid cycle of length $n$, there exists a minimal $\ell \geq 1$
  such that $\lab_{\Delta'}((i_0+\ell+1,i_0+\ell))=\labplus$ (this requires
  $n>1$). A part of the update digraph corresponding to $\Delta'$ is pictured below.
  \begin{center}
    \includegraphics{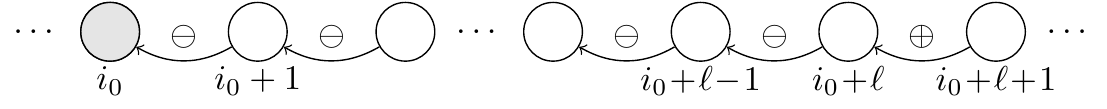}
  \end{center}
  By definition of the labels and the minimality of $\ell$ we have that %by induction that
  \[
    \forall 0 \leq k < \ell: f^{(\Delta')}(x)_{i_0+k}=x_{i_0+\ell+1}.
  \]
%  It follows that $f^{(\Delta')}(x)_{i_0}=x_{i_0+\ell+1}$, and
  Since for the update schedule $\Delta$ we have $f^{(\Delta)}(x)_{i_0}=x_{i_0+1}$,
  it is always possible to construct a configuration $x$ with
  $x_{i_0+1} \neq x_{i_0+\ell+1}$ such that the two dynamics differ, \ie,
  $f^{(\Delta)}(x)_{i_0} \neq f^{(\Delta')}(x)_{i_0}$.
  The result holds by Formula~\ref{eq:max-sensitive}.
\end{proof}

% 3, 12, 15, 34, 60, 136
%\subsubsection{ECA rules 3, 12, 15, 34, 60, 136}
\label{ss:3_12_15_34_60_136}

Generalizing the idea behind the construction used for ECA rule $170$ one
may prove that ECA rules $3, 12, 15, 34, 60, 136$ have identical sensitivity functions.

%A construction more general yet similar to the one for ECA rule 170 leads to
%the same measure of sensitivity to synchronism for the ECA rules
%3, 12, 15, 34, 60, 136.

\begin{theorem}
  \label{theorem:3-12-15-34-60-136}
  $\sens{f_{\alpha,n}}=\frac{2^n-1}{3^n-2^{n+1}+2}$ for any $n \geq 2$
  and for all $\alpha\in\set{3,12,15,34,60,136}$.
\end{theorem}
\begin{proof}
  We present the case when the interaction digraph has only arcs of type $(i+1,i)$
  (such as rule 170), the case $(i,i+1)$ is symmetric. Fix $n\geq 2$ and choose
  two update schedules $\Delta, \Delta'\in\Pn$ such that $\Delta \not \equiv \Delta'$, then
  it holds that
  \[
  \begin{array}{c}
    \exists i_0\in\intz{n}\; \lab_{\Delta}((i_0+1,i_0))=\labplus \wedge
    \lab_{\Delta'}((i_0+1,i_0))=\labminus\\[.5em]
    \exists \ell\in\intz{n}\;
    (\forall 0 \leq k < \ell\; \lab_{\Delta'}((i_0+k+1,i_0+k))=\labminus) \wedge
    (\lab_{\Delta'}((i_0+\ell+1,i_0+\ell))=\labplus).
  \end{array}
  \]
  Fix $\alpha\in\set{3,12,15,34,60,136}$ and let $r$ be the corresponding Boolean function.
  Moreover let $f=f_{\alpha,n}$. % be the rule and size under consideration,
  We know that for any $x \in \bool^n$ we will have
   $f^{(\Delta)}(x)_{i_0}=r(b,x_{i_0},x_{i_0+1})$ for any $b \in \bool$.
  Our goal is to construct a configuration $x\in\zu^n$ such that
  $f^{(\Delta)}(x)_{i_0} \neq f^{(\Delta')}(x)_{i_0}$.
  In order to start, we need
  \begin{align}
    \label{eq:start_small}
    \exists x_{i_0},x_{i_0+1},o_1,o_2 \in \bool\,\forall b,b' \in \bool\;
    \begin{array}[t]{l}
      \phantom{\text{and }} r(b,x_{i_0},x_{i_0+1}) \neq r(b',x_{i_0},o_1)\\
      \text{and } r(x_{i_0},x_{i_0+1},o_2)=o_1.
    \end{array}
  \end{align}
  In other words, we can choose $x_{i_0},x_{i_0+1}$ so that there is a target output $o_1$ 
  for $f^{(\Delta')}(x)_{i_0+1}$, such that if
  $f^{(\Delta')}(x)_{i_0+1}=o_1$ then $f^{(\Delta)}(x)_{i_0} \neq
  f^{(\Delta')}(x)_{i_0}$ (the values of $b$ and $b'$ do not matter). Similarly,
  given $x_{i_0},x_{i_0+1},o_1$, there is a target output $o_2$ for
  $f^{(\Delta')}(x)_{i_0+2}$. We can now construct $x$ by finite induction by expressing
  such target outputs. In order to continue we want
  (the idea is to use this by induction for all $0 < k \leq \ell$)
  \begin{align}
    \label{eq:continue_small}
    \forall x_{i_0+k-1},o_k \in\bool\,\exists x_{i_0+k},o_{k+1}\in\bool\; 
    r(x_{i_0+k-1},x_{i_0+k},o_{k+1})=o_k.
  \end{align}
  
  If Formulas~\ref{eq:start_small} and~\ref{eq:continue_small} hold %for some rule, 
  then we can construct the desired configuration $x$. Indeed, Formula~\ref{eq:start_small}
  gives $x_{i_0}, o_1, x_{i_0+1}, o_2$, and by induction, knowing
  $x_{i_0+k-1},o_k$ Formula~\ref{eq:start_small} gives $x_{i_0+k},o_{k+1}$ for
  $0 < k \leq \ell$. The construction ends with
  $x_{i_0+\ell+1}=o_{\ell+1}$.
  A sequence $(x_i,o_i)\in\zu^2$ for $i\in\intz{l}$ which satisfies both Formulas~\ref{eq:start_small} and~\ref{eq:continue_small} 
  is called a \emph{witness} sequence. 
  %All objectives are achieved, hence
  Given a witness sequence $(x_i,o_i)_{i\in\intz{l}}$ it holds
  $f^{(\Delta)}(x)_{i_0} \neq f^{(\Delta')}(x)_{i_0}$, and hence, by
  Formula~\ref{eq:max-sensitive} we have the result. 
  %Exhaustive searches in
  %constant time for each rule give the following results.
  We end by providing the witness sequences for all the local rules
  in the hypothesis. We start by those rules which have an interaction digraph made by arcs of type $(i+1,i)$.
%  \medskip

%  Witnesses for rules whose 
  \begin{itemize}
    \item Rule $ 34 $ :
    \begin{itemize}
      \item Formula~\eqref{eq:start_small}: $x_{i_0}= \0 $, $x_{i_0-1}= \1 $, $o_1= \1 $, $o_2= \1 $.
      \item Formula~\eqref{eq:continue_small}: 
      %for $(x_{i_0+k-1},o_k)$ in lexicographic order, there exist\\
      $(x_{i_0+k},o_{k+1})= 
      \begin{cases} (\0, \0),&\text{if $k$ is even}\\ (\0, \1), &\text{otherwise}% (\0, \0), (\0, \1)
      \end{cases}$
    \end{itemize}  
    \item Rule $ 136 $ :
    \begin{itemize}
      \item Formula~\eqref{eq:start_small}: $x_{i_0}= \1 $, $x_{i_0-1}= \1 $, $o_1= \0 $, $o_2= \0 $.
      \item Formula~\eqref{eq:continue_small}: 
      %for $x_{i_0+k-1},o_k$ in lexicographic order, there exist\\ 
      $(x_{i_0+k},o_{k+1})= 
      \begin{cases} (\0, \0),&\text{if $k$ is even}\\ (\1, \1),&\text{otherwise}\end{cases}$%(\0, \0), (\1, \1) $.
    \end{itemize}
  \end{itemize}

%  \medskip
We conclude with the witness sequences for the rules which have
interaction digraph made by arcs of type $(i,i+1)$.
  \begin{itemize}
    \item Rule $ 3 $ :
    \begin{itemize}
      \item Formula~\eqref{eq:start_small}: $x_{i_0}= \0 $, $x_{i_0+1}= \0 $, $o_{-1}= \1 $, $o_{-2}= \0 $.
      \item Formula~\eqref{eq:continue_small}: 
      %for $x_{i_0-k+1},o_{-k}$ in lexicographic order, there exist\\ 
      $(x_{i_0-k},o_{-k-1})=\begin{cases} (\0, \1),&\text{if $k$ is even}\\ (\0, \0),&\text{otherwise}\end{cases}$% (\0, \1), (\0, \0) $.
    \end{itemize}
    \item Rule $ 12 $ :
    \begin{itemize}
      \item Formula~\eqref{eq:start_small}: $x_{i_0}= \1 $, $x_{i_0+1}= \1 $, $o_{-1}= \0 $, $o_{-2}= \1 $.
      \item Formula~\eqref{eq:continue_small}: 
      %for $x_{i_0-k+1},o_{-k}$ in lexicographic order, there exist\\ 
      $(x_{i_0-k},o_{-k-1})=\begin{cases} (\0, \0),&\text{if $k$ is even}\\ (\1, \0),&\text{otherwise}\end{cases}$ % (\0, \0), (\1, \0) $.
    \end{itemize}
    \item Rule $ 15 $ :
    \begin{itemize}
      \item Formula~\eqref{eq:start_small}: $x_{i_0}= \0 $, $x_{i_0+1}= \0 $, $o_{-1}= \1 $, $o_{-2}= \0 $.
      \item Formula~\eqref{eq:continue_small}: 
      %for $x_{i_0-k+1},o_{-k}$ in lexicographic order, there exist\\ 
      $(x_{i_0-k},o_{-k-1})=\begin{cases} (\0, \1),&\text{if $k$ is even}\\ (\0, \0),&\text{otherwise}\end{cases}$ % (\0, \1), (\0, \0) $.
    \end{itemize}
    \item Rule $ 60 $ :
    \begin{itemize}
      \item Formula~\eqref{eq:start_small}: $x_{i_0}= \0 $, $x_{i_0+1}= \0 $, $o_{-1}= \1 $, $o_{-2}= \1 $.
      \item Formula~\eqref{eq:continue_small}: 
      %for $x_{i_0-k+1},o_{-k}$ in lexicographic order, there exist\\ 
      $(x_{i_0-k},o_{-k-1})=\begin{cases} (\0, \0),&\text{if $k$ is even}\\ (\0, \1),&\text{otherwise}\end{cases}$ % (\0, \0), (\0, \1) $.
    \end{itemize}
  \end{itemize}
\end{proof}

\begin{example}
        Consider the ECA rule $\alpha=34$ and a size $n=6$.  Given 
        the two distinct update schedules
        $\Delta=(\{3,4\},\{5\},\{2\},\{0,1\})$
        and
        $\Delta'=(\{3,4\},\{5\},\{0,1,2\})$. Let
        $i_0=1$ and $\ell=2$. The following is a witness sequence 
        (see the proof of Theorem~\ref{theorem:3-12-15-34-60-136}):
        $x_{i_0-1}=\1,x_{i_0}=\0,
        o_1=\1,x_{i_0+1}=\0,
        o_2=\1,x_{i_0+2}=\0,
        o_3=\1,x_{i_0+3}=\0,
        x_{i_0+4}=o_4=\1$.
        %ensuring
        By construction it ensures
        \begin{align*}
          &f^{(\Delta')}(\1\0\0\0\1\0)_{i_0}
          = r(\1,\0,r(\0,\0,r(\0,\0,\1)))
          = r(\1,\0,r(\0,\0,\1))
          = r(\1,\0,\1)
          = \1\\
          \neq
          &f^{(\Delta)}(\1\0\0\0\1\0)_{i_0}
          = r(\1\0\0)
          = \0.
        \end{align*}
        \qed
\end{example}

%%%%%%%%%%%%%%%%%%%%%%%%%%%%%%%%
% 28,32,44,140
\subsubsection{Exploiting patterns in the update digraph}

In this subsection we are going to develop a proof technique which
characterizes the number of non-equivalent update schedules according
to the presence of specific patterns in their interaction digraph. This
will concern ECA rules $28, 32, 44$ and $140$.
When $n$ is clear from the context, we will simply denote
$f_{\alpha}$ instead of $f_{\alpha,n}$ with $n\in\set{28,32,44,140}$.

%\subsubsection*{\TODO{(to be organized)} ECA rules 28, 32, 44, 140}
\label{s:28_32_44_140}
%\TODO{move this section before 128, 160, 162}

% 32
%\subsubsection*{\TODO{(to be organized)} ECA rules 32}
We begin with the ECA Rule $32$ which is based on the Boolean function 
$r_{32}(x_1,x_2,x_3)= x_1 \wedge \neg x_2 \wedge x_3$.

\begin{lemma}
\label{pattern32}
Fix $n\in\N$. For any update schedule $\Delta\in\Pn$, for any 
configuration $x\in\zu^n$ and for any $i\in\intz{n}$, the following holds:
\begin{equation*}
  f^{(\Delta)}_{32}(x)_i =\1 \iff \lab_{\Delta}((i+1,i))=\lab_{\Delta}((i-1,i))=\labplus \wedge (x_{i-1},x_i,x_{i+1})=(\1,\0,\1).
\end{equation*}
\end{lemma}
\begin{proof}\mbox{}\\
%  \noindent$(\Rightarrow)$ 
\noindent$(\Leftarrow)$ 
Since $lab_{\Delta}((i+1,i))=lab_{\Delta}((i-1,i))=\labplus$ (see Figure~\ref{figurepatt32}) that means that cell $i$ is not updated after cells $i-1$ and $i+1$, therefore $f^{(\Delta)}_{32}(x)_i =r_{32}(x_{i-1},x_i,x_{i+1})=r_{32}(\1,\0,\1)=\1$.\\
%$(\Leftarrow)$ 
\noindent$(\Rightarrow)$
Choose $x\in\zu^n$ such that $f^{(\Delta)}_{32}(x)_i=\1$.
Assume that $\lab_{\Delta}((i+1,i))=\lab_{\Delta}((i-1,i))=\labplus$ but $(x_{i-1},x_i,x_{i+1}) \neq (\1,\0,\1)$. 
By the same reasoning as above, we have 
$f^{(\Delta)}_{32}(x)_i =r_{32}(x_{i-1},x_i,x_{i+1})=\0$, 
since $(x_{i-1},x_i,x_{i+1}) \neq (\1,\0,\1)$. 
Now, assume $\lab_{\Delta}((i+1,i))=\labminus$ or $\lab_{\Delta}((i-1,i))=\labminus$. 
Then,
\begin{equation*}
  f^{(\Delta)}_{32}(x)_i =
\begin{cases}
r_{32}(\0,\0,\1)=\0, \text{if }  lab_{\Delta}((i-1,i))=\labminus \wedge lab_{\Delta}((i+1,i))=\labplus\\
r_{32}(\1,\0,\0)=\0, \text{if }  lab_{\Delta}((i-1,i))=\labplus \wedge lab_{\Delta}((i+1,i))=\labminus\\
r_{32}(\0,\0,\0)=\0, \text{otherwise}
\end{cases}
\end{equation*}
which contradicts the hypothesis.
\end{proof}

\begin{figure}
\centerline{\includegraphics{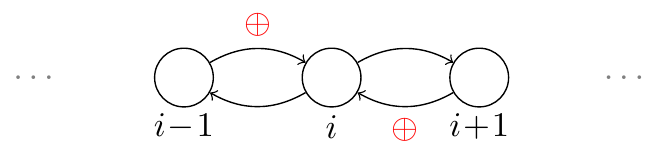}}
	\caption{labeling presented in the Lemma \ref{pattern32}.  This is the only situation in which we can obtain a cell updated equal to \1.}
	\label{figurepatt32}
	\end{figure}
	
\begin{corollary}
\label{corollaryminus}
  Fix $n\in\N$.
  For any update schedule $\Delta\in\Pn$, for any configuration $x\in\zu^n$ and $i\in\intz{n}$, if $\lab_{\Delta}((i-1,i))=\labminus$ or $\lab_{\Delta}((i+1,i))=\labminus$, then $f^{(\Delta)}_{32}(x)_i =\0$.
\end{corollary}

\begin{lemma}
For any $n\in\N$. Consider a pair of update schedules 
$\Delta, \Delta'\in\Pn$. Then, $D_{f^{(\Delta)}_{32,n}} \neq D_{f^{(\Delta')}_{32,n}}$
if and only if there exists $i\in\intz{n}$ such that one of the following 
holds:
% \iff \exists 0 \leq i \leq n-1 $ such that 
\begin{enumerate}
\item 
$\lab_{\Delta}((i+1,i))=\lab_{\Delta}((i-1,i))=\labplus$ and either $\lab_{\Delta'}((i+1,i))=\labminus$ or $\lab_{\Delta'}((i-1,i))=\labminus$;
\item
 $\lab_{\Delta'}((i+1,i))=\lab_{\Delta'}((i-1,i))=\labplus$ and either $\lab_{\Delta}((i+1,i))=\labminus$ or $\lab_{\Delta}((i-1,i))=\labminus$.
\end{enumerate}
\end{lemma}

\begin{proof}\mbox{}\\
\noindent $(\Leftarrow)$ WLOG, suppose that $\lab_{\Delta}((i+1,i))=\lab_{\Delta}((i-1,i))=\labplus$ and $\lab_{\Delta'}((i+1,i))=\labminus$ or $\lab_{\Delta'}((i-1,i))=\labminus$ (the other case is the same where $\Delta$ and $\Delta'$ are exchanged). 
Then, by Lemma \ref{pattern32} one can finds $f^{(\Delta)}_{32}(x)_i =\1$ and by Corollary~\ref{corollaryminus}, $f^{(\Delta')}_{32}(x)_i =\0$. Therefore, $D_{f^{(\Delta)}_{32,n}} \neq D_{f^{(\Delta')}_{32,n}}$.\\
\noindent $(\Rightarrow)$ Suppose that for every 
$i\in\intz{n}$ one of the following holds:
\begin{itemize}
\item[] (Case 1) $\lab_{\Delta}((i+1,i))=\lab_{\Delta}((i-1,i))=\lab_{\Delta'}((i+1,i))=\lab_{\Delta'}((i-1,i))=\labplus$
\item[] (Case 2) $(\lab_{\Delta}((i+1,i)), \lab_{\Delta}((i-1,i))) \neq (\labplus,\labplus) \neq (\lab_{\Delta'}((i+1,i)), \lab_{\Delta'}((i-1,i)))$.
\end{itemize}
We will show that in both cases $D_{f^{(\Delta)}_{32,n}} = D_{f^{(\Delta')}_{32,n}}$.
Let $j\in\intz{n}$ and consider a configuration $x\in\zu^n$ such that:
$(x_{j-1},x_j,x_{j+1}) \neq (\1,\0,\1)$, then by Lemma~\ref{pattern32}, $f^{(\Delta)}_{32}(x)_j = f^{(\Delta')}(x)_j =\0$.
Now suppose $(x_{j-1},x_j,x_{j+1})=(\1,\0,\1)$. 
If we are in Case $1$, then $f^{(\Delta)}_{32}(x)_j = f^{(\Delta')}_{32}(x)_j =\1$.
If we are in Case $2$, then $f^{(\Delta)}_{32}(x)_j = f^{(\Delta')}_{32}(x)_j =\0$.
By the generality of $j$, $f^{(\Delta)}_{32}(x) = f^{(\Delta')}_{32}(x)$ and by the generality of $x$, $D_{f^{(\Delta)}_{32,n}} = D_{f^{(\Delta')}_{32,n}}$.
\end{proof}
%\begin{theorem}
%	\label{theorem:32}
%  $\sens{f_{32,n}}=\frac{2^n-1}{3^n-2^{n+1}+2}$ for any $n \geq 3$.
%\end{theorem}
%
%\begin{proof}
%Given a configuration of length $n$, the pattern in Lemma~\ref{pattern32} may be present in $k$ cells with $ 1\leq k \leq n$ (it must be present in at least one cell because otherwise we would have a $\labminus$ cycle).
%Therefore, there are $\sum_{k=1}^n \binom{n}{k}=2^n -1$ different dynamics.  
%\end{proof}

% 28, 44, 140
%\subsection*{\TODO{(to be organized)} ECA rules 28, 44, 140}

For the ECA rules $28, 44$ and $140$ we are going to develop a similar construction as the one for 
ECA rule $32$ but before let us recall the Boolean functions which they are based on.
We start with ECA rule $44$ which is based on the Boolean function 
$r_{44}(x_1,x_2,x_3)= (\neg x_1 \wedge x_2 ) \vee (x_1\wedge \neg x_2\wedge x_3)$ which implies that  
$r_{44}(\1,\0,\1)=r_{44}(\0,\1,\1)=r_{44}(\0,\1,\0)=\1$.

%\paragraph{\TODO{(to be organized)} Rule 44}

\begin{notation}
Let us call $\star$ a possible value of a cell in the configuration that has no effect on the result of the update procedure over the cells under consideration.
At the same time, we will use a letter to represent the value of a cell in the configuration that is unknown but which has an impact on the result of the update procedure over the cells under consideration.
\end{notation}

\begin{lemma}
\label{eq44}
Given two update schedules $\Delta$ and $\Delta'$, if there exists $i \in \intz{n}$ such that:
\begin{itemize}
\item $\lab_{\Delta}((i,i-1))=\labminus \wedge\lab_{\Delta'}((i,i-1))=\labplus $
\item $\lab_{\Delta}((i-1,i))=\lab_{\Delta'}((i-1,i))=\labplus$ 
\item $\lab_{\Delta}((j,j-1))=\lab_{\Delta'}((j,j-1)) \wedge\lab_{\Delta}((j-1,j))=\lab_{\Delta'}((j-1,j))$ for each $j \neq i, j \in \intz{n}$
\end{itemize}
then $D_{f^{(\Delta)}_{44,n}} = D_{f^{(\Delta')}_{44,n}}$.
\end{lemma}

\begin{proof}
Given two update schedules $\Delta$ and $\Delta'$, we prove that $(f^{(\Delta)}(x)_{i-1} = f^{(\Delta')}(x)_{i-1}) \wedge (f^{(\Delta)}(x)_{i} = f^{(\Delta')}(x)_{i})$ for every possible starting configuration $x$. \\
  Starting from the case with $x_{i-1}=x_{i}=\1$ (see Figure \ref{figure44case11}), one obtains cells $i-1$ and $i$ updated to states $r_{44}(y,\1,\0)$ and $\0$ (respectively) according to the $\Delta$ update schedule and to states $r_{44}(y,\1,\1)$ and $\0$ (respectively) according to the $\Delta'$ update schedule.
According to the rule, we know that $r_{44}(\0,\1,\0)=r_{44}(\0,\1,\1)=\1$ and $r_{44}(\1,\1,\0)=r_{44}(\1,\1,\1)=\0$ consequently the equivalence holds in the case of $x_{i-1}=x_{i}=\1$.\\
\begin{figure}
\centering
 \centerline{\includegraphics[width=\textwidth]{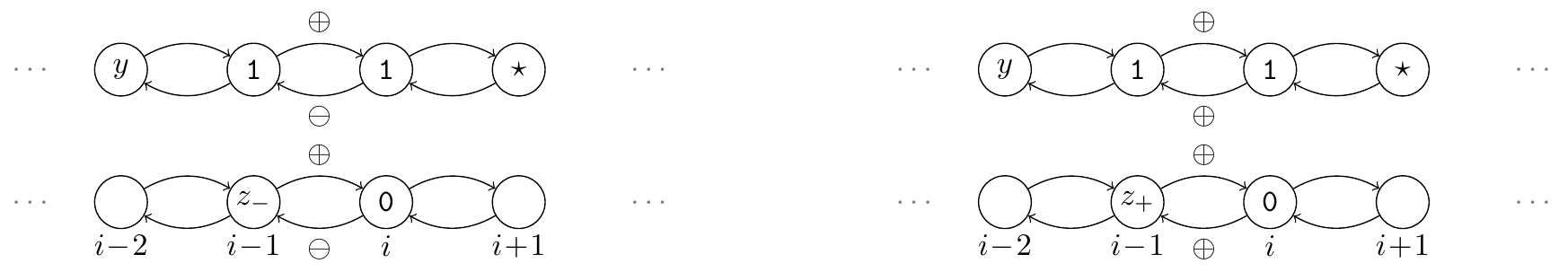}}
 	\caption{on the left, starting from a configuration $(y,\1,\1,\star)$ the schedule $\Delta$ updates before the cell $i$ which becomes \0, after the rule is applied on the cell $i-1$, therefore $z_-=r_{44}(y,\1,\0)$; on the right, the two cells are updated at the same time or after and $z_+=r_{44}(y,\1,\1)$.}
	\label{figure44case11}
\end{figure}
\begin{figure}
\centering
\centerline{\includegraphics[width=\textwidth]{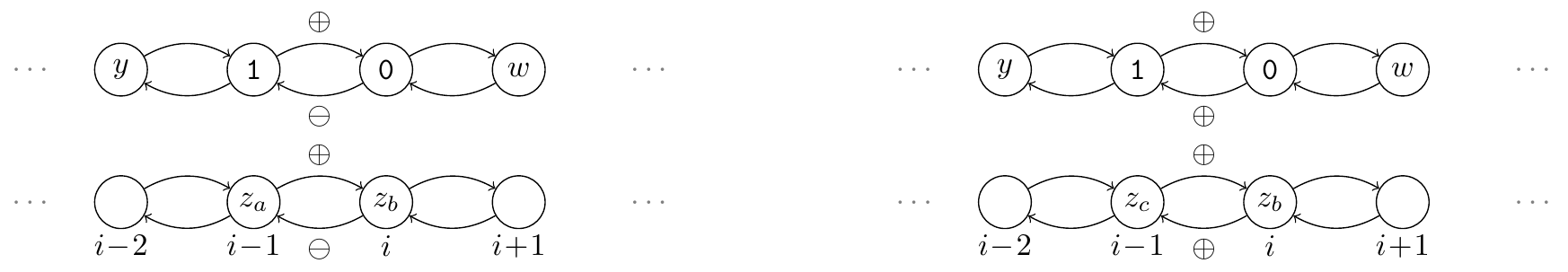}}
	\caption{on the left, starting from a configuration $(y,\1,\0,w)$ the schedule $\Delta$ updates before the cell $i$ to $z_b=r_{44}(\1,\0,w)$, after the rule is applied on the cell $i-1$ to obtain $z_a=r_{44}(y,\1,z_b)$; on the right, the two cells are updated at the same time or after, therefore $z_c=r_{44}(y,\1,\0)$ and $z_b=r_{44}(\1,\0,w)$.}
	\label{figure44case10}
\end{figure}
If we consider $x_{i-1}=\1$ and $x_{i}=\0$ (see Figure \ref{figure44case10}), one obtains cells $i-1$ and $i$ updated to states  $r_{44}(y,\1,r_{44}(\1,\0,w))$ and $r_{44}(\1,\0,w)$ (respectively) according to the $\Delta$ update schedule and to states $r_{44}(y,\1,\0)$ and $r_{44}(\1,\0,w)$ (respectively) according to the $\Delta'$ update schedule.
Like in the previous case, the result of the update procedure depends only on the $y$ value which will be the same in $\Delta$ and in $\Delta'$ consequently the equivalence holds in this case.
\begin{figure}
\centering
 \centerline{\includegraphics[width=\textwidth]{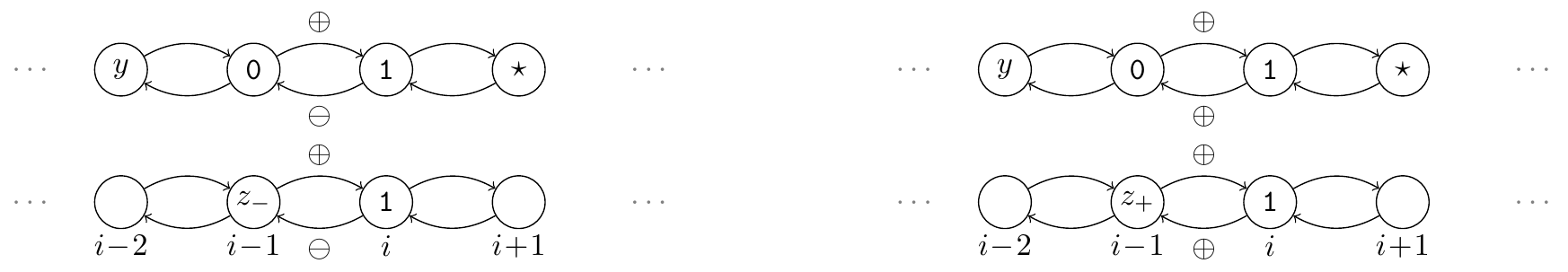}}	
 \caption{on the left, starting from a configuration $(y,\0,\1,\star)$ the schedule $\Delta$ updates before the cell $i$ which becomes $r_{44}=(\0,\1,\star)=\1$, after the rule is applied on the cell $i-1$, where $z_-=r_{44}=(y,\0,\1)$; on the right, the two cells are updated at the same time or after, we can obtain $x_i=\1$ and $x_{i-1}=z_+=z_-=r_{44}=(y,\0,\1)$.}
	\label{44case01}
\end{figure}
If we consider the opposite case $x_{i-1}=\0$ and $x_{i}=\1$ (see Figure \ref{44case01}), one obtains cells $i-1$ and $i$ updated to states $r_{44}(y,\0,\1)$ and $\1$ (respectively) according to the $\Delta$ update schedule and to states $r_{44}(y,\0,\1)$ and $\1$  (respectively) according to the $\Delta'$ update schedule, consequently the equivalence holds also in this case.
\begin{figure}
\centering
\centerline{\includegraphics[width=\textwidth]{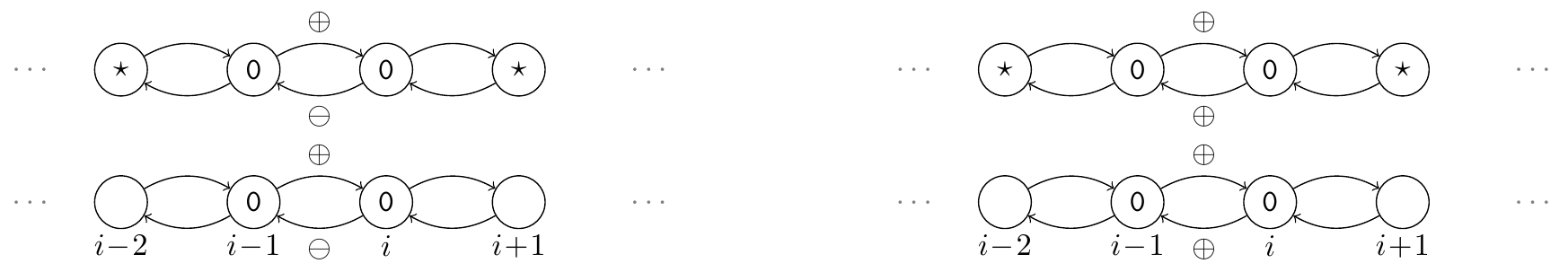}}
	\caption{on the left, starting from a configuration $(\star,\0,\0,\star)$ the schedule $\Delta$ updates before the cell $i$ which becomes \0, after the rule is applied on the cell $i-1$; on the right, the two cells are updated at the same time or after. Consider that $r_{44}=(\0,\0,\star)=r_{44}=(\star,\0,\0)=\0$.}
	\label{case00}
\end{figure}
The last case corresponds to $x_{i-1}=x_{i}=\0$ (see Figure \ref{case00}), one obtains cells $i-1$ and $i$ updated to states $\0$ and $\0$ according to $\Delta$ and $\Delta'$, consequently the equivalence holds. The two different update schedules give the same configurations independently from the initial configuration, in other words $D_{f^{(\Delta)}_{44,n}} = D_{f^{(\Delta')}_{44,n}}$.
\end{proof}

\begin{lemma}
\label{noteq44}
Given two update schedules $\Delta$ and $\Delta'$, $D_{f^{(\Delta)}_{44,n}} \neq D_{f^{(\Delta')}_{44,n}} \iff \exists i,i \in \intz{n}$ such that $ lab_{\Delta}((i-1,i))=\labminus \wedge lab_{\Delta'}((i-1,i))=\labplus $.
\end{lemma}

\begin{proof}
We can consider $lab_{\Delta}((i,i-1))=lab_{\Delta'}((i,i-1))=\labplus$ because according to Lemma \ref{eq44} the value of $lab_{\Delta'}((i,i-1))$ cannot  change the dynamics that we are considering and the value of $lab_{\Delta}((i,i-1))$ must be $\labplus$ given the $\labminus$ in the opposite sense.\\
We can consider equal labelings over the other transitions.\\
Let $j$ be a cell such that $lab_{\Delta'}((j,j+1))=\labplus$ and $lab_{\Delta'}((j+k,j+k+1))=\labminus$ for all $1\leq k \leq i-j-1$. 
Such a $j$ must exist since otherwise we would have a $\labminus$ cycle of length $n$.
Now, let $x\in\{\0,\1\}^n$ be any configuration of length $n$ such that
\begin{equation}
\label{redblue}
x_{[j,i+1]}=
\begin{cases}
\1(\1)^{(i-1)-j-1}\0\1\1,\mbox{ if }(i-1)-j-1 \mod2=0\mbox{ or }j=i-2\\
\0(\1)^{(i-1)-j-1}\0\1\1, \mbox{ otherwise }
\end{cases}
\end{equation}
Then we have
$$
\left(f^{\Delta}(x)_{[j+1,i]}\right)=
\begin{cases}
\0(\1\0)^{\lfloor \frac{(i-1)-j-2}{2} \rfloor}\1\1\0,\mbox{ if }(i-1)-j-1 \mod2=0\\
\1\0,\mbox{ if }j=i-2\\
\1(\0\1)^{\lfloor \frac{(i-1)-j-2}{2} \rfloor}\1\0, \mbox{ otherwise }
\end{cases}
.
$$
In general we can always obtain $f^{\Delta}(x)_{i}=\0$ (see Figure \ref{noteqfigure44}). 
The update schedule $\Delta'$ gives $f^{\Delta'}(x)_{i}=\1$. 
Therefore, $f^{\Delta}(x)_{i} \neq f^{\Delta'}(x)_{i}$ and $D_{f^{(\Delta)}_{44,n}} \neq D_{f^{(\Delta')}_{44,n}}$.
\end{proof}

\begin{figure}
\centering
\centerline{\includegraphics[scale=0.7]{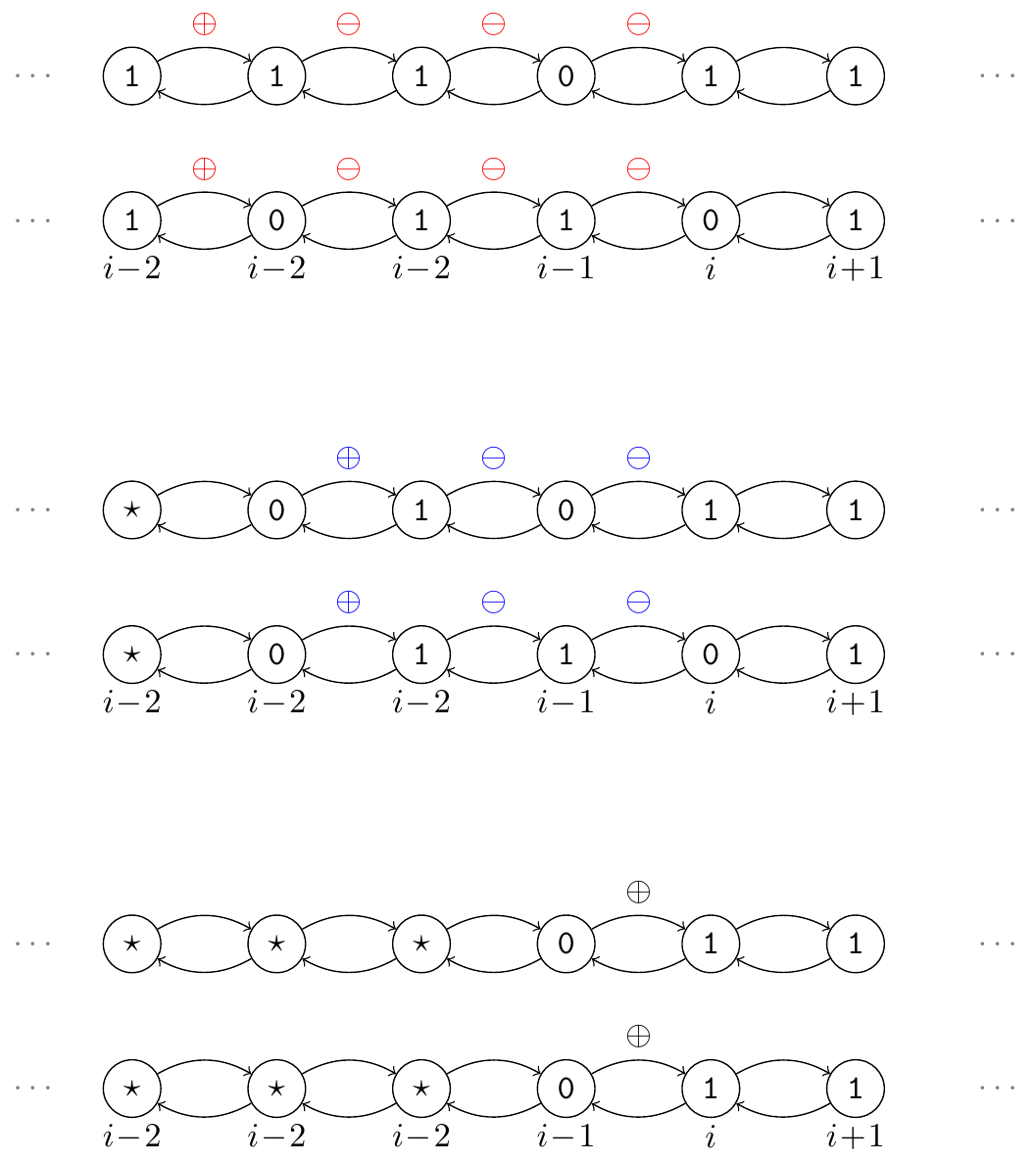}}
	\caption{since we cannot have a $\labminus$ cycle in the labeled interaction digraph of update schedule $\Delta'$ there must be a cell $j$ such that $lab_{\Delta'}((j,j+1))=\labplus$ and $lab_{\Delta'}((j+k,j+k+1))=\labminus$ for all $1\leq k \leq i-j-1$. The blue and the red updates shows the different cases of equation \ref{redblue} and the black one represent $\Delta'$.}
	\label{noteqfigure44}
	\end{figure}
	
Consider that the previous lemma is sufficient to determine that two update schedules that differ 
in at least one cell i such that $lab_{\Delta}((i-1,i))=\labminus \wedge lab_{\Delta'}((i-1,i))=\labplus$ 
generate two different dynamics. In fact, we can focus on one of these cells to build a configuration in
which the cell updated in different values.
\smallskip	
	
%\paragraph{\TODO{(to be organized)} ECA rule 28\\\\}

Consider now the ECA rule $28$, it is based on 
$r_{28}(x_1,x_2,x_3)= (\neg x_1\wedge x_2) \vee ( x_1 \wedge \neg x_2 \wedge \neg x_3)$
and hence $r_{28}(\1,\0,\0)=r_{28}(\0,\1,\1)=r_{28}(\0,\1,\0)=\1$. 
Remark that the Lemma \ref{eq44} holds also for this rule.
The only difference is in the proof, for completeness we show the 
equivalence that holds for every possible starting configuration 
(see Figure~\ref{figure2generalcase}).\\\\
For this rule also the Lemma \ref{noteq44} can be applied. The main idea is the same. 
In fact, let $x\in\set{\0,\1}^n$ be any configuration of length $n$ such that
\begin{equation}
%\label{eq:redblue-28}
x_{[j,i+1]}=
\begin{cases}
\1(\1)^{(i-1)-j-1}\1\0\0,\mbox{ if }(i-1)-j-1 \mod2=0\mbox{ or }j=i-2\\
\0(\1)^{(i-1)-j-1}\1\0\0, \mbox{ otherwise }
\end{cases}
\end{equation}
Then we have
$$
\left(f^{\Delta}(x_{[j+1,i]})\right)=
\begin{cases}
(\0\1)^{\lfloor \frac{(i-1)-j}{2} \rfloor}\0\0,\mbox{ if }(i-1)-j-1 \mod2=0\\
\0\0,\mbox{ if }j=i-2\\
\1(\0\1)^{\lfloor \frac{(i-1)-j-2}{2} \rfloor}\0\0, \mbox{ otherwise }
\end{cases}
.
$$
In general we obtain $f^{\Delta}(x)_{i}=\0$. 
The update schedule $\Delta'$ gives $f^{\Delta'}(x)_{i}=\1$. 
Therefore, $f^{\Delta}(x)_{i} \neq f^{\Delta'}(x)_{i}$ and $D_{f^{(\Delta)}_{28,n}} \neq D_{f^{(\Delta')}_{28,n}}$.

\begin{figure}
	\centering
	\centerline{\includegraphics[scale=0.7]{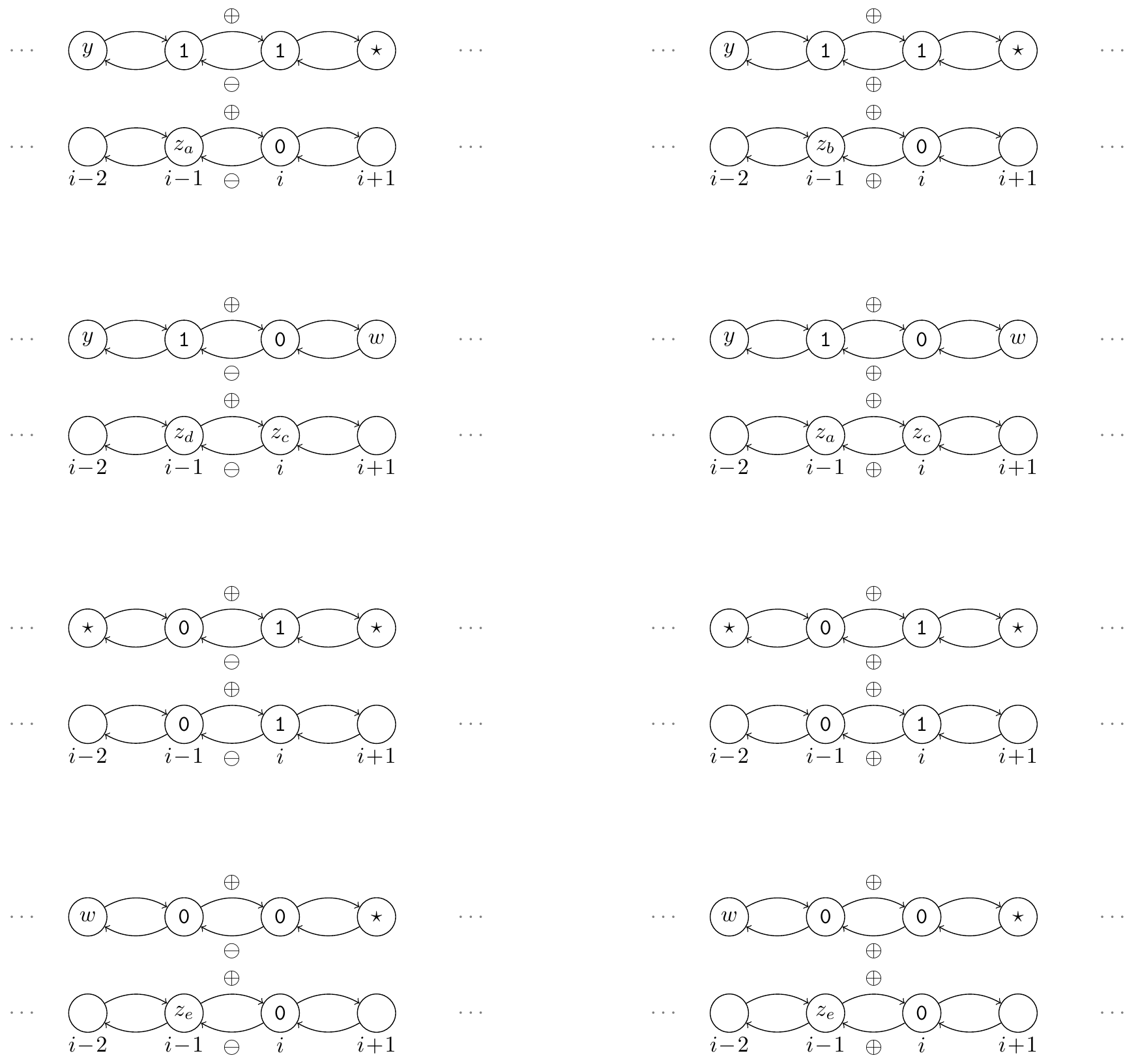}}
	\caption{starting from every possible configuration, only a difference in a label over the edge between $i+1$ and $i$ is not sufficient in order to obtain different final configurations for rule 28. In the figure: $z_a=r_{28}(y,\1,\0)$, $z_b=r_{28}(y,\1,\1)$, $z_c=r_{28}(\1,\0,w)$, $z_d=r_{28}(y,\1,z_c)$ and $z_e=r_{28}(w,\0,\0)$. We need also to consider that $z_a=z_b=y$ and $z_c=z_d=\0$ if $y=\1$, $z_c=z_d=\1$ otherwise.}
	\label{figure2generalcase}
\end{figure}

Consider that the previous lemma is sufficient to determine that two update schedules that differ 
in at least one cell i such that $lab_{\Delta}((i-1,i))=\labminus \wedge lab_{\Delta'}((i-1,i))=\labplus$ 
generate two different dynamics. In fact, we can focus on one of these cells to build a configuration in
which the cell updated in different values.
\smallskip

Let us now focus our attention on ECA rule $140$ which is based on the Boolean function
$r_{140}(x_1,x_2,x_3)= ( \neg x_1 \lor x_3) \wedge x_2$, 
that is to say  $r_{140}(\1,\1,\1)=r_{140}(\0,\1,\1)=r_{140}(\0,\1,\0)=\1$.
%\paragraph{\TODO{(to be organized)} ECA rule $140$}

\begin{lemma}
\label{eq140}
For any $n>3$, given two update schedules $\Delta, \Delta'\in\Pn$, if there exists $i \in \intz{n}$ such that
\begin{itemize}
\item $\lab_{\Delta}((i,i+1))=\labminus$ and $\lab_{\Delta'}((i,i+1))=\labplus $
\item $\lab_{\Delta}((i+1,i))=\lab_{\Delta'}((i+1,i))=\labplus$ 
\item $\lab_{\Delta}((j,j-1))=\lab_{\Delta'}((j,j-1))$ and $\lab_{\Delta}((j-1,j))=\lab_{\Delta'}((j-1,j))$ for each $j \neq i+1, j \in \intz{n}$
\end{itemize}
then $D_{f^{(\Delta)}_{140,n}} = D_{f^{(\Delta')}_{140,n}}$.

\end{lemma}

\begin{proof}
Given the two update schedules $\Delta, \Delta'\in\Pn$, using the same reasoning as for Lemma~\ref{eq44}, one can prove that $f^\Delta(x)_{i} = f^{\Delta'}(x)_{i}$ and  $f^\Delta(x)_{i+1} = f^{\Delta'}(x)_{i+1}$ for every possible starting configuration $x\in\zu^n$.
It is easy to see from the Figures~\ref{figure140case11}, \ref{figure140case10}, \ref{figure140case01} and \ref{figure140case00} that the equivalence holds for every possible initial configuration.
The two different update schedules give the same configurations independently from the initial configuration, in other words $D_{f^{(\Delta)}_{140,n}} = D_{f^{(\Delta')}_{140,n}}$.
\end{proof}

\begin{figure}
\centering
 \centerline{\includegraphics[width=\textwidth]{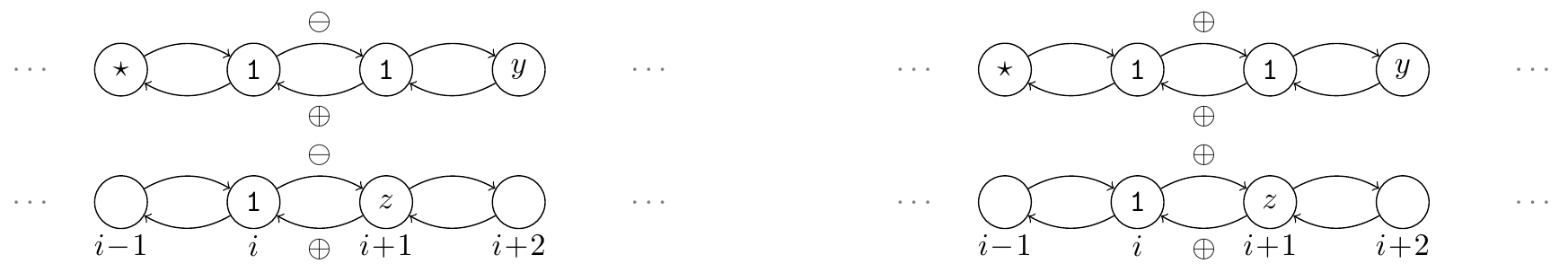}}
	\caption{on the left, starting from a configuration $(\star,\1,\1,y)$ the $\Delta$ update updates before the cell $i$ which becomes $r_{140}=(\star,\1,\1)=\1$, after the rule is applied on the cell $i+1$, where $z=r_{140}(\1,\1,y)$; on the right, the two cells are updated at the same time or after.}
	\label{figure140case11}
\end{figure}

\begin{figure}[htb]
\centering
 \centerline{\includegraphics[width=\textwidth]{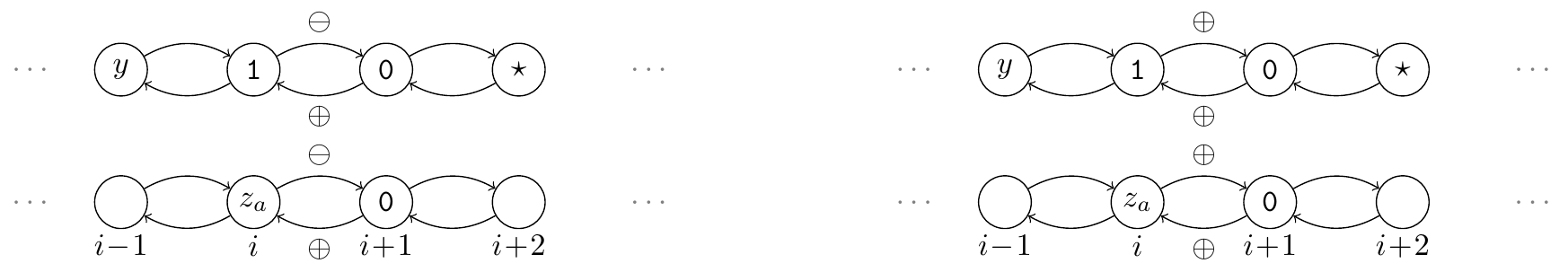}}
 	\caption{on the left, starting from a configuration $(y,\1,\0,\star)$ the $\Delta$ update updates before the cell $i$, $z_a=r_{140}(y,\1,\0)$, after the rule is applied on the cell $i+1$, therefore $r_{140}(r_{140}(y,\1,\0),\0,\star)=\0$; on the right, the two cells are updated at the same time or after. Remember that $r_{140}(\1,\0,\star)=\0$.}
	\label{figure140case10}
\end{figure}

\begin{figure}[htb]
\centering
 \centerline{\includegraphics[width=\textwidth]{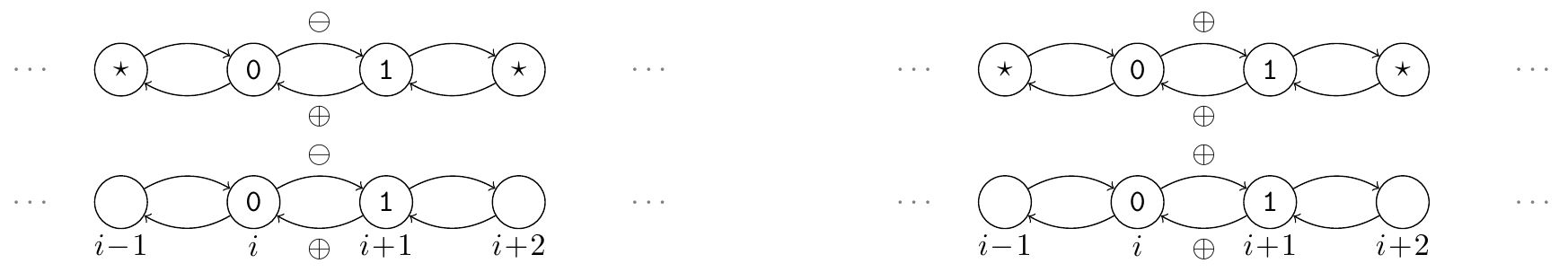}}
 	\caption{on the left, starting from a configuration $(\star,\0,\1,\star)$ the $\Delta$ update updates before the cell $i$ which becomes \0 (in fact $r_{140}(\star,\0,\1)=\0$), after the rule is applied on the cell $i+1$ updated to $r_{140}(\0,\1,\star)=\1$; on the right, the two cells are updated at the same time or after.}
	\label{figure140case01}
\end{figure}

\begin{figure}[htb]
\centering
\centerline{\includegraphics[width=\textwidth]{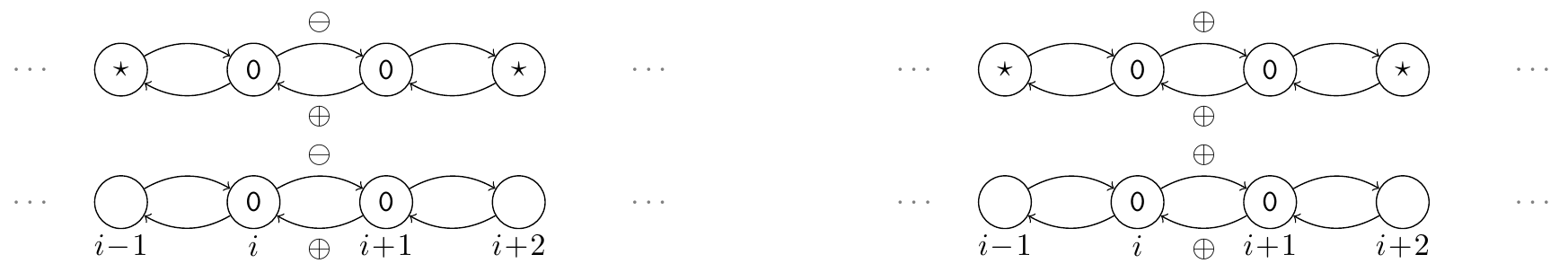}}
	\caption{on the left, starting from a configuration $(\star,\0,\0,\star)$ the $\Delta$ update updates before the cell $i$ which becomes \0, after the rule is applied on the cell $i+1$; on the right, the two cells are updated at the same time or after. In fact, we know that $r_{140}(\star,\0,\0)=\0$ and $r_{140}(\0,\0,\star)=\0$.}
	\label{figure140case00}
\end{figure}

\begin{remark}
	\label{rem:zerorule140}
	The ECA rule $140$ is such that $r_{140}(x_1,\0,x_2)=\0$ for any $x_1,x_2 \in \bool$, hence for any given update schedule $\Delta$ a cell that is in state \0 will remain in such a state throughout the whole evolution.
\end{remark}

\begin{lemma}
\label{noteq140}
For any $n>3$, given two update schedules $\Delta, \Delta'\in\Pn$, it holds
\[
D_{f^{(\Delta)}_{140,n}} \neq D_{f^{(\Delta')}_{140,n}} \iff \exists i\in \intz{n}\text{ s. t. } \lab_{\Delta}((i+1,i))=\labminus \text{ and }
\lab_{\Delta'}((i+1,i))=\labplus\enspace.
\]

\end{lemma}

\begin{proof}
Choose $n, \Delta$ and $\Delta'$ as in the hypothesis. We are going to prove that there exists a configuration such that 
$f^{\Delta}(x)_{i} \neq f^{\Delta'}(x)_{i}$. 
Consider the following initial configuration $(x_{i-1},x_i, x_{i+1},x_{i+2})=(\1,\1,\1,\0)$ and assume that $\lab_{\Delta}((i-1,i))= \oplus$ 
(according to Lemma~\ref{eq140}, this is not changing the dynamics).
Moreover, assume that $\lab_{\Delta}((i+1,i))\neq\lab_{\Delta'}((i+1,i))$ is the only difference between the two update schedules.
According to $\Delta'$, $i$ and $i+1$ are updated together, therefore the final configuration is $(\1,\1,\0,\0)$.
In the case of $\Delta$, the cell $i+1$ is updated before than $i$ holding $r_{140}(\1,\1,\0)=\0$. In a second moment, the cell $i$ is updated and $r_{140}(\1,\1,\0)=\0$. It follows that $f^{\Delta}(x)_{i} \neq f^{\Delta'}(x)_{i}$ and $D_{f^{(\Delta)}_{140,n}} \neq D_{f^{(\Delta')}_{140,n}}$.
Remark that the cell $i+1$ can be influenced from a $\ominus$ chain, but a cell with value $\0$ is frozen at this state. 
\end{proof}

The previous lemma is sufficient to determine that two update schedules that differ 
in at least one cell $i$ such that $\lab_{\Delta}((i+1,i))=\labminus$ and 
$\lab_{\Delta'}((i+1,i))=\labplus$ generate two different dynamics. Indeed, one 
can focus on one of these cells to build a configuration in
which the cell updated produces different values.

\begin{theorem}
  \label{theorem:28-44-140}
  $\sens{f_{\alpha,n}}=\frac{2^n-1}{3^n-2^{n+1}+2}$
  for any $n>3$ and for all ECA rules $\alpha \in \set{28,32, 44,140}$. 
\end{theorem}
\begin{proof}
	Given a configuration of length $n>3$, the patterns in Lemma~\ref{pattern32} (ECA rule $32$) (resp., Lemma~\ref{noteq44} for ECA rule $44$ and Lemma~\ref{noteq140} for ECA rule $140$) may be present in $k$ cells out of $n$ with $ 1\leq k \leq n$ (it must be present in at least one cell because otherwise we would have a $\labminus$ cycle).
	Therefore, there are $\sum_{k=1}^n \binom{n}{i}=2^n -1$ different dynamics.  
\end{proof}
%%%%%%%%%%%%%%%%%%%%%%%%%%%%%%%%

%%%%%%%%%%%%%%%%%%%%%%%%%%%%%%%%
% 8
%\subsection{\TODO{(to be organized)} ECA rule 8}
\subsection{Class III: Medium sensitivity rules}
\label{s:8}
%\TODO{move this section so that it is the last section}

This subsection is concerned uniquely with ECA Rule $8$ which is based on the following
Boolean function $r_8(x_1,x_2,x_3)=\neg x_1 \wedge x_2 \wedge x_3$.
As we will see, finding the expression of sensitivity function for this rule is somewhat peculiar
and require to develop specific techniques.
The sensitivity function obtained tends to $\frac{1+\phi}{3}$, where $\phi$ is the golden ratio.

\begin{remark}
  \label{rem:8-0}
  For any $x_1,x_3 \in \bool$, it holds $r_8(x_1,\0,x_3)=\0$. Hence, for any
  update schedule a cell that is in state \0 will remain in state \0 forever.
\end{remark}

We will first see in Lemma~\ref{lemma:counterclockwise} that as soon as two update
schedules differ on the labeling of an arc $(i,i-1)$, then the two dynamics are
different.
Then, given two update schedules $\Delta,\Delta'$ such that
$\lab_{\Delta}((i,i-1))=\lab_{\Delta'}((i,i-1))$ for all $i \in \intz{n}$,
Lemmas~\ref{lemma:8eq} and~\ref{lemma:8neq} will respectively
give sufficient and necessary conditions for the equality of the two dynamics.

\begin{lemma}
  \label{lemma:counterclockwise}
  Consider two update schedules $\Delta, \Delta'\in\Pn$ for $n\geq 3$.
  If there exists $i \in \intz{n}$ such that $\lab_{\Delta}((i,i-1)) \neq \lab_{\Delta'}((i,i-1))$, then
  $D_{f^{(\Delta)}_{8,n}} \neq D_{f^{(\Delta')}_{8,n}}$.
\end{lemma}

\begin{proof}
  Choose $n\geq 3$ and fix some $i\in\intz{n}$.
  WLOG, assume that $\lab_\Delta((i,i-1))=\labplus$ and
  $\lab_{\Delta'}((i,i-1))=\labminus$ and take $x\in\{\0,\1\}^n$ such that
  ${(x_{i-2},x_{i-1},x_i)=(\0,\1,\1)}$.
  Cell $i-2$ will not change its state, hence when it is time for cell $i-1$ to
  be updated it will have a $\0$ at its left (cell $i-2$) in both cases.
  For $\Delta$, when cell $i-1$ is to be updated, its neighborhood will
  be $(\0,\1,\1)$, and its state will become $\1$ after the iteration.
  As for $\Delta'$,
  when cell $i$ is to be updated, cell $i-1$ is still in state $\1$,
  therefore its state will become $\0$ and when its time for cell $i-1$ to be
  updated, it will have a $\0$ at its right (cell $i$) and its state will become
  $\0$ after the iteration.
  We conclude that $f_{8,n}^{(\Delta)}(x)_{i-1} \neq f_{8,n}^{(\Delta')}(x)_{i-1}$
  and the result follows.
\end{proof}

Now consider two update schedules $\Delta,\Delta'$ whose labelings are equal on
all \emph{counter-clockwise} arcs (\ie of the form $(i,i-1)$).
Lemma~\ref{lemma:8eq} states that, if $\Delta$ and $\Delta'$ differ
only on one arc $(i-1,i)$ such that
$\lab_\Delta((i+1,i))=\lab_{\Delta'}((i+1,i))=\labminus$, then the two dynamics are
identical. By transitivity, if there are more differences but only on arcs of
this form, then the dynamics are also identical.

\begin{lemma}
  \label{lemma:8eq}
  Suppose $\Delta$ and $\Delta'$ are two update schedules over a configuration
  of length $n\geq 3$ and there is $i \in \intz{n}$ such that
  \begin{itemize}
    \item $\lab_\Delta((i+1,i)) = \lab_{\Delta'}((i+1,i))=\labminus$;
    \item $\lab_\Delta((i-1,i)) \neq \lab_{\Delta'}((i-1,i))$;
    \item $\lab_\Delta((j_1,j_2)) = \lab_{\Delta'}((j_1,j_2))$, for all $(j_1,j_2)\neq (i-1,i)$.
  \end{itemize}
  Then $D_{f^{(\Delta)}_{8,n}} = D_{f^{(\Delta')}_{8,n}}$.
\end{lemma}

\begin{proof}
  Fix $n\geq 3$ and choose $i\in\intz{n}$
  WLOG suppose that $\lab_\Delta((i-1,i))=\labplus$
  and $\lab_{\Delta'}((i-1,i))=\labminus$.
  By Theorem~\ref{theorem:lab_valid} and the fact that
  $\lab_{\Delta}((i+1,i))=\lab_{\Delta'}((i+1,i))=\labminus$,
  it follows that $\lab_\Delta((i,i+1))=\lab_{\Delta'}((i,i+1))=\labplus$, 
  otherwise a forbidden cycle of length two is created.
  See Figure~\ref{fig:8eq} for an illustration of the setting.

  \begin{figure}
    \centerline{\includegraphics{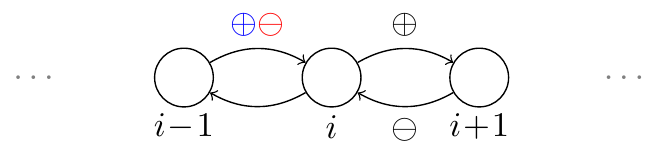}}
    \caption{
      illustration of $\lab_\Delta$ in blue/black and $\lab_{\Delta'}$ in red/black,
      in Lemma~\ref{lemma:8eq}.
      All other labels are equal (the label of arc $(i+1,i)$ is
      $\labminus$ in both update schedules by hypothesis).
    }
    \label{fig:8eq}
  \end{figure}

  The two update schedules $\Delta$ and $\Delta'$ are very similar.   Indeed,
  for any cell $j \in \intz{n} \setminus \{i\}$ the chain of influences are
  identical, \ie $\dleft{\Delta}(j)=\dleft{\Delta'}(j)$ and 
  $\dright{\Delta}(j)=\dright{\Delta'}(j)$. We deduce from
  Lemma~\ref{lemma:di} that for any configuration $x \in \bool^n$ and
  any $j \neq i$ the images under update schedules $\Delta$ and $\Delta'$,
  \ie $f_{8,n}^{(\Delta)}(x)_j=f_{8,n}^{(\Delta')}(x)_j$. As a consequence,
  it only remains to consider cell $i$.
  Let $x \in \bool^n$ be any configuration % of length $n \geq 3$
  (if $n \leq 2$ then $i-1=i+1$, but $\lab_{\Delta}((i-1,i))=\labplus$ whereas
  $\lab_{\Delta}((i+1,i))=\labminus$).
  
  By Remark \ref{rem:8-0}, 
  if $x_i=\0$, then $f_{8,n}^{(\Delta)}(x)_i=f_{8,n}^{(\Delta')}(x)_i=\0$.
  Now suppose $x_i=\1$. Since
  $\lab_\Delta((i,i+1))=\lab_{\Delta'}((i,i+1))=\labplus$, by the time cell
  $i+1$ is updated, there is a $\1$ at its left (cell $i$) in both cases, hence
  $f_{8,n}^{(\Delta)}(x)_{i+1}=f_{8,n}^{(\Delta')}(x)_{i+1}=\0$.
  Then, when cell $i$ is updated in both cases, there will be a $\0$ at its right
  (cell $i+1$), therefore
  $f_{8,n}^{(\Delta)}(x)_{i}=f_{8,n}^{(\Delta')}(x)_{i}=\0$.

  We conclude that for all $x \in \bool^n$ and all $j \in \intz{n}$
  we have $f_{8,n}^{(\Delta)}(x)_{j}=f_{8,n}^{(\Delta')}(x)_{j}$,
  \ie $D_{f^{(\Delta)}_{8,n}} = D_{f^{(\Delta')}_{8,n}}$.
\end{proof}

Lemma~\ref{lemma:8neq} states that, as soon as $\Delta$ and $\Delta'$
differ on arcs of the form $(i-1,i)$ such that 
$\lab_\Delta((i+1,i))=\lab_{\Delta'}((i+1,i))=\labplus$, then the two dynamics are
different (remark that in this case we must have
$\lab_\Delta((i,i-1))=\lab_{\Delta'}((i,i-1))=\labplus$ otherwise one of $\Delta$ or
$\Delta'$ has an invalid cycle of length two between the nodes $i-1$ and $i$). This lemma can 
be applied if at least one cell of the configuration contains the pattern.

\begin{lemma}
  \label{lemma:8neq}
  For $n\geq 5$, consider two update schedules $\Delta, \Delta'\in\Pn$.
  If there exists (at least one cell) $i \in \intz{n}$ such that
  \begin{itemize}
    \item $\lab_\Delta((i+1,i))=\lab_{\Delta'}((i+1,i))=\labplus$;
    \item $\lab_\Delta((i-1,i))\neq \lab_{\Delta'}((i-1,i))$;
    \item $\lab_\Delta((j,j-1))=\lab_{\Delta'}((j,j-1))$, for all $j \in \intz{n}$;
  \end{itemize}
  then $D_{f^{(\Delta)}_{8,n}} \neq D_{f^{(\Delta')}_{8,n}}$.
\end{lemma}

\begin{proof}
  Choose $n, \Delta$ and $\Delta'$ as in the hypothesis.
  WLOG, assume that $\lab_\Delta((i-1,i))=\labplus$
  and $\lab_{\Delta'}((i-1,i))=\labminus$ for $i\in\intz{n}$.
  %It follows from
  %Theorem~\ref{theorem:lab_valid} and the equality of $\Delta$ and
  %$\Delta'$ on arcs of the form $(j,j-1)$ that
  %$\lab_{\Delta}((i,i-1))=\lab_{\Delta'}((i,i-1))=\labplus$
  %otherwise a forbidden cycle of length two is created.
  See Figure~\ref{fig:8neq} for an illustration of the setting.
  We are going to construct a configuration $x \in \bool^n$ such 
  that $f_{8,n}^{(\Delta)}(x)_i=\0$ whereas $f_{8,n}^{(\Delta')}(x)_i=\1$,
  \ie such that the two dynamics differ in the image of cell $i$.

  \begin{figure}
    \centerline{\includegraphics{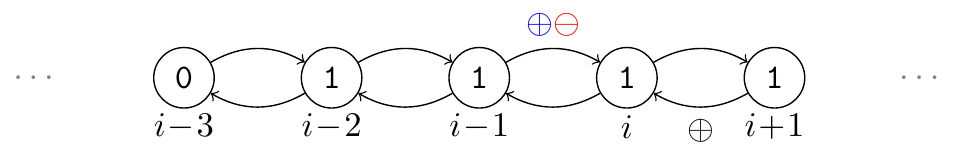}}
    \caption{
      illustration of $\lab_\Delta$ in blue/black, $\lab_{\Delta'}$ in
      red/black, in Lemma~\ref{lemma:8neq}.
      Other labels on arcs of the form $(j-1,j)$ are {\em a priori} unknown
      (they may be equal or different in $\Delta$ and $\Delta'$),
      however labels on arcs of the form $(j,j-1)$ are equal by hypothesis.
      States inside the nodes correspond to configuration $x$
      such that the image of cell $i$ under update schedule $\Delta$ is $\0$,
      whereas under update schedule $\Delta'$ it is $\1$.
    }
    \label{fig:8neq}
  \end{figure}
  
  The construction of $x \in \bool^n$ %with $n \geq 5$ 
  only requires
  to set the pattern 
  $(x_{i-3},x_{i-2},x_{i-1},x_i,x_{i+1})=(\0,\1,\1,\1,\1)$.
  Regarding $\Delta$, from the $\labplus$ labels of arcs $(i-1,i)$
  and $(i+1,i)$ we have
  $f_{8,n}^{(\Delta)}(x)_i = r_8(x_{i-1},x_i,x_{i+1}) = r_8(\1,\1,\1)=\0$.
  Regarding $\Delta'$, let us deduce by denoting $y$ the image of $x$
  (\ie $y_i = f_{8,n}^{(\Delta')}(x)_i$) that whatever the value of $\dleft{\Delta'}(i)$
  we have $f_{8,n}^{(\Delta')}(x)_i=\1$ (\ie $y_i = \1$).
  \begin{itemize}
    \item If $\dleft{\Delta'}(i) = 2$ then cell $i-1$ is updated and then cell $i$,
      \begin{itemize}
        \item $y_{i-1} = r_8(x_{i-2},x_{i-1},x_{i}) = r_8(\1,\1,\1) = \0$,
        \item $y_i = r_8(y_{i-1},x_i,x_{i+1}) = r_8(\0,\1,\1) = \1$.
      \end{itemize}
    \item if $\dleft{\Delta'}(i) = 3$ then cell $i-2$ is updated then cell $i-1$
      and then cell $i$,
      \begin{itemize}
        \item $y_{i-2} = r_8(x_{i-3},x_{i-2},x_{i-1}) = r_8(\0,\1,\1) = \1$,
        \item $y_{i-1} = r_8(y_{i-2},x_{i-1},x_{i}) = r_8(\1,\1,\1) = \0$,
        \item $y_i = r_8(y_{i-1},x_i,x_{i+1}) = r_8(\0,\1,\1) = \1$.
      \end{itemize}
    \item if $\dleft{\Delta'}(i) \geq 4$ then cell $i-3$ is updated then cell $i-2$
      then cell $i-1$ and then cell $i$,
      \begin{itemize}
        \item $y_{i-3} = \0$ by Remark~\ref{rem:8-0} since $x_{i-3} = \0$,
        \item $y_{i-2} = r_8(y_{i-3},x_{i-2},x_{i-1}) = r_8(\0,\1,\1) = \1$,
        \item $y_{i-1} = r_8(y_{i-2},x_{i-1},x_{i}) = r_8(\1,\1,\1) = \0$,
        \item $y_i = r_8(y_{i-1},x_i,x_{i+1}) = r_8(\0,\1,\1) = \1$.
      \end{itemize}
  \end{itemize}
  Remark that $n \geq 5$ is required by the consideration of cells $i-3$ to $i+1$
  in the third case.
\end{proof}

Lemmas~\ref{lemma:counterclockwise}, \ref{lemma:8eq}
and~\ref{lemma:8neq} characterize completely for rule $8$ the cases
when two update schedules $\Delta,\Delta'$ lead to
\begin{itemize}
  \item the same dynamics,
    \ie $\dynamics{f^{(\Delta)}_{8,n}} = \dynamics{f^{(\Delta')}_{8,n}}$, or
  \item different dynamics,
    \ie $\dynamics{f^{(\Delta)}_{8,n}} \neq \dynamics{f^{(\Delta')}_{8,n}}$.
\end{itemize}
Indeed, Lemma~\ref{lemma:counterclockwise} shows that counting
$|\dynamics{f_{8,n}}|$ can be partitioned according to the word given
by $\lab_\Delta((i,i-1))$ for $i \in \intz{n}$, and then for each labeling of
the $n$ arcs of the form $(i,i-1)$, Lemmas~\ref{lemma:8eq}
and~\ref{lemma:8neq} provide a way of counting the number of dynamics.
We first give an example of application, and then the general counting result
establishing a relation to the bisection of Lucas numbers.

\begin{example}
  \label{example:8}
  Consider the set of non-equivalent update schedules
  $\Delta \in \Px{10}$ such that
  $$\big(\lab_\Delta((i,i-1))\big)_{i \in \intz{10}}=(\labplus,\labminus,\labplus,\labplus,\labplus,\labminus,\labplus,\labminus,\labminus,\labplus).$$
  We have the following disjunction for $i \in \intz{n}$
  (see Figure~\ref{fig:8example}):
  \begin{itemize}
    \item[\textbf{A-}] if $\lab((i,i-1))=\labminus$ then
      $\lab((i-1,i))=\labplus$ according to Theorem~\ref{theorem:lab_valid},
      else
    \begin{itemize}
      \item[\textbf{B-}] if $\lab((i+1,i))=\labminus$ then $\lab((i-1,i))$ does
        not change the dynamics according to Lemma~\ref{lemma:8eq},
      \item[\textbf{C-}] if $\lab((i+1,i))=\labplus$ then the two possibilities
        for $\lab((i-1,i))$ each lead to different dynamics according to
        Lemma~\ref{lemma:8neq}.
    \end{itemize}
  \end{itemize}
  Therefore, on overall, there are $2^{3}=8$ different dynamics for such update schedules.
\end{example}

\begin{figure}
  \centerline{\includegraphics{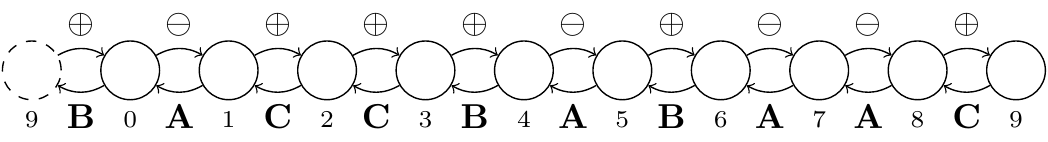}}
  \caption{
     counting the number of different dynamics for ECA rule $8$
     when the labeling of arcs $(i-1,i)$ for $i \in\intz{10}$ is
     $(\labplus,\labminus,\labplus,\labplus,\labplus,\labminus,\labplus,\labminus,\labminus,\labplus)$.
%    
%    Example for rule $8$ of counting the number of different dynamics for
%    update schedules such that the labeling of arcs $(i-1,i)$ for $i \in
%    \intz{10}$ is
%    $(\labplus,\labminus,\labplus,\labplus,\labplus,\labminus,\labplus,\labminus,\labminus,\labplus)$.
    Labels \textbf{A} is enforced to be $\labplus$ by
    Theorem~\ref{theorem:lab_valid}, labels \textbf{B} have no influence
    according to Lemma~\ref{lemma:8eq}, and any combination of labels
    \textbf{C} gives a different dynamics according to
    Lemma~\ref{lemma:8neq}.
  }
  \label{fig:8example}
\end{figure}

\begin{theorem}
  \label{theorem:8}
  $\sens{f_{8,n}}=\frac{\phi^{2n}+\phi^{-2n}-2^n}{3^n-2^{n+1}+2}$ for any $n \geq 5$,
  with $\phi=\frac{1+\sqrt{5}}{2}$ the golden ratio.
\end{theorem}

\begin{proof}
  According to Lemma~\ref{lemma:counterclockwise},
  the set $\dynamics{f_{8,n}}$ can be partitioned as
  $$\dynamicsu{f_{8,n}}{u} = \big\{ D_{f^{(\Delta)}} \mid
      \Delta \in \Pn
      \text{ and }
      \big( \lab_\Delta((i,i-1)) \big)_{i \in \intz{n}} = u
    \big\} \text{ for } u \in \{\labplus,\labminus\}^n.$$
  That is, in $\dynamicsu{f_{8,n}}{u}$ the labels of arcs of the form
  $(i,i-1)$ for $i \in \intz{n}$ are fixed according to some word
  $u \in \{\labplus,\labminus\}^n$.
  Therefore we have
  $$
    |\dynamics{f_{8,n}}| = \sum_{u \in \{\labplus,\labminus\}^n} |\dynamicsu{f_{8,n}}{u}|.
  $$
  Then, given some word $u \in \{\labplus,\labminus\}^n \setminus
  \{\labplus^n,\labminus^n\}$, according to Theorem~\ref{theorem:lab_valid} and
  Lemmas~\ref{lemma:8eq} and~\ref{lemma:8neq} we have (see
  Example~\ref{example:8} for details)
  $$
    |\dynamicsu{f_{8,n}}{u}| = 2^{\countplus{u}-\countplusminus{u}}
  $$
  where $\countplus{u}$ is the number of $\labplus$ in word $u$,
  and $\countplusminus{u}$ is the number of $\labplus\labminus$ factors
  in word $u$ considered periodically,
  \ie $\countplusminus{u}=|\{ i \in \intz{n} \mid u_i=\labplus \text{ and } u_{i+1}=\labminus \}|$.
  
  According to Theorem~\ref{theorem:lab_valid},
  the cases $u \in \{\labplus^n,\labminus^n\}$ are particular. Indeed, 
  for any $n$:
  \begin{itemize}
    \item all labels $\labminus$ (\ie, $u=\labminus^n$) is an invalid cycle
      hence $|\dynamicsu{f_{8,n}}{\labminus^n}|=0$,
    \item all labels $\labplus$ (\ie, $u=\labplus^n$) forces the labels of
      all arcs of the form $(i-1,i)$ for $i \in \intz{n}$ to be also labeled
      $\labplus$ otherwise a forbidden cycle is created, hence
      $|\dynamicsu{f_{8,n}}{\labminus^n}|=1$.
  \end{itemize}
  Given that $2^{\countplus{\labminus^n}-\countplusminus{\labminus^n}}=2^0=1$
  (instead of $0$) and $2^{\countplus{\labplus^n}-\countplusminus{\labplus^n}}=2^n$
  (instead of $1$), we deduce that
  \begin{equation}
    \label{eq:8}
    |\dynamics{f_{8,n}}| = \left( \sum_{u \in \{\labplus,\labminus\}^n} 2^{\countplus{u}-\countplusminus{u}} \right) - 2^n.
  \end{equation}
  
  In order to study the summation term in Equation~\ref{eq:8}, let us 
  denote it $S(n)$. 
  We will consider recurrence relations according to the following
  partition of the set $\set{\labplus,\labminus}^n$:
  for $\sigma,\sigma' \in \set{\labplus,\labminus}$ let
  $\partplusminus{\sigma}{\sigma'}(n)$ be the set of words beginning with label
  $\sigma$ and ending with label $\sigma'$, \ie
  $\partplusminus{\sigma}{\sigma'}(n)=\{ u=u_0\ldots u_{n-1} \in
  \set{\labplus,\labminus}^n \mid u_0=\sigma \text{ and } u_{n-1}=\sigma' \}$.
  Denoting
  $$
    \sumplusminus{\sigma}{\sigma'}(n)=
    \sum_{u \in \partplusminus{\sigma}{\sigma'}}
    2^{\countplus{u}-\countplusminus{u}}
  $$
  the recurrence relations are, for all $n \geq 1$
  (although Equation~\ref{eq:8} holds only for $n \geq 5$,
  the value starting from which Lemmas~\ref{lemma:counterclockwise},
  \ref{lemma:8eq} and~\ref{lemma:8neq} hold),
  \begin{itemize}
    \item $\sumplusminus{\labplus}{\labplus}(n+1) =
        2 \sumplusminus{\labplus}{\labplus}(n)
      + 2 \sumplusminus{\labplus}{\labminus}(n)$,
    \item $\sumplusminus{\labplus}{\labminus}(n+1) =
      \frac{1}{2} \sumplusminus{\labplus}{\labplus}(n)
      + \sumplusminus{\labplus}{\labminus}(n)$,
    \item $\sumplusminus{\labminus}{\labplus}(n+1) =
        2 \sumplusminus{\labminus}{\labplus}(n)
      + \sumplusminus{\labminus}{\labminus}(n)$,
    \item $\sumplusminus{\labminus}{\labminus}(n+1) =
        \sumplusminus{\labminus}{\labplus}(n)
      + \sumplusminus{\labminus}{\labminus}(n)$,
  \end{itemize}
  and we have $S(n)=\sumplusminus{\labplus}{\labplus}(n)
  +\sumplusminus{\labplus}{\labminus}(n)
  +\sumplusminus{\labminus}{\labplus}(n)
  +\sumplusminus{\labminus}{\labminus}(n)$.
  Indeed, for example regarding $\sumplusminus{\labplus}{\labplus}(n+1)$,
  consider a word $u=u_0\ldots u_{n-1} \in \set{\labplus,\labminus}^n$
  and the concatenation of a label $\sigma \in \set{\labplus,\labminus}$
  at the end of $u$, then $u'=u_0\ldots u_{n-1}\sigma) \in 
  \partplusminus{\labplus}{\labplus}(n+1)$ if and only if $\sigma=\labplus$
  and $u_0=\labplus$, \ie $\sigma=\labplus$ and
  ($u \in \partplusminus{\labplus}{\labplus}(n)$ or 
  $u \in \partplusminus{\labplus}{\labminus}(n)$).
  It follows that,
  \begin{itemize}
    \item if $u \in \partplusminus{\labplus}{\labplus}(n)$ then
      $\countplus{u'}=\countplus{u}+1$ and
      $\countplusminus{u'}=\countplusminus{u}$,
    \item if $u \in \partplusminus{\labplus}{\labminus}(n)$ then
      $\countplus{u'}=\countplus{u}+1$ and
      $\countplusminus{u'}=\countplusminus{u}$,
  \end{itemize}
  which gives the first recurrence.
  A similar reasoning lead to the three other recurrence relations.
  Also remark that by symmetry we always have
  $\sumplusminus{\labplus}{\labminus}(n)=\sumplusminus{\labminus}{\labplus}(n)$,
  though this fact will not be used in the coming proof.

  In order to solve the recurrence, we establish a relation to known formulas
  by remarking that $S(n)=3S(n-1)-S(n-2)$, which corresponds to the bisection
  of Fibonacci-like integer sequences ({\em aka} Lucas sequences):
  \begin{align*}
    S(n) &= \sumplusminus{\labplus}{\labplus}(n)
          + \sumplusminus{\labplus}{\labminus}(n)
          + \sumplusminus{\labminus}{\labplus}(n)
          + \sumplusminus{\labminus}{\labminus}(n)\\
         &= 2 \sumplusminus{\labplus}{\labplus}(n-1)
          + 2 \sumplusminus{\labplus}{\labminus}(n-1)
          + \frac{1}{2} \sumplusminus{\labplus}{\labplus}(n-1)
          + \sumplusminus{\labplus}{\labminus}(n-1)\\
          &\quad
          + 2 \sumplusminus{\labminus}{\labplus}(n-1)
          + \sumplusminus{\labminus}{\labminus}(n-1)
          + \sumplusminus{\labminus}{\labplus}(n-1)
          + \sumplusminus{\labminus}{\labminus}(n-1)\\[.5em]
         &= 3 \sumplusminus{\labplus}{\labplus}(n-1)
          + 3 \sumplusminus{\labplus}{\labminus}(n-1)
          + 3 \sumplusminus{\labminus}{\labplus}(n-1)
          + 3 \sumplusminus{\labminus}{\labminus}(n-1)\\
          &\quad
          - \frac{1}{2} \sumplusminus{\labplus}{\labplus}(n-1)
          - \sumplusminus{\labminus}{\labminus}(n-1)\\
         &= 3 S(n-1)
          - \sumplusminus{\labplus}{\labplus}(n-2)
          - \sumplusminus{\labplus}{\labminus}(n-2)
          - \sumplusminus{\labminus}{\labplus}(n-2)
          - \sumplusminus{\labminus}{\labminus}(n-2)\\
         &= 3 S(n-1) - S(n-2)
  \end{align*}
  Finally, since we have $S(1)=2$ and $S(2)=3$ we deduce that $S(n)$ is the
  bisection of Lucas numbers, sequence \texttt{A005248} of
  OEIS~\cite{oeisA005248}.
  The nice closed form involving the golden ratio is a folklore adaptation
  of Binet's formula to Lucas numbers,
  and Equation~\ref{eq:8} gives the result.
\end{proof}

\subsection{Class IV: Almost max-sensitive rules}
%%%%%%%%%%%%%%%%%%%%%%%%%%%%%%%%
% 128, 160, 162
%\subsection{\TODO{(to be organized)} ECA rules 128, 160, 162}
\label{s:128_160_162}

This last class contains three ECA rules, namely $128$, $160$ and $162$,
for which the sensitivity function tends to $1$.
The study of sensitivity to synchronism for these rules
is based on the characterization of pairs of update schedule leading to the
same dynamics.
A pair of update schedules $\Delta, \Delta'\in\Pn$ is \emph{special for rule $\alpha$}
if $\Delta\not\equiv\Delta'$
but $D_{f_{\alpha,n}^{(\Delta)}} = D_{f_{\alpha,n}^{(\Delta')}}$.
We will count the special pairs for rules $128$, $160$ and $162$.

Given an update schedule $\Delta\in\Pn$,
define the \emph{left rotation} 
$\sigma(\Delta)$ and the \emph{left/right exchange} $\rho(\Delta)$
,such that, $\forall i\in\intz{n}$ it holds that $\lab_{\sigma(\Delta)}((i,j))=\lab_{\Delta}((i+1,j+1))$ and
$\lab_{\rho(\Delta)}((i,j))=\lab_{\Delta}((j,i))$.
It is clear that if a pair of update schedules $\Delta, \Delta'\in\Pn$ is
special then $\sigma(\Delta), \sigma(\Delta')$
is also special.
Furthermore, when rule $\alpha$ is left/right symmetric
(meaning that $\forall x_1,x_2,x_3 \in \bool$ we have
$r_\alpha(x_1,x_2,x_3)=r_\alpha(x_3,x_2,x_1)$,
which is the case of rules $128$ and $162$, but not $160$)
then $\rho(\Delta),\rho(\Delta')$ is also special.
We say that special pairs in a set $S$ are {\em disjoint} when no update schedule belongs to
more than one pair \ie, if three update schedules $\Delta, \Delta', \Delta''\in S$ are such that both $(\Delta, \Delta')$ and $(\Delta, \Delta'')$
are special pairs then $\Delta'=\Delta''$.
When it is clear from the context,
we will omit to mention the rule relative to which some pairs are special.

% 128
\subsubsection{ECA rule 128}

The Boolean function associated with the ECA rule $128$ is 
$r_{128}(x_1,x_2,x_3)=x_1 \wedge x_2 \wedge x_3$. Its simple
definition will allow us to better illustrate the role played by special pairs.

\begin{remark}\label{rem:dset-eq-intzn}
When $\dset{\Delta}(i) = \intz{n}$ for some cell $i$, %in the case of rule $128$
the only possibility to get $f_{128}^{(\Delta)}(x)_i=\1$ is $x=\1^n$. However
for $x=\1^n$ we have $f_{128}^{(\Delta)}(x)_i=\1$ for any $\Delta$. 
\end{remark}

The previous remark combined with an observation in the spirit of 
Lemma~\ref{lemma:dset}, gives the next characterization.
Let us introduce the notation $\dset{\Delta}=\dset{\Delta'}$ for cases in which 
$\dset{\Delta}(i)=\dset{\Delta'}(i)$ holds in every cell $i \in \intz{n}$.

\begin{lemma}
  \label{lemma:128d}
  For any $n\in\N$, choose $\Delta, \Delta'\in\Pn$ such that
  $\Delta \not\equiv \Delta'$. Then,
  $\dset{\Delta}=\dset{\Delta'}$
  if and only if 
  $D_{f^{(\Delta)}_{128,n}}=D_{f^{(\Delta')}_{128,n}}$.
\end{lemma}

\begin{proof}
  For $n=1,2,3$ we have $\dset{\Delta}(i)=\dset{\Delta'}(i)=\intz{n}$
  for all $i \in \intz{n}$, and only one dynamics (see Section~\ref{s:exp}),
  therefore the result holds. Now, consider $n \geq 4$.
\begin{itemize}
\item[$(\Rightarrow)$] 
Assume $\dset{\Delta}=\dset{\Delta'}$.
%  \noindent$(\Rightarrow)$ Assume $\dset{\Delta}=\dset{\Delta'}$.
  Given $i\in\intz{n}$, two cases are possible:
  \begin{itemize}
    \item $\dset{\Delta}(i) \neq \intz{n}$. In this case,
      one can deduce from $\dset{\Delta}(i)=\dset{\Delta'}(i)$ that
      $\dright{\Delta}(i)=\dright{\Delta'}(i)$ and
      $\dleft{\Delta}(i)=\dleft{\Delta'}(i)$ (see the proof of
      Lemma~\ref{lemma:dset}). From Formula~\ref{eq:d} it follows
      $f^{(\Delta)}_{128}(x)_i=f^{(\Delta')}_{128}(x)_i$ for any $x \in \bool^n$.
    \item %For all $i \in \intz{n}$ such that
      $\dset{\Delta}(i)=\dset{\Delta'}(i)=\intz{n}$. In this case, in order to have
      $f^{(\Delta)}_{128}(x)_i \neq f^{(\Delta')}_{128}(x)_i$ one must have one
      of them equal to $\1$ and the other equal to $\0$. WLOG 
      assume $f^{(\Delta)}_{128}(x)_i=\1$. From the
      definition of the ECA rule $128$ and Formula~\ref{eq:d}, since
      $\dset{\Delta}(i)=\intz{n}$ the only possibility is $x=\1^n$, but this
      also implies $f^{(\Delta')}_{128}(x)_i=\1$.
%    \item[\ ]We conclude that $f^{(\Delta)}_{128}(x)_i=f^{(\Delta')}_{128}(x)_i$ for any
%  $x \in \bool^n$ and $i \in \intz{n}$, which is equivalently formulated as
%  $D(f^{(\Delta)}_{128,n})=D(f^{(\Delta')}_{128,n})$.
  \end{itemize}
  We conclude that $f^{(\Delta)}_{128}(x)_i=f^{(\Delta')}_{128}(x)_i$ for any
  $x \in \bool^n$ and $i \in \intz{n}$, which is equivalently formulated as
  $D_{f^{(\Delta)}_{128,n}}=D_{f^{(\Delta')}_{128,n}}$.
\end{itemize} %  
\item[$(\Leftarrow)$]
  Assume $\dset{\Delta}(i) \neq \dset{\Delta'}(i)$ for some $i \in \intz{n}$.
  Then, WLOG there exists $j \in \intz{n}$
  such that $j \in \dset{\Delta}(i) \setminus \dset{\Delta'}(i)$.
  The configuration $x$ with $x_i=\0 \iff i=j$ gives,
  again by Formula~\ref{eq:d}, that
  $f^{(\Delta)}_{128}(x)_i \neq f^{(\Delta')}_{128}(x)_i$. Indeed,
  \begin{itemize}
    \item the image of $i$ under the update schedule $\Delta$ depends on
    $x_j=\0$ which, by the definition of the ECA rule $128$, ensures that
    $f^{(\Delta)}_{128}(x)_i = \0$, and
  \item the image of $i$ under the update schedule $\Delta'$ depends only on
    cells in state $\1$ which, by the definition of the ECA rule $128$, ensures that
    $f^{(\Delta')}_{128}(x)_i = \1$.
  \end{itemize}
  Consequently, $D_{f^{(\Delta)}_{128,n}} \neq D_{f^{(\Delta')}_{128,n}}$.
\end{proof}

Lemma~\ref{lemma:128d} characterizes exactly the pairs of non-equivalent
update schedules for which the dynamics of rule $128$ differ, \ie,
the set of special pairs for rule $128$, which are the set pairs
$\Delta,\Delta' \in \Pn$ such that $\Delta \not\equiv \Delta'$
but $\dset{\Delta}=\dset{\Delta'}$.
Computing $\sens{f_{128,n}}$ is now a combinatorial problem
of computing the number of possible $\dset{\Delta}$ for $\Delta \in \Pn$.

\begin{remark}
  Lemma~\ref{lemma:128d} does not hold for all rules, since some of them are
  max-sensitive even though there exist $\Delta \not\equiv \Delta'$
  with $\dset{\Delta}(i)=\dset{\Delta'}(i)$ for all $i \in \intz{n}$.
\end{remark}

We are going to prove that for any $n>6$,
there exist $10n$ disjoint special pairs of schedules of size $n$
(Lemma~\ref{lemma:counting-special-pairs}).
We will first argue that special pairs differ in the labeling of exactly one arc
(Lemma~\ref{lemma:special-onearc}), then exhibit $10n$ special pairs of schedules of
size $n$ (which come down to five cases up to rotation and left/right exchange)
and finally argue that these pairs are disjoint. This will lead to
Theorem~\ref{theorem:128}.
The coming proofs will make heavy use of the following lemma
(see Figure~\ref{fig:special-labminus}).

\begin{lemma}
  \label{lemma:special-labminus}
  For any $n \geq 4$, consider a special pair $\Delta,\Delta' \in \Pn$ for rule $128$
  such that
  $\lab_\Delta((i+1,i))=\labplus$ and $\lab_{\Delta'}((i+1,i))=\labminus$
  for some $i\in\intz{n}$.
  For all $j\in\intz{n}\setminus\set{i,i+1,i+2}$, it holds that
  $\lab_\Delta((j,j+1))=\labminus$ and $\lab_\Delta((j+1,j))=\labplus$.
\end{lemma}

\begin{figure}
  \centerline{\includegraphics{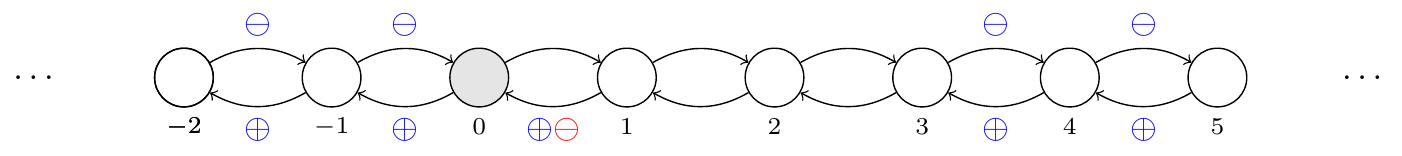}}
  \caption{
    Illustration of Lemma~\ref{lemma:special-labminus}, with $\Delta$ in blue
    and $\Delta'$ in red: hypothesis on the labelings of arc $(1,0)$
    imply many $\labminus$ labels on arcs of the form $(j,j+1)$, and
    $\labplus$ labels on arcs of the form $(j+1,j)$, for $\Delta$.
  }
  \label{fig:special-labminus}
\end{figure}

\begin{proof}
  From Lemma~\ref{lemma:128d}, we must have $\dset{\Delta}=\dset{\Delta'}$.
  Hence, in particular, $\dset{\Delta}(i)=\dset{\Delta'}(i)$. 
  However, from the hypothesis on the labelings of arc $(i+1,i)$, 
  the only possibility is that $\dset{\Delta}(i)=\dset{\Delta'}(i)=\intz{n}$. 
  Indeed, we have $i+2 \in \dset{\Delta'}(i)$, but on $\Delta$ to the right 
  we have $\dright{\Delta}(i)=0$ thus for the chain of influences of cell $i$ 
  to contain cell $i+2$ we must have $\dleft{\Delta}(i) \geq n-2$, which corresponds to
  $\lab_\Delta((j+1,j))=\labminus$ for all $j \in \intz{n}\setminus\set{i,i+1,i+2}$.
  It follows that for these $j$ we have $\lab_\Delta((j,j+1))=\labplus$ otherwise an invalid 
  cycle of length two is created (Theorem~\ref{theorem:lab_valid}).
\end{proof}

\begin{lemma}
  \label{lemma:special-onearc}
  For any $n > 6$, if $\Delta,\Delta' \in \Pn$ is a special pair for rule $128$
  then $\Delta$ and $\Delta'$ differ on the labeling of exactly one arc.
\end{lemma}

\begin{proof}
  First, by definition of special pair, we have $\Delta \not\equiv \Delta'$. 
  Hence, $\Delta$ and $\Delta'$ must differ on the labeling of at least one 
  arc. Up to rotation and right/left exchange, let us suppose WLOG that
  $\lab_\Delta((1,0))=\labplus$ and $\lab_{\Delta'}((1,0))=\labminus$.
  Now, for the sake of contradiction, assume that they also differ on 
  another arc, and consider the following cases disjunction
  (remark that the order of the case study is chosen so that
  cases make reference to previous cases).
  \begin{enumerate}[label=(\alph*)]
    \item If $\lab_\Delta((i,i+1))=\labplus$ and   
      $\lab_{\Delta'}((i,i+1))=\labminus$
      for some $i \in \intz{n}$, then by applying Lemma~\ref{lemma:special-labminus} to
      the two arcs where $\Delta$ and $\Delta'$ differ leads to a contradiction
      on the labeling of some arc according to $\Delta$.
      Indeed, Lemma~\ref{lemma:special-labminus} is applied to two arcs in 
      different directions, one application leaves three arcs of the form 
      $(j,j+1)$ not labeled $\labminus$ in $\Delta$ and three arcs of the form 
      $(j+1,j)$ not labeled $\labplus$ in $\Delta$, the converse for the other 
      application, hence starting from $n=7$ these labelings overlap in a contradictory fashion.
    \item If $\lab_\Delta((i+1,i))=\labminus$ and 
      $\lab_{\Delta'}((i+1,i))=\labplus$
      for some $i \in \intz{n}\setminus\set{0}$, then $i \in \set{2,3}$ otherwise there is a forbidden cycle of length two in $\Delta$ with some
      $\labminus$ label given by the application of
      Lemma~\ref{lemma:special-labminus} to the arc $(1,0)$. 
      However, for $i\in\set{2,3}$ the application of Lemma~\ref{lemma:special-labminus} to the arc $(i+1,i)$ gives $\lab_{\Delta'}((0,1))=\labminus$, creating a
      forbidden cycle of length two in $\Delta'$.
    \item\label{item:special-counterclockwiseplusminus2}
      If $\lab_\Delta((i+1,i))=\labplus$ and $\lab_{\Delta'}((i+1,i))=\labminus$
      for some $i \in \intz{n} \setminus\{1\}$, then applying
      Lemma~\ref{lemma:special-labminus} to the two arcs where $\Delta$ and $\Delta'$
      differ leads to a forbidden cycle of length $n$ in $\Delta$
      (contradiction Theorem~\ref{theorem:lab_valid}).
      Indeed, if $i \notin\set{1,2}$ then we have $\labminus$ labels on arcs
      of the form $(j,j+1)$ for all $j \in \intz{n}$, and if $i=2$ then the
      forbidden cycle contains the arc $(3,2)$ labeled $\labplus$.
      The case $i=0$ is not a second difference.
    \item\label{item:special-clockwiseminusplus2}
      If $\lab_\Delta((i,i+1))=\labminus$ and $\lab_{\Delta'}((i,i+1))=\labplus$
      for some $i \in \intz{n}$, then applying Lemma~\ref{lemma:special-labminus} to
      arc $(1,0)$ gives $\lab_\Delta((j+1,j))=\labplus$ for all $j \in
      \intz{n} \setminus \{0,1,2\}$, and applying
      Lemma~\ref{lemma:special-labminus} to arc $(i,i+1)$ gives
      $\lab_{\Delta'}((j+1,j))=\labminus$ for all $j \in \intz{n} \setminus
      \set{i,i-1,i-2}$. Starting from $n=7$ we have $(\intz{n} \setminus \set{0,1,2}) 
      \cap (\intz{n}\setminus\set{i,i-1,i-2}) \neq \emptyset$, and as a consequence
      there is an arc $((j+1,j))$ in the case of
      Item~\ref{item:special-counterclockwiseplusminus2}.
    \item 
      If $\lab_\Delta((2,1))=\labplus$ and $\lab_{\Delta'}((2,1))=\labminus$,
      then applying Lemma~\ref{lemma:special-labminus} to arc $(2,1)$
      gives $\lab_\Delta((0,1))=\labminus$, however since by hypothesis
      $\lab_{\Delta'}((1,0))=\labminus$ we also have $\lab_{\Delta'}((0,1))=\labplus$
      otherwise there is a forbidden cycle of length two in $\Delta'$
      (Theorem~\ref{theorem:lab_valid}). As a consequence, the arc $(0,1)$ is in
      the case of Item~\ref{item:special-clockwiseminusplus2}.
  \end{enumerate}
  We conclude that in any case a second difference leads to a contradiction,
  either because an invalid cycle is created, or because repeated 
  applications of Lemma~\ref{lemma:special-labminus} give contradictory 
  labels (both $\labplus$ and $\labminus$) to some arc for some update 
  schedule.
\end{proof}

\begin{lemma}
  \label{lemma:counting-special-pairs}
  For any $n > 6$, there exist $10n$ disjoint special pairs of schedules of size $n$
  for rule $128$.
\end{lemma}

\begin{proof}
  Fix $n\geq 6$ and
  consider the set of special pairs $\Delta,\Delta' \in \Pn$
  which have a difference between $\Delta$ and $\Delta'$
  on the labeling of arc $(1,0)$,
  with $\lab_\Delta((1,0))=\labplus$ and $\lab_{\Delta'}((1,0))=\labminus$.
  Lemma~\ref{lemma:special-labminus} fixes the labels of many arcs of $\Delta$,
  and from Lemma~\ref{lemma:special-onearc} the same labels hold for $\Delta'$
  since there is already a difference on arc $(1,0)$:
  \begin{align*}
    \text{for all } j \in \intz{n} \setminus\{0,1,2\} \text{ we have }
    \lab_\Delta((j,j+1)) &= \lab_{\Delta'}((j,j+1)) = \labminus\\
     \text{and } \lab_\Delta((j+1,j)) &= \lab_{\Delta'}((j+1,j)) = \labplus.
  \end{align*}
  Furthermore the labeling of arc $(1,0)$ is given by our hypothesis, and from
  Theorem~\ref{theorem:lab_valid} (to avoid a forbidden cycle of length two in
  $\Delta$) and Lemma~\ref{lemma:special-onearc} (equality of $\lab_\Delta$ and
  $\lab_{\Delta'}$ except for the arc $(1,0)$) we also have
  $\lab_\Delta((0,1)) = \lab_{\Delta'}((0,1)) = \labplus$.
  As a consequence it remains to consider $2^4$ possibilities for the labelings of arcs
  \[
    (1,2), (2,3), (2,1) \text{ and } (3,2)
  \]
  (which are equal on $\Delta$ and $\Delta'$, again by Lemma~\ref{lemma:special-onearc}).

  Among these, seven possibilities create a forbidden cycle of length two when
  the labels of the two arcs between cells 1 and 2, or 2 and 3, are both $\labminus$
  (see Figure~\ref{fig:special-cycle2}).
  
  \begin{figure}
    \centerline{\includegraphics{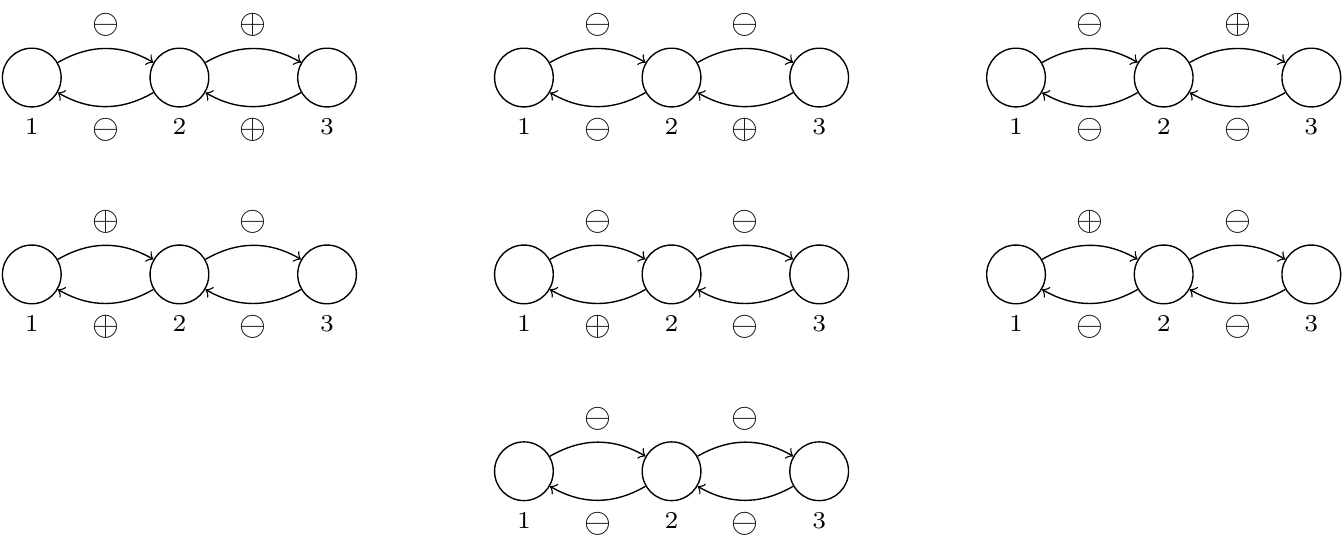}}
    \caption{
      seven labelings of arcs $(1,2)$, $(2,3)$, $(2,1)$ and $(3,2)$
      giving a forbidden cycle of length two, in the proof of
      Lemma~\ref{lemma:counting-special-pairs}.
    }
    \label{fig:special-cycle2}
  \end{figure}

  Among the remaining possibilities, four create a forbidden cycle of length $n$
  in $\Delta$, when the labels of arcs $(2,1)$ and $(3,2)$ are set to $\labplus$ 
  (see Figure~\ref{fig:special-cyclen}).

  \begin{figure}[htb]
    \centerline{\includegraphics{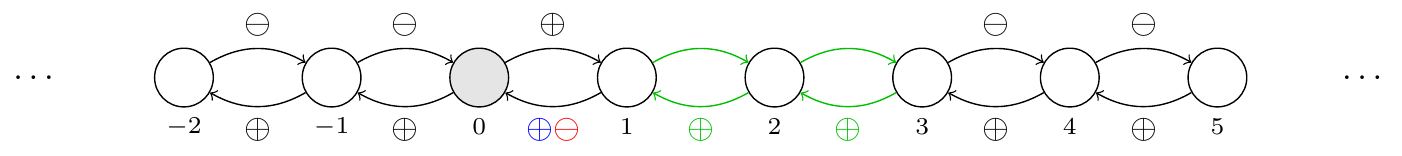}}
    \caption{
      four labelings of arcs $(1,2)$, $(2,3)$, $(2,1)$ and $(3,2)$
      giving a forbidden cycle of length $n$ in $\Delta$, in the proof of
      Lemma~\ref{lemma:counting-special-pairs},
      with $\lab_\Delta$ in blue, $\lab_{\Delta'}$ in
      red, in black the labels on which they are equal,
      and in green are highlighted the arcs on which we consider the
      $2^4$ possibilities.
      For any combination of $\labplus$ and $\labminus$ labels
      on arcs $(1,2)$ and $(2,3)$, the forbidden cycle is
      $0 \overset{\labminus}{\to}
      -1 \overset{\labminus}{\to}
      -2 \overset{\labminus}{\to}
      \dots \overset{\labminus}{\to}
      5 \overset{\labminus}{\to}
      4 \overset{\labminus}{\to}
      3 \overset{\labplus}{\to}
      2 \overset{\labplus}{\to}
      1 \overset{\labplus}{\to}
      0$
      (recall that the orientation of $\labminus$ arcs is reversed,
      Theorem~\ref{theorem:lab_valid}).
    }
    \label{fig:special-cyclen}
  \end{figure}

  The five remaining possibilities are presented on Figure~\ref{fig:special},
  one can easily check that they indeed correspond to special pairs:
  \begin{itemize}
    \item neither $\Delta$ nor $\Delta'$ contain a forbidden cycle. Hence, they are pairs
      of non-equivalent update schedule,
    \item for any $i \in \intz{n} \setminus \set{0}$, we have
      $\dleft{\Delta}(i)=\dleft{\Delta'}(i)$ and 
      $\dright{\Delta}(i)=\dright{\Delta'}(i)$. Hence,
      $\dset{\Delta}(i)=\dset{\Delta'}(i)$, and for cell $0$
      we have $\dset{\Delta}(0)=\dset{\Delta'}(0)=\intz{n}$.
  \end{itemize}

  \begin{figure}[htb]
    \centerline{\includegraphics{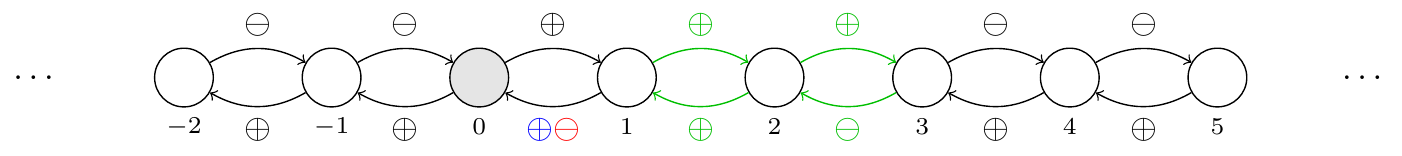}}
    \centerline{\includegraphics{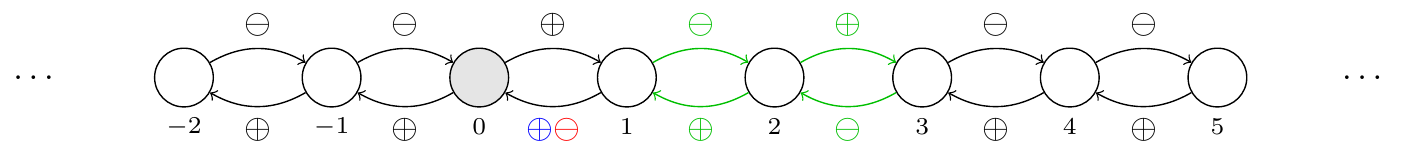}}
    \centerline{\includegraphics{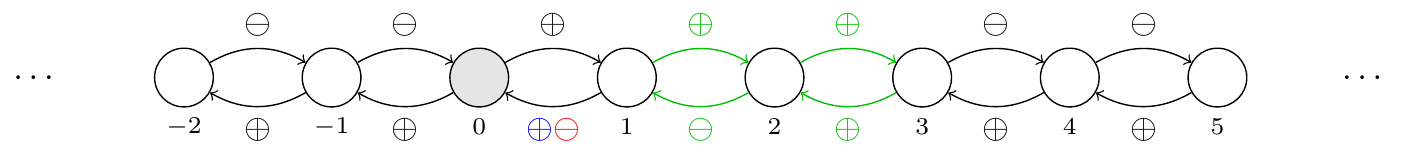}}
    \centerline{\includegraphics{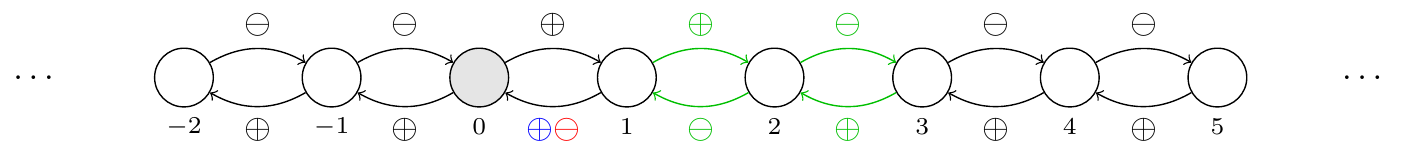}}
    \centerline{\includegraphics{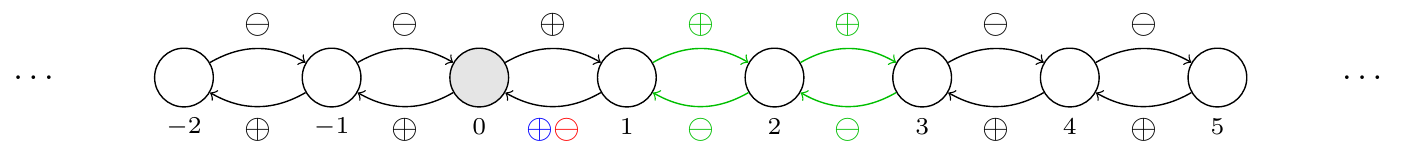}}
    \caption{
      five base special pairs $\Delta,\Delta'$ for rule $128$ in the proof of
      Lemma~\ref{lemma:counting-special-pairs},
      with $\lab_\Delta$ in blue, $\lab_{\Delta'}$ in
      red, in black the labels on which they are equal,
      and in green are highlighted the arcs on which we consider the
      five remaining possibilites.
    }
    \label{fig:special}
  \end{figure}

  \medskip

  We have seen so far that there are exactly five special pairs with their
  unique difference (Lemma~\ref{lemma:special-onearc}) on arc $(1,0)$.
  Let us call them the five {\em base} pairs and denote them as
  $\Delta^i,\Delta'^i$ for $i \in \intz{5}$.
  When we consider the $n$ rotations plus the left/right exchange
  (recall that rule $128$ is symmetric), we obtain
  $10n$ pairs:
  \begin{equation}
    \label{eq:special}
    \rho^k(\sigma^j(\Delta^i)),\rho^k(\sigma^j(\Delta'^i))
    \text{ for } i \in \intz{5}, j \in \intz{n}, k \in \intz{2}.
  \end{equation}
  Let us finally argue that these pairs are disjoint,
  \ie an update schedule belongs to at most one pair.

  First, one can straightforwardly check on Figure~\ref{fig:special} that the
  ten update schedules with a difference on arc $(1,0)$ are all distinct,
  hence the five base pairs are disjoint.

  Second, the $n$ rotations of these ten update schedules are all distinct
  when $n>6$, as can be noticed from $n$ letter words on alphabet
  $\{\labplus,\labminus\}$ given by
  \[
    \big(\lab_\Delta((i,i+1)))_{i \in \intz{n}}
    \text{ for some } \Delta.
  \]
  Indeed, each of these words contains a unique factor
  $\labminus\labminus\labminus\labplus$
  which allows to identify the number of left rotations applied to
  some $\Delta^i$ or $\Delta'^i$ with $i \in \intz{5}$
  in order to obtain $\Delta$.
  As a consequence, two distinct base update schedules remain distinct when
  some rotation is applied to one of them.

  Third, the left/right exchange of these $5n$ update schedules
  (base plus rotations) give $10n$ distinct update schedules,
  as can be noticed on the number of $\labminus$ labels on arcs of the form
  $(i,i+1)$ for $i \in \intz{n}$.
  Indeed, denoting
  \[
    \countminus{\Delta}=|\{ (i,i+1) \mid i \in \intz{n} \text{ and }
    \lab_\Delta((i,i+1))=\labminus \}|,
  \]
  we have for any $n>6$ that $\countminus{\Delta} > n-3$ when
  $\Delta$ is a base update schedule, the quantity is preserved
  by rotation, \ie $\countminus{\sigma(\Delta)}=\countminus{\Delta}$,
  but it holds that $\countminus{\Delta} > n-3$ if and only if
  $\countminus{\rho(\Delta)} < n-3$.
  As a consequence, two distinct update schedules
  (among the $5n$ update schedules $\sigma^j(\Delta^i),\sigma^j(\Delta^i)$
  for $i \in \intz{5}$ and $j \in \intz{n}$)
  remain distinct when the left/right exchange is applied to
  one of them. When the left/right exchange is applied to both of
  them then the situation is symmetric to the previous considerations.
  %
  %\medskip

  We conclude that the $10n$ pairs given by Formula~\ref{eq:special}
  are special and disjoint.
\end{proof}

As a consequence of Lemma~\ref{lemma:counting-special-pairs} we have
$|\set{ \dset{\Delta} \mid \Delta\in\Pn}| = 3^n-2^{n+1}-10n+2$ for any $n>6$,
and the result follows from Lemma~\ref{lemma:128d}.
\begin{theorem}
  \label{theorem:128}
  $\sens{f_{128,n}}=\frac{3^n-2^{n+1}-10n+2}{3^n-2^{n+1}+2}$ for any $n > 6$.
\end{theorem}

%%%%%%%%%%%%%%%%%%%%%%%%%%%%%%%%
% 162
\subsubsection{ECA rule 162}

The ECA rule $162$ is based on the Boolean function 
$r_{162}(x_1,x_2,x_3)= (x_1 \vee \neg x_2) \wedge x_3$.
Let $f$ be a shorthand for $f_{162,n}$ when the context is clear.

The structure of the reasoning is to first prove that for any
special pair $\Delta,\Delta' \in \Pn$ for rule $162$,
the labelings of arcs of the form $(i+1,i)$ for all $i \in \intz{n}$ are
identical in $\Delta$ and $\Delta'$ (Lemma~\ref{lemma:162-counterclockwise}).
Second, given a difference on the labels of some arc $(i,i+1)$, prove
that it forces all other labels both in $\Delta$ and in $\Delta'$
(Lemma~\ref{lemma:162-clockwise}).
Third, for the remaining case, prove that it is indeed a special pair,
thus generating $n$ disjoint special pairs by rotation (for any $n \geq 5$),
leading to Theorem~\ref{theorem:162}.

\begin{lemma}
  \label{lemma:162-counterclockwise}
  For any $n \geq 2$, if $\Delta,\Delta' \in \Pn$ is a special pair for rule $162$,
  then for all $i \in \intz{n}$ we have
  $\lab_\Delta((i+1,i))=\lab_{\Delta'}((i+1,i))$.
\end{lemma}

\begin{proof}
  By contradiction,
  assume that there exists $i \in \intz{n}$
  such that, WLOG, $\lab_\Delta((i+1,i))=\labplus$ whereas
  $\lab_{\Delta'}((i+1,i))=\labminus$.
  This implies $\lab_{\Delta'}((i,i+1))=\labplus$ otherwise
  there is a forbidden cycle of length two in $\Delta'$
  (Theorem~\ref{theorem:lab_valid}).
  See Figure~\ref{fig:162-counterclockwise} for an illustration of the setting.

  \begin{figure}
    \centerline{\includegraphics{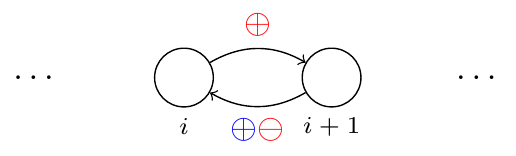}}
    \caption{
      illustration of the setting for the contradiction
      in Lemma~\ref{lemma:162-counterclockwise}, with $\lab_\Delta$ in blue
      and $\lab_{\Delta'}$ in red.
    }
    \label{fig:162-counterclockwise}
  \end{figure}

  Consider some $x \in \bool^n$ with $x_i=\0$ and $x_{i+1}=\1$
  (this require $n \geq 2)$).
  From our knowledge of $\Delta$ we have for some unknown $y_{i-1} \in \bool$
  that
  \[
    f^{(\Delta)}(x)_i=r_{162}(y_{i-1},x_i,x_{i+1})=r_{162}(y_{i-1},\0,\1)=\1.
  \]
  From our knowledge of $\Delta'$ we have for some unknown $y_{i-1},y_{i+2} \in \bool$
  that
  \begin{align*}
    f^{(\Delta')}(x)_i&=r_{162}(y_{i-1},x_i,r_{162}(x_i,x_{i+1},y_{i+2}))
    =r_{162}(y_{i-1},\0,r_{162}(\0,\1,y_{i+1}))\\
    &=r_{162}(y_{i-1},\0,\0)=\0.
  \end{align*}
  Thus $f^{(\Delta)}(x)_i \neq f^{(\Delta')}(x)_i$, a contradiction 
  to the fact that $\Delta,\Delta'$ is a special pair.
\end{proof}

From Lemma~\ref{lemma:162-counterclockwise} and the fact that $\Delta \not\equiv \Delta'$, we will now consider a special pair with a difference on some arc $(i,i+1)$ for $i \in \intz{n}$, and prove that this first difference enforces all the other labels both in $\Delta$ and in $\Delta'$.

\begin{lemma}
  \label{lemma:162-clockwise}
  For any $n \geq 3$, there is a unique special pair $\Delta,\Delta'\in\Pn$
  for rule $162$ with $\lab_\Delta((i-1,i)) \neq \lab_{\Delta'}((i-1,i))$
  for some $i \in \intz{n}$, and its labels are depicted on Figure~\ref{fig:162-special}.
\end{lemma}

\begin{figure}%[htb]
  \centerline{\includegraphics{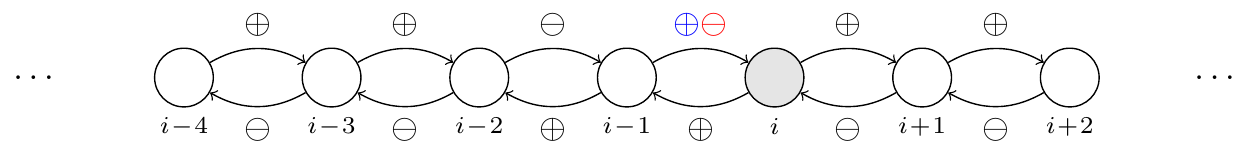}}
  \caption{
    the special pair $\Delta,\Delta' \in \Pn$ for rule $162$
    in Lemma~\ref{lemma:162-clockwise},
    with $\lab_{\Delta}((i-1,i))=\labplus$ and $\lab_{\Delta'}((i-1,i))=\labminus$
    for some $i \in \intz{n}$.
    $\lab_\Delta$ in blue, $\lab_{\Delta'}$ in red,
    and in black the labels on which they are equal.
    %Remark that taking $\rho(\Delta), \rho(\Delta')$ produces an invalid
    %cycle either in $\Delta$ or in $\Delta'$.
  }
  \label{fig:162-special}
\end{figure}

\begin{proof}
  WLOG, as on Figure~\ref{fig:162-special}, assume that 
  $\lab_{\Delta}((i-1,i))=\labplus$ and $\lab_{\Delta'}((i-1,i))=\labminus$.
  We deduce that $\lab_{\Delta'}((i,i-1))=\labplus$ otherwise
  there is a forbidden cycle of length two in $\Delta'$ (Theorem~\ref{theorem:lab_valid})
  and from Lemma~\ref{lemma:162-counterclockwise} it follows that we also have
  $\lab_{\Delta}((i,i-1))=\labplus$.

  We are going to prove that this forces all the other labels of $\Delta$
  $\Delta'$, \ie there is a unique such special pair. From the hypothesis
  that $\Delta,\Delta'$ is a special pair, we will use the fact that
  for all $x \in \bool^n$ and for all $j \in \intz{n}$ we have
  $f^{(\Delta)}(x)_j = f^{(\Delta')}(x)_j$.

%  \medskip

  From Lemma~\ref{lemma:162-counterclockwise} we deduce that
  $\dleft{\Delta}(i)=\dleft{\Delta}(i)$, meaning that at the time cell $i$
  is updated, the state of its right neighbor (cell $i+1$) are identical
  under update schedules $\Delta$ and $\Delta'$.
  Let us denote $y_{i+1} \in \bool$ this state for the rest of this proof.

  By contradiction assume that it is possible to have some configuration
  $x \in \bool^n$ such that
  \begin{equation}
    \label{eq:162-yi+1eq1}
    x_{i-1}=\0 \text{ and } x_i=\1 \text{ and } y_{i+1}=\1
  \end{equation}
  (this requires $n \geq 3$).
  In this case we have
  \[
    f^{(\Delta)}(x)_i=r_{162}(x_{i-1},x_i,y_{i+1})=r_{162}(\0,\1,\1)=\0
  \]
  but for some unknown $y_{i-2} \in \bool^n$ we always have
  \begin{align*}
    f^{(\Delta')}(x)_i&=r_{162}(r_{162}(y_{i-2},x_{i-1},x_i),x_i,y_{i+1})=r_{162}(r_{162}(y_{i-2},\0,\1),\1,\1)\\
    &=r_{162}(\1,\1,\1)=\1
  \end{align*}
  \ie $f^{(\Delta)}(x)_i \neq f^{(\Delta')}(x)_i$ which contradicts
  the hypothesis that $\Delta$,$\Delta'$ is a special pair.
  We conclude that it must be impossible to have simultaneously
  $x_{i-1}=\0$, $x_i=\1$ and $y_{i+1}=\1$. This hints at the fact that
  the value of $\dright{\Delta}(i)=\dright{\Delta'}(i)$ must be close to $n$
  so that the constraints on $x_{i-1}$ and $x_i$ make it impossible
  to obtain $y_{i+1}=\1$ when updating the chain of influence to the right of
  cell $i$.
  This is what we are going to prove formally, via the following case disjunction.
  \begin{itemize}
    \item If $\dright{\Delta}(i)=\dright{\Delta'}(i) < n-1$ then
      consider $x \in \bool^n$ with
      $x_{i-1}=\0$ and $x_i=x_{i+1}=\dots=x_{i+\dright{\Delta}(i)}=\1$.
      From our current hypothesis on $\dright{\Delta}(i)$ and for $n \geq 3$,
      such a configuration exists.
      We deduce from the definition of rule $162$ that the updates
      (in this order, both in $\Delta$ and $\Delta'$) of cells
      $i+\dright{\Delta}(i), i+\dright{\Delta}(i)-1, \dots, i+1$ all give
      state $\1$, \ie in particular $y_{i+1}=\1$, leading to a contradiction
      as developed from Equation~\ref{eq:162-yi+1eq1}.
    \item If $\dright{\Delta}(i)=\dright{\Delta'}(i) \geq n$ then
      there is a forbidden cycle of length $n$ in $\Delta$:
      \begin{equation}
        \label{eq:162-fcycle}
        i \overset{\labminus}{\to} i+1 \overset{\labminus}{\to} \dots
        \overset{\labminus}{\to} i-2 \overset{\labminus}{\to} i-1
        \overset{\labplus}{\to} i
      \end{equation}
      (recall that the orientation of $\labminus$ arcs is reversed, see
      Theorem~\ref{theorem:lab_valid}).
      As a consequence we discard this case.
    \item If $\dright{\Delta}(i)=\dright{\Delta'}(i) = n-1$ then it means that
      we have $\lab_\Delta((j+1,j))=\lab_{\Delta'}((j+1,j))=\labminus$ for all
      $j \in \intz{n}\setminus\{i-2,i-1\}$, and
      $\lab_\Delta((i-1,i-2))=\lab_{\Delta'}((i-1,i-2))=\labminus$,
      but also $\lab_\Delta((j,j+1))=\lab_{\Delta'}((j,j+1))=\labplus$ for all
      $j \in \intz{n}\setminus\{i-2,i-1\}$ from Theorem~\ref{theorem:lab_valid}.

      Therefore it only remains to consider the labels of arc $(i-2,i-1)$
      in schedules $\Delta$ and $\Delta'$. To avoid a forbidden cycle of length
      $n$ in $\delta$, similar to Equation~\ref{eq:162-fcycle} with
      $i-2 \overset{\labplus}{\to} i-1$, we need to set
      $\lab_\Delta((i-2,i-1))=\labminus$. It only remains to consider
      $\lab_{\Delta'}((i-2,i-1))$.

      Suppose for the contradiction that $\lab_{\Delta'}((i-2,i-1))=\labplus$,
      then similarly to our previous reasoning, for some $x \in \bool^n$
      with $x_{i-2}=\0$, $x_{i-1}=\1$ and $x_i=\1$, we have
      for some unknown $y_{i-3} \in \bool$ that
      \begin{align*}
        f^{(\Delta)}(x)_{i-1}&=r_{162}(r_{162}(y_{i-3},x_{i-2},x_{i-1}),x_{i-1},x_i)=r_{162}(r_{162}(y_{i-3},\0,\1),\1,\1)\\
        &=r_{162}(\1,\1,\1)=\1
      \end{align*}
      whereas
      \[
        f^{(\Delta)}(x)_{i-1}=r_{162}(x_{i-2},x_{i-1},x_i)=r_{162}(\0,\1,\1)=\0
      \]
      thus $f^{(\Delta)}(x)_{i-1} \neq f^{(\Delta')}(x)_{i-1}$, contradicting
      the fact that $\Delta,\Delta'$ is a special pair.
  \end{itemize}

  We conclude that there is only one remaining possible special pair with a difference
  on the labelings of arc $(i-1,i)$, and that it is the one given on
  Figure~\ref{fig:162-special}.

  %\medskip

  Let us finally prove that this is indeed a special pair.
  One easily checks on Figure~\ref{fig:162-special} that
  the update schedules of this pair have no forbidden cycle and are non-equivalent.
  For any $j \in \intz{n}\setminus\{i\}$ we have
  $\dleft{\Delta}(j)=\dleft{\Delta'}(j)$ and 
  $\dright{\Delta}(j)=\dright{\Delta'}(j)$, \ie the chain of influences are identical
  hence for all $x \in \bool^n$ we have
  $f^{(\Delta)}(x)_{j} = f^{(\Delta')}(x)_{j}$.

  Regarding cell $i$, we have $\dright{\Delta}(i)=\dright{\Delta'}(i)$,
  meaning that at the time cell $i$ is updated, its right neighbor (cell $i+1$)
  will be in the same state (denoted $y_{i+1}$) in both update schedules.
  Given some $x \in \bool^n$, we proceed to a case disjunction.
  \begin{itemize}
    \item If $y_{i+1}=\0$ then
      \[
        f^{(\Delta)}(x)_i=r_{162}(x_{i-1},x_i,y_{i+1})=r_{162}(x_{i-1},x_i,\0)=\0
      \]
      and for some unknown $y_{i-1} \in \bool$ we have
      \[
        f^{(\Delta')}(x)_i=r_{162}(y_{i-1},x_i,y_{i+1})=r_{162}(y_{i-1},x_i,\0)=\0
      \]
      therefore we conclude $f^{(\Delta)}(x) = f^{(\Delta')}(x)$.
    \item If $y_{i+1}=\1$ then, from the reasoning we have just made above
      and since cell $i+2$ is updated prior to cell $i+1$, we deduce that cell
      $i+2$ is updated to state $\1$ otherwise cell $i+1$ would be updated to state $\0$
      (in both $\Delta$ and $\Delta'$, contradicting our last hypothesis
      that $y_{i+1}=\1$). This applies to cell $i+3$, {\em etc}, until cell $i-2$
      which must also be updated to state $\1$, and finally cell $i-1$ which must be
      in state $\1$, \ie $x_{i-1}=\1$.

      We deduce from $y_{i+1}=\1$ and $x_{i-1}=\1$ that
      \[
        f^{(\Delta)}(x)_i=r_{162}(x_{i-1},x_i,y_{i+1})=r_{162}(\1,x_i,\1)=\1.
      \]
      Regarding cell $i$ in the update schedule $\Delta'$, we proceed to a last case disjunction.
      \begin{itemize}
        \item If $x_i=\0$ then for some unknown $y_{i-1}\in\bool$ we have
          \[
            f^{(\Delta')}(x)_i=r_{162}(y_{i-1},x_i,y_{i+1})=r_{162}(y_{i-1},\0,\1)=\1
          \]
          and we conclude $f^{(\Delta)}(x) = f^{(\Delta')}(x)$.
        \item If $x_i=\1$ then we can use our prior deduction that cell $i-2$
          is updated to state $\1$, therefore
          \begin{align*}
            f^{(\Delta')}(x)_i&=r_{162}(r_{162}(\1,x_{i-1},x_i),x_i,y_{i+1})=r_{162}(r_{162}(\1,\1,\1),\1,\1)\\
            &=r_{162}(\1,\1,\1)=\1
          \end{align*}
          and we also conclude $f^{(\Delta)}(x) = f^{(\Delta')}(x)$ in this
          ultimate case.
      \end{itemize}
  \end{itemize}
  We have seen that for any $x \in \bool^n$, $f^{(\Delta)}(x) = f^{(\Delta')}(x)$.
  Thus, $\Delta,\Delta'$ is a special pair,
  and, from the first part of this proof, it is unique.
\end{proof}

\begin{theorem}
\label{theorem:162}
  $\sens{f_{162,n}}=\frac{3^n-2^{n+1}-n+2}{3^n-2^{n+1}+2}$ for any $n \geq 3$.
\end{theorem}

\begin{proof}
  From Lemmas~\ref{lemma:162-counterclockwise} and~\ref{lemma:162-clockwise}
  there are $n$ pairs of special pairs for rule $162$
  (no pair with a difference on the label of an arc of the form $(i+1,i)$
  for some $i \in \intz{n}$ by Lemma~\ref{lemma:162-counterclockwise},
  and exactly one pair with a difference on arc $(i,i+1)$ for each $i \in \intz{n}$
  by Lemma~\ref{lemma:162-clockwise}).
  Denoting $\Delta,\Delta'$ the special pair given by
  Lemma~\ref{lemma:162-clockwise} with a difference on the arc $(0,1)$,
  the $n$ special pairs are $\sigma^j(\Delta),\sigma^j(\Delta')$ for $j \in \intz{n}$.
  Lemmas~\ref{lemma:162-counterclockwise} and~\ref{lemma:162-clockwise} hold
  for any $n \geq 3$, and for any such $n$
  one easily checks by considering the word formed by the labels of arcs
  $(i-1,i)$ for $i \in \intz{n}$ (this word is identical for both schedules of
  each special pair, by Lemma~\ref{lemma:162-counterclockwise}) that these
  pairs are disjoint: these words contain exactly one factor $\oplus\oplus$
  whose position differs for any rotation of $\Delta,\Delta'$.
%  two letters $\labplus$ whose
%  couple of positions differ for any rotation of $\Delta,\Delta'$.
  It follows that among the $3^n-2^{n+1}+2$ non-equivalent update schedules,
  we have $\dynamics{f_{162,n}}=3^n-2^{n+1}+2-n$, as stated.
\end{proof}

% 160
\subsubsection{ECA rule 160}

The ECA rule $160$ somewhat similar to ECA rule $128$. Ideed, it is based on the 
Boolean function $r_{160}(x_1,x_2,x_3)= x_1 \wedge x_3$.

\begin{remark}
  \label{remark:160}
  It is clear from the definition of $r_{160}$ that for any update schedule
  $\Delta$ and any configuration $x\in\zu^n$ such that $x_{i}=x_{i+1}=\0$ (or
  $x_{i-1}=x_{i}=\0$) for some $i\in\intz{n}$ it holds
  $f^{(\Delta)}(x)_i=\0$. 
  %Remark that if $d_\Delta(i)=\intz{n}$ for some
  %$i\in\intz{n}$, then the only possibility to have $f^{(\Delta)}(x)_i=\1$ is
  %either $x=\1^n$
  %\TODO{(kevin: is this last part true??) or $x=y$ where $y$ is such that $y_j=\1$ iff $|i-j|$ is odd.}
\end{remark}

We are going to adopt a adopt a reasoning analogous to rule
$128$ for the study of the sensitivity to synchronism of rule $160$.
Lemma~\ref{lemma:160-special-labminus} will be the first stone showing that as
soon as two update schedules form a special pair for rule $160$,
the position of their difference enforces the labels of many other arcs.
Then Lemma~\ref{lemma:160-special-onearc} will use applications of
Lemma~\ref{lemma:160-special-labminus} according to some carefuly crafted
case disjunction, in order to prove that a special pair with more than one
difference among the two update schedules (\ie differences on the labels of at
least two arcs) is contradictory. Finally, Lemma~\ref{lemma:counting-special-pairs-160}
will use these previous results to characterize exactly the special pairs of
update schedule for rule $160$, which comes down to six disjoint base special pairs,
leading to $12n$ special pairs when considering left/right exchange and rotations.
This will give Theorem~\ref{theorem:160}.

\begin{lemma}
  \label{lemma:160-special-labminus}
  For any $n > 4$, consider a special pair $\Delta,\Delta' \in \Pn$ for rule $160$
  such that
  $\lab_\Delta((i+1,i))=\labplus$ and $\lab_{\Delta'}((i+1,i))=\labminus$
  for some $i\in\intz{n}$.
  For all $j\in\intz{n}\setminus\set{i,i+1,i+2,i+3}$, it holds
  $\lab_\Delta((j,j+1))=\labminus$, $\lab_\Delta((j+1,j))=\labplus$
  and also $\lab_{\Delta'}((i+2,i+1))=\labminus$, $\lab_{\Delta'}((i+1,i+2))=\labplus$.
\end{lemma}

\begin{proof}
  Let us prove that, with the hypothesis of the statement,
  we must have $\dleft{\Delta}(i) \geq n-4$
  (which implies the $\labminus$ labels on $\Delta$) and also
  $\lab_{\Delta'}((i+2,i+1))=\labminus$. The complete result follows by
  application of Theorem~\ref{theorem:lab_valid} to get the $\labplus$
  labels (in order to avoid any forbidden cycle of length two).

  \medskip

  For the first part, if $\dleft{\Delta}(i) < n-3$ then
  we can construct the following configuration $x \in \bool^n$
  without a contradiction on the states of cells $i-\dleft{\Delta}(i)$ and $i+3$:
  \begin{itemize}
    \item $x_{i+2}=x_{i+3}=\0$
    \item $x_i=x_{i-1}=\dots=x_{i-\dleft{\Delta}(i)}=\1$.
  \end{itemize}
  This requires $n \geq 5$,
  see Figure~\ref{fig:160-special-labminus} for an illustration.
  Regarding $\Delta'$, it follows from Remark~\ref{remark:160} that cell $i+2$
  remains in state $\0$, and as a consequence, regardless of the label of arc
  $(i+2,i+1)$, cell $i+1$ is updated to state $\0$, then so is $i$.
  However in $\Delta$, we have $x_j=\1$ for all $j \in \dset{\Delta}(i)$,
  \ie cell $i$ depends only on cells in state $\1$, and we deduce that it is
  updated to state $\1$.
  Thus $f_{160,n}^{(\Delta)}(x)_i \neq f_{160,n}^{(\Delta')}(x)_i$,
  a contradiction to the fact that $\Delta,\Delta'$ is a special pair.

  \medskip

  For the second part, suppose for the contradiction that
  $\lab{\Delta'}((i+2,i+1)=\labplus$, and consider the configuration $x \in \bool^n$
  with $x_{i+1}=\0$ and state $\1$ in all other cells. In $\Delta$, at time cell $0$
  is updated it has a state $\0$ on its right (cell $i+1$, not yet updated), and
  $f_{160,n}^{(\Delta)}(x)_i=\0$. In $\Delta'$, cell $i+1$ is updated prior to
  its left and right neighbors (from Theorem~\ref{theorem:lab_valid} again we have
  $\lab_{\Delta'}((i,i+1))=\labplus$) thus it goes to state $\1$. We can deduce
  from this that all cells will go to state $\1$ because they all have two neighbors
  in state $\1$ at the time they are updated. Therefore in particular
  $f_{160,n}^{(\Delta')}(x)_i=\1$, again a contradiction.
\end{proof}

\begin{figure}
  \centerline{\includegraphics{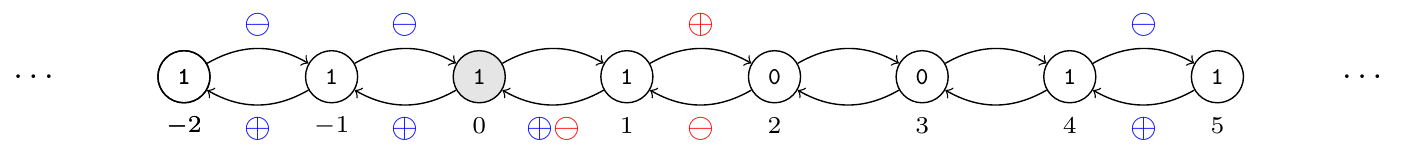}}
  \caption[\TODO{problem with dleft in caption solved by adding this stupid short text}]{
    illustration of Lemma~\ref{lemma:160-special-labminus}, with $\Delta$ in blue
    and $\Delta'$ in red: hypothesis on the labelings of arc $(1,0)$
    imply many $\labminus$ labels on arcs of the form $(j,j+1)$, and
    $\labplus$ labels on arcs of the form $(j+1,j)$, for $\Delta$,
    and labels on arcs $(2,1)$ and $(1,2)$ for $\Delta'$.
    Inside the cells are depicted the states corresponding to a contradictory
    configuration (having different images on cell $0$)
    when $\dleft{\Delta}(0) \leq n-4$.
  }
  \label{fig:160-special-labminus}
\end{figure}

Let us recall that rule $160$ is symmetric, therefore
Lemma~\ref{lemma:160-special-labminus} also applies with a left/right exchange.

\begin{lemma}
  \label{lemma:160-special-onearc}
  For any $n > 8$, if $\Delta,\Delta' \in \Pn$ is a special pair for rule $160$
  then $\Delta$ and $\Delta'$ differ on the labeling of exactly one arc.
\end{lemma}
\begin{proof}
  Up to rotation and right/left exchange, let us suppose WLOG that
  $\lab_\Delta((1,0))=\labplus$ and $\lab_{\Delta'}((1,0))=\labminus$.
  Now, for the sake of contradiction, assume that they also differ on 
  another arc, and consider the following cases disjunction
  (remark that the order of the case study is chosen so that
  cases make reference to previous cases).
  \begin{enumerate}[label=(\alph*)]
    \item If $\lab_\Delta((i,i+1))=\labplus$ and   
      $\lab_{\Delta'}((i,i+1))=\labminus$
      for some $i \in \intz{n}$, then by applying Lemma~\ref{lemma:160-special-labminus} to
      the two arcs where $\Delta$ and $\Delta'$ differ leads to a contradiction
      on the labeling of some arc according to $\Delta$.
      Indeed, Lemma~\ref{lemma:160-special-labminus} is applied to two arcs in 
      different directions, one application leaves four arcs of the form 
      $(j,j+1)$ not labeled $\labminus$ in $\Delta$ and three arcs of the form 
      $(j+1,j)$ not labeled $\labplus$ in $\Delta$, the converse for the other 
      application, hence starting from $n=8$ these labelings overlap in a contradictory fashion.
    \item If $\lab_\Delta((i+1,i))=\labminus$ and 
      $\lab_{\Delta'}((i+1,i))=\labplus$
      for some $i \in \intz{n}\setminus\set{0}$, then $i \in \set{1,2,3}$ otherwise 
      there is a forbidden cycle of length two in $\Delta$ with some
      $\labminus$ label given by the application of
      Lemma~\ref{lemma:160-special-labminus} to the arc $(1,0)$. 
      However, for $i\in\set{2,3,4}$ the application of Lemma~\ref{lemma:160-special-labminus} 
      to the arc $(i+1,i)$ gives $\lab_{\Delta'}((0,1))=\labminus$, creating a
      forbidden cycle of length two in $\Delta'$. 
    \item\label{item:special-counterclockwiseplusminus}
      If $\lab_\Delta((i+1,i))=\labplus$ and $\lab_{\Delta'}((i+1,i))=\labminus$
      for some $i \in \intz{n} \setminus\{1,2\}$, then applying
      Lemma~\ref{lemma:160-special-labminus} to the two arcs where $\Delta$ and $\Delta'$
      differ leads to a forbidden cycle of length $n$ in $\Delta$
      (contradiction Theorem~\ref{theorem:lab_valid}).
      Indeed, if $i \notin\set{1,2,3}$ then we have $\labminus$ labels on arcs
      of the form $(j,j+1)$ for all $j \in \intz{n}$, and if $i=3$ then the
      forbidden cycle contains the arc $(4,3)$ labeled $\labplus$. 
            The case $i=0$ is not a second difference.
    \item\label{item:special-clockwiseminusplus}
      If $\lab_\Delta((i,i+1))=\labminus$ and $\lab_{\Delta'}((i,i+1))=\labplus$
      for some $i \in \intz{n}$, then applying Lemma~\ref{lemma:160-special-labminus} to
      arc $(1,0)$ gives $\lab_\Delta((j+1,j))=\labplus$ for all $j \in
      \intz{n} \setminus \{0,1,2,3\}$, and applying
      Lemma~\ref{lemma:160-special-labminus} to arc $(i,i+1)$ gives
      $\lab_{\Delta'}((j+1,j))=\labminus$ for all $j \in \intz{n} \setminus
      \set{i,i-1,i-2}$. Starting from $n=9$ we have $(\intz{n} \setminus \set{0,1,2,3}) 
      \cap (\intz{n}\setminus\set{i,1-1,i-2}) \neq \emptyset$, and as a consequence
      there is an arc $((j+1,j))$ in the case of
      Item~\ref{item:special-counterclockwiseplusminus}.
    \item 
      If $\lab_\Delta((2,1))=\labplus$ and $\lab_{\Delta'}((2,1))=\labminus$,
      then applying Lemma~\ref{lemma:160-special-labminus} to arc $(2,1)$
      gives $\lab_\Delta((0,1))=\labminus$, however since by hypothesis
      $\lab_{\Delta'}((1,0))=\labminus$ we also have $\lab_{\Delta'}((0,1))=\labplus$
      otherwise there is a forbidden cycle of length two in $\Delta'$
      (Theorem~\ref{theorem:lab_valid}). As a consequence, the arc $(0,1)$ is in
      the case of Item~\ref{item:special-clockwiseminusplus}. The arc $(3,2)$ is 
      involved in the same situation.
  \end{enumerate}
  We conclude that in any case a second difference leads to a contradiction,
  either because an invalid cycle is created, or because repeated 
  applications of Lemma~\ref{lemma:160-special-labminus} give contradictory 
  labels (both $\labplus$ and $\labminus$) to some arc for some update 
  schedule.
  \end{proof}

\begin{lemma}
  \label{lemma:counting-special-pairs-160}
  For any $n > 8$, there exist $12n$ disjoint special pairs of schedules of size $n$
  for rule $160$.
\end{lemma}
\begin{proof}
  The structure of this proof is very similar to
  Lemma~\ref{lemma:counting-special-pairs}.
  Fix $n>8$ and consider the set of special pairs $\Delta,\Delta' \in \Pn$
  which have a difference between $\Delta$ and $\Delta'$
  on the labeling of arc $(1,0)$,
  with $\lab_\Delta((1,0))=\labplus$ and $\lab_{\Delta'}((1,0))=\labminus$.
  Lemma~\ref{lemma:160-special-labminus} fixes the labels of many arcs of $\Delta$,
  and from Lemma~\ref{lemma:160-special-onearc} the same labels hold for $\Delta'$
  since there is already a difference on arc $(1,0)$:
  \begin{align*}
    \text{for all } j \in \intz{n} \setminus\{0,1,2,3\} \text{ we have }
    \lab_\Delta((j,j+1)) &= \lab_{\Delta'}((j,j+1)) = \labminus\\
    \text{and } \lab_\Delta((j+1,j)) &= \lab_{\Delta'}((j+1,j)) = \labplus,\\
    \text{and furthermore }
    \lab_\Delta((1,2)) &= \lab_{\Delta'}((1,2)) = \labplus\\
    \text{and } \lab_\Delta((2,1)) &= \lab_{\Delta'}((2,1)) = \labminus,
  \end{align*}
  Furthermore the labeling of arc $(1,0)$ is given by our hypothesis, and from
  Theorem~\ref{theorem:lab_valid} (to avoid a forbidden cycle of length two in
  $\Delta$) and Lemma~\ref{lemma:160-special-onearc} (equality of $\lab_\Delta$ and
  $\lab_{\Delta'}$ except for the arc $(1,0)$) we also have
  $\lab_\Delta((0,1)) = \lab_{\Delta'}((0,1)) = \labplus$.
  As a consequence it remains to consider $2^4$ possibilities for the labelings of arcs
  \[
    (2,3), (3,4), (3,2) \text{ and } (4,3)
  \]
  (which are equal on $\Delta$ and $\Delta'$, again by Lemma~\ref{lemma:160-special-onearc}).

  Among these, seven possibilities create a forbidden cycle of length two when
  the labels of the two arcs between cells 1 and 2, or 2 and 3, are both $\labminus$
  (see Figure~\ref{fig:special-cycle2} relative to rule $128$, the seven possibilities
  for rule $160$ are analogous with the four respective arcs we are now considering).
  
% % % % % % % % % % % % % % % % % % % % % % % % % % %
  Among the nine remaining possibilities, three do not correspond to special pairs,
  as we will prove now by exhibit for each of them a configuration $x\in\bool^n$
  such that the images at cell $0$ differ in $\Delta$ and $\Delta'$.
  These three possibilities are depicted on Figure~\ref{fig:160-wanted},
  let us denote them $\hat{\Delta}^i,\smash{\hat{\Delta'}}^i$ for $i \in \intz{3}$.
  \begin{itemize}
    \item For $\hat{\Delta}^0,\smash{\hat{\Delta'}}^0$ we have $x\in\bool^n$ with
      $x_2=\0$ and all other cells in state $\1$,
    \item For $\hat{\Delta}^1,\smash{\hat{\Delta'}}^1$ we have $x\in\bool^n$ with
      $x_2=\0$ and all other cells in state $\1$,
    \item For $\hat{\Delta}^2,\smash{\hat{\Delta'}}^2$ we have $x\in\bool^n$ with
      $x_3=\0$ and all other cells in state $\1$.
  \end{itemize}
  One can check that in these three cases $i \in \intz{n}$ with these three respective configurations,
  we have $f^{(\hat{\Delta}^i)}(x)_0 = \1$ but $f^{(\smash{\hat{\Delta'}}^i)}(x)_0=\0$,
  because in both update schedules of each pair the left neighbor of cell $0$
  (cell $-1$) will be updated to state $\1$, and the right neighbor of cell $0$
  (cell $1$) will be updated to state $\0$ before the update of cell $0$ in $\smash{\hat{\Delta'}}^i$
  whereas it is still in state $x_1=\1$ when cell $0$ is updated in $\hat{\Delta}^i$.
 
   \begin{figure}
     \centerline{\includegraphics{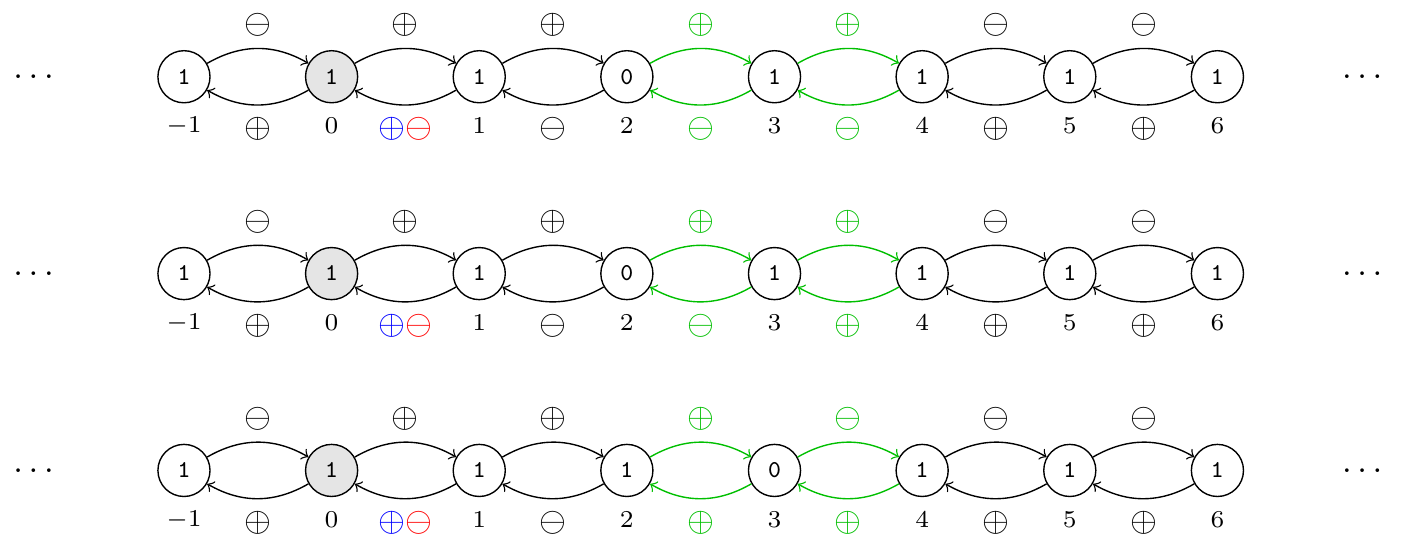}}
    \caption{
      three pairs $\hat{\Delta}^i,\smash{\hat{\Delta'}}^i$ for $i \in \intz{3}$ for rule $160$
      ($\hat{\Delta}^0,\smash{\hat{\Delta'}}^0$ on the top;
      $\hat{\Delta}^1,\smash{\hat{\Delta'}}^1$ in the middle;
      $\hat{\Delta}^2,\smash{\hat{\Delta'}}^2$ on the bottom)
      not corresponding to special pairs because for each of them there exists a configuration
      $x\in\bool^n$ such that $f^{(\hat{\Delta}^i)}(x)_0 = \1 \neq
      \0 = f^{(\smash{\hat{\Delta'}}^i)}(x)_0$. The states of configuration $x$ are given inside
      the cells.
    }
     \label{fig:160-wanted}
   \end{figure}

% % % % % % % % % % % % % % % % % % % % % % % % % % % 

  The six remaining possibilities are presented on Figure~\ref{fig:160-special}.
  Let us argue that they indeed correspond to special pairs:
  \begin{itemize}
    \item neither $\Delta$ nor $\Delta'$ contain a forbidden cycle. Hence, they are pairs
      of non-equivalent update schedule,
    \item for any $i \in \intz{n} \setminus \set{0}$, we have
      $\dleft{\Delta}(i)=\dleft{\Delta'}(i)$ and 
      $\dright{\Delta}(i)=\dright{\Delta'}(i)$. Hence,
      $f^{(\Delta)}(x)_i=f^{(\Delta')}(x)_i$ for any $x \in \bool^n$
      (Lemma~\ref{lemma:di}).
      For cell $0$ let us show that $f^{(\Delta)}(x)_0=f^{(\Delta')}(x)_0$ for any $x\in\zu^n$.
      In order to have a difference in the update of cell $0$, one of the
      two update schedules must update it to state $\1$.
      Now remark that, given the definition of rule $160$,
      the only possibility for cell $0$ to be updated to state $\1$ in some
      update schedule (recall that $\dleft{\Delta}(0)=\dleft{\Delta'}(0)$) is that
      $x_0=x_{-1}=x_{-2}=\dots=x_{-\dleft{\Delta}(0)-2}=\1$, and
      $x_{-\dleft{\Delta}(0)}=\1$. Indeed, if any of these cells is in state
      $\0$, then at some point in the update of the chain of influence to the
      left of cell $0$ (in this order: cell $-\dleft{\Delta}$ then
      $-\dleft{\Delta}+1$ then \dots then $-1$ and finally $0$) some cell will
      be updated to state $\0$, and then all subsequent cells will be updated
      to state $\0$ as well.
      Given that $\dleft{\Delta}(0)=\dleft{\Delta'}(0) \geq n-3$,
      this would enforce the states of all cells in $x$ except (in the order of
      Figure~\ref{fig:160-special}):
      \begin{itemize}
          \item cells $1$ and $3$ for the first and third pairs,
          \item cells $1,2$ and $4$ for the second, fourth and fifth pairs,
          \item cell $2$ for the sixth pair.
      \end{itemize}
      A straightfoward exhaustive analysis of these $2 \times 2^2 + 3 \times 3^2 + 2$ cases
      would convince the reader that, for any configuration $x \in \bool^n$ where
      cell $0$ may be updated to state $\1$ in $\Delta$ or in $\Delta'$
      (otherwise $f^{(\Delta)}(x)_0=f^{(\Delta')}(x)_0=\0$),
      it turns out that $f^{(\Delta)}(x)_0=f^{(\Delta')}(x)_0$ (this is tedious but
      reveals the nice combinatorics of green labels on Figure~\ref{fig:160-special}).
 \end{itemize}

  \begin{figure}[htb]
    \centerline{\includegraphics{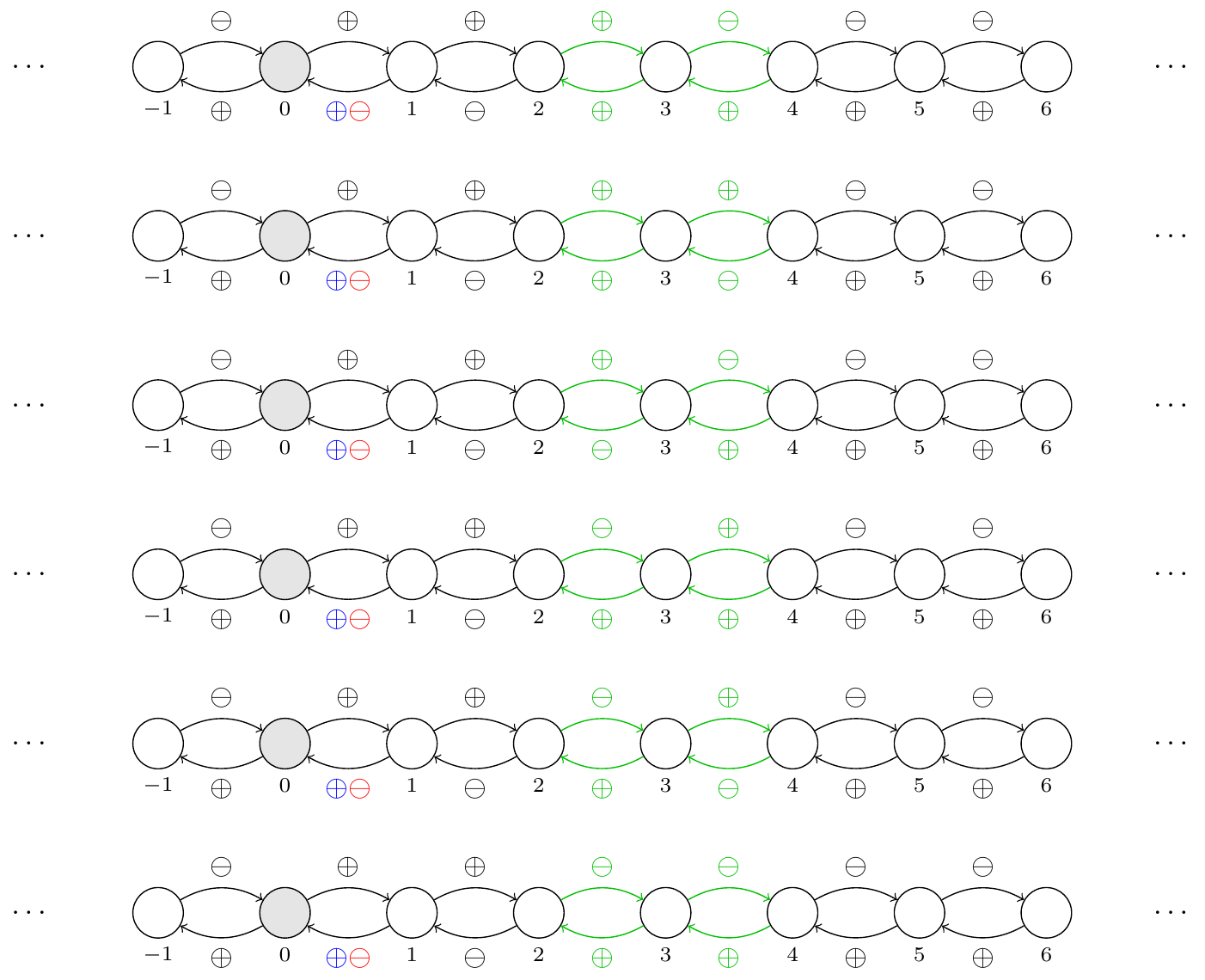}}
    \caption{
      six base special pairs $\Delta,\Delta'$ for rule $160$ in the proof of
      Lemma~\ref{lemma:counting-special-pairs-160},
      with $\lab_\Delta$ in blue, $\lab_{\Delta'}$ in
      red, in black the labels on which they are equal,
      and in green are highlighted the arcs on which we consider the
      six remaining possibilities.
    }
    \label{fig:160-special}
  \end{figure}

  \medskip

  We have seen so far that there are exactly six special pairs with their
  unique difference (Lemma~\ref{lemma:160-special-onearc}) on arc $(1,0)$.
  Let us finally argue that these six {\em base} pairs for rule $160$
  give $12n$ distinct pairs when considering their rotations and
  left/right exchange, \ie an update schedule belongs to at most one pair.
  
  It is clear from Figure~\ref{fig:160-special} that all the base pairs are all
  disjoint.
  Moreover, considering the pattern 
  $\labminus\labminus\labminus\labminus\labplus\labplus$
  and any of the $24n$ update schedules,
  for any $n > 8$ either it appears exactly once on arcs of the form
  $(i,i+1)$, or its mirror appears exaclty once on arcs of the form
  $(i+1,i)$, but not both.
  This allows to uniquely determine the left/right exchange and
  rotations applied to some base pair,
  and the remaining labelings straighforwardly allow to determine
  one of the six base special pair, and one of $\Delta$ or $\Delta'$.
  Therefore all special pairs are disjoint.
\end{proof}

As a consequence of Lemma~\ref{lemma:counting-special-pairs-160} we have
the following result.

\begin{theorem}
\label{theorem:160}
  $\sens{f_{160,n}}=\frac{3^n-2^{n+1}-12n+2}{3^n-2^{n+1}+2}$ for any $n>8$.
\end{theorem}

%%%%%%%%%%%%%%%%%%%%%%%%%%%%%%%%
% conclusion and perspectives
%%%%%%%%%%%%%%%%%%%%%%%%%%%%%%%%
\section{Conclusion and perspectives}
\label{s:conclusions}

Asynchrony highly impacts the dynamics of CA and new original dynamical behaviors are introduced.
In this new model, the dynamics become dependent from the update schedule of cells.
However, not all schedules produce original dynamics. For this reason, a measure to quantify
the sensitivity of ECA \wrt to changes of the update schedule has been introduced in~\cite{rmmop18}.
All ECA rules were then classified into two classes: max-sensitive and non-max sensitive.

This paper provides a finer study of the sensitivity measure \wrt the size of the configurations.
Indeed, we found that there are four classes (see Table~\ref{tab:results}). In particular, it is interesting to remark that
the asymptotic behavior is not dichotomic, \ie, the sensitivity function does not always either go
to $0$ or to $1$ when the size of configurations grows. The ECA rule $8$ when considered as a classical
ECA (\ie, when all cells are updated synchronously) has a very simple dynamical behavior but its asynchronous
version has a sensitivity to asynchronism function which tends to $\frac{1+\phi}{3}$ when $n$ tends to infinity
($\phi$ is the golden ratio). Remark that in the classical case, the limit set of the ECA rule $8$ is the same as ECA
rule $0$ after just two steps. It would be interesting to understand which are the relations
between the limit set (both in the classical and in the asynchronous cases) and the sensitivity to asynchronism.

Indeed, remark that in our study the sensitivity is defined on one step of the dynamics. It would be interesting
to compare how changes the sensitivity function of an ECA when the limit set is considered.
This idea has been investigated in works on {\em block-invariance} \cite{gmmmm18,gmmmrf15},
with the difference that it concentrates only on the set of
configurations in attractors, and discards the transitions within these sets.

Remark also that this study focus on block-sequential updating schemes. However, block-parallel updating
schedules are gaining growing interest~\cite{DemongeotS20}. It is a promising research direction to investigate
how the sensitivity functions change when block-parallel schedules are considered.

Another interesting research direction would consider the generalization of our study to arbitrary CA in order to verify
if a finer grained set of classes appear or not. Maybe, the set of possible functions is tightly related to the structure
of the neighborhood.

Finally, another possible generalization would consider infinite configurations in the spirit of \cite{rommp19}.
However, it seems much more difficult to come out with precise asymptotic results in this last case.

\section*{Acknowledgments}

The work of K\'evin Perrot was funded mainly by his salary as a French
State agent and therefore by French taxpayers' taxes,
affiliated to
Aix-Marseille Univ, Univ. de Toulon, CNRS, LIS, France,
and to
Univ. C\^{o}te d'Azur, CNRS, I3S, France,
and secondarily by
ANR-18-CE40-0002 FANs project,
ECOS-Sud C16E01 project,
STIC AmSud CoDANet 19-STIC-03 (Campus France 43478PD) project. 
P.P.B. thanks the Brazilian agencies CAPES and CNPq for the projects 
CAPES 88881.197456/2018-01, CAPES-Mackenzie PrInt project 88887.310281/2018-00 and CNPq-PQ 305199/2019-6.

%%%%%%%%%%%%%%%%%%%%%%%%%%%%%%%%
% biblio
%%%%%%%%%%%%%%%%%%%%%%%%%%%%%%%%
\bibliographystyle{plain}
\bibliography{biblio}

\end{document}